%% file: LCHS_arxiv_v1.tex
\documentclass[11pt]{article}

\usepackage{authblk}
\usepackage{amsmath,amssymb,amsthm}
\usepackage{mathtools}
\usepackage{geometry}
\usepackage{graphicx}
\usepackage[percent]{overpic}
\usepackage{float}
\usepackage{url}
\usepackage{color}
\usepackage{enumitem}
\usepackage[numbers,sort]{natbib}
\usepackage[colorlinks=true, linkcolor=blue, citecolor=blue, urlcolor=blue]{hyperref}
\usepackage[normalem]{ulem}
\usepackage{cancel}
\usepackage{booktabs,siunitx}
\title{Optimal kernel functions for linear combination of \\ Hamiltonian simulation}

\author{Harriet Apel\thanks{hapel@psiquantum.com} }
\author{Bjorn K. Berntson\thanks{bberntson@psiquantum.com}}
\affil{PsiQuantum, 700 Hansen Way, Palo Alto, CA 94304, USA}

\theoremstyle{plain}
\newtheorem{thm}{Theorem}
\newtheorem{lem}[thm]{Lemma}
\newtheorem{prop}[thm]{Proposition}
\newtheorem{cor}[thm]{Corollary}

\theoremstyle{definition}
\newtheorem{prob}[thm]{Problem}
\newtheorem{defn}[thm]{Definition}

\newcommand{\mc}[1]{\mathcal{#1}}
\newcommand{\m}[1]{\mathsf{#1}}

\newcommand{\ee}{\mathrm{e}}
\newcommand{\ii}{\mathrm{i}}

\newcommand{\C}{\mathbb{C}}
\newcommand{\N}{\mathbb{N}}
\renewcommand{\H}{\mathbb{H}}
\newcommand{\R}{\mathbb{R}}
\newcommand{\Z}{\mathbb{Z}}

\newcommand{\ccdot}{\,\cdot\,}
\newcommand{\abs}[1]{\lvert#1\rvert}

\DeclarePairedDelimiter{\norm}{\lVert}{\rVert}

\begin{document}

\maketitle

\begin{abstract}
Linear combination of Hamiltonian simulation (LCHS) provides an asymptotically optimal framework for the quantum simulation of linear non-unitary dynamics.
The LCHS approach approximately implements the Peano--Baker
propagator through an integral representation determined by a
\textit{kernel function}. We establish an LCHS approximation theorem
for a broadened class of admissible kernel functions, defined by natural complex-analytic
conditions.
In the standard time-independent block-encoding access model, we explicitly construct, for each target error $\varepsilon$, the unique kernel function in this class minimizing the query-cost functional and show that the minimum of this functional has the asymptotic expansion
$\tfrac{2\ee}{\pi}\big(\log\frac{1}{\varepsilon}
-\log\log\frac{1}{\varepsilon}+o(1)\big)$ as
${\varepsilon\downarrow0}$.
The provably optimal leading coefficient $2\ee/\pi$ improves by a factor of approximately 2.1 on the coefficient obtained empirically by Low and Somma through numerical minimization over their family of optimally scaling kernel functions.
We construct a spectrally convergent quadrature with explicit bounds on operator error and linear combination of unitaries (LCU) subnormalization, proving that this asymptotic improvement persists after discretization into a finite linear combination of unitaries suitable for implementation. Our results are based on Hardy space theory and convex analysis.
\end{abstract}

\newpage

\setcounter{tocdepth}{2} 
\tableofcontents

\section{Introduction}
Linear combination of Hamiltonian simulation (LCHS) is a quantum algorithm
for solving linear systems of ordinary differential equations (ODEs), extending quantum simulation to general linear dynamics, including non-unitary evolution \cite{an2023linear,An_2025}. The central idea is to represent the solution of
a linear ODE system as a continuous superposition of solutions to Hamiltonian
simulation problems, i.e., Schr\"{o}dinger equations generated by Hermitian
matrices.
The choice of kernel function affects the cost of the resulting quantum algorithm.
Our aim is to characterize the optimal choice within a natural complex-analytic class
and to construct a quadrature that preserves the resulting cost improvement.

\newpage

Fix $n\in\N$ and $\tau\in (0,\infty)$; the initial-value problem
treated by LCHS is as follows. 

\begin{prob}[Initial value problem for linear ODEs]
\label{prob:IVP}
Given $\m{A}:[0,\tau]\to \C^{2^n\times 2^n}$ and 
$\m{u}_0\in \C^{2^n}$, find
an absolutely continuous function $\m{u}:[0,\tau]\to \C^{2^n}$ such that $\m{u}(0)=\m{u}_0$ and
\begin{equation}\label{eq:ODE}
\frac{\mathrm{d}}{\mathrm{d}t}\m{u}(t)=-\m{A}(t)\m{u}(t),
\end{equation}
for almost every $t\in[0,\tau]$.
\end{prob}
Under mild regularity conditions on $\m{A}$, Problem~\ref{prob:IVP} has a unique solution represented by a time-ordered exponential \cite[Chaps.~2--3]{coddington1955}.

\begin{thm}[Peano--Baker]
\label{thm:PB}
Suppose that $\m{A}$ is measurable and satisfies
\begin{equation}
\int_{0}^{\tau} \lVert \m{A}(t)\rVert_2\,\mathrm{d}t < \infty ,	
\end{equation}
where $\lVert\ccdot\rVert_2$ is the spectral norm. Then, Problem~\ref{prob:IVP} has a unique solution
\begin{equation}
\label{eq:propagator}
\m{u}(t)=\m{U}(t)\m{u}_0 \qquad (t\in [0,\tau]),
\end{equation}
where
\begin{align}\label{eq:U}
\m{U}(t)\coloneqq&\;  \mc{T}\exp\Bigg(-\int_0^t \m{A}(s)\,\mathrm{d}s\Bigg) \nonumber \\
\coloneqq &\; \m{I}+\sum_{j=1}^{\infty} (-1)^j \int_{0\leq t_j\leq \cdots\leq t_1\leq t} \m{A}(t_1)\cdots \m{A}(t_j)\,\mathrm{d}t_j\cdots \mathrm{d}t_1. 
\end{align}
\end{thm}
LCHS approximates the propagator $\m{U}(t)$ by a weighted integral of unitary propagators. To identify these propagators, we write the Cartesian decomposition of $\m{A}(t)$,
\begin{equation}
\label{eq:Cartesian}
\m{A}(t)=\m{L}(t)+\ii\m{H}(t), \qquad \m{L}(t)\coloneqq \frac{\m{A}(t)+\m{A}(t)^{\dag}}{2}, \qquad \m{H}(t)\coloneqq \frac{\m{A}(t)-\m{A}(t)^{\dag}}{2\ii}.	
\end{equation}
Suppose that $\m{L}(t)\succeq 0$ for almost every $t\in[0,\tau]$.\footnote{This assumption
can always be achieved by a change of variables in \eqref{eq:ODE}. Let
$\m{v}(t)\coloneqq\ee^{-\int_0^t \lambda_-(s)\,\mathrm{d}s}\m{u}(t)$,
where $\lambda_-(t)\coloneqq \max \{0,-\lambda_{\min}(\m{L}(t))\}$. Then $\m{v}(t)$ solves \eqref{eq:ODE} with $\m{A}(t)$ replaced by $\m{A}(t)+\lambda_-(t)\m{I}$, whose Hermitian part is positive semidefinite.
The transformed matrix remains integrable since $\lambda_-(t)\leq \lVert\m{L}(t)\rVert_2\leq\lVert\m{A}(t)\rVert_2$ for every $t\in [0,\tau]$.
This rescaling preserves the normalized solution but can affect
the state-preparation cost; see~\cite[Sec.~1.1.4]{low2025}.
Recovering the original propagator multiplies absolute operator
errors by $\ee^{\int_0^t\lambda_-(s)\,\mathrm{d}s}$, which can
increase the cost of achieving a fixed target accuracy.}
For $z\in\C$, define
\begin{equation}
\label{eq:UE}
\mc{U}(z;t)
\coloneqq
\mc{T}\exp\Bigg(
-\ii\int_0^t
(z\m{L}(t')+\m{H}(t'))\,\mathrm{d}t'
\Bigg)
\qquad
(t\in[0,\tau]).
\end{equation}
At $z=-\ii$, we recover the target propagator, $
\mc{U}(-\ii;t)=\m{U}(t)$,
whereas for $z=x\in\R$, the matrix
$x\m{L}(t)+\m{H}(t)$ is Hermitian and
$\mc{U}(x;t)$ is correspondingly unitary for every $t\in [0,\tau]$. \\

We consider kernel functions $f$ defined on certain domains in the lower half-plane, having a unique simple pole at $z=-\ii$ with
prescribed residue
\begin{equation}
\label{eq:res_cond}
\operatorname*{Res}_{z=-\ii}f(z)
=
\frac{\ii}{\sqrt{2\pi}},
\end{equation}
and admitting an integrable real-axis boundary value $f_0$.\footnote{See \eqref{eqn:upper_boundary_value} for the precise
definition. Under the stronger boundary regularity assumptions imposed
in existing treatments of LCHS, $f_0$ is simply the pointwise
restriction of $f$ to $\R$. In the present work, the existence of
$f_0$ as an $L^1$ boundary value follows from the admissibility
conditions introduced below.}
For a truncation radius $R>0$, the corresponding LCHS approximation to the Peano--Baker propagator \eqref{eq:U} is
\begin{equation}
\label{eq:Of}
\m{O}_f(t;R)
\coloneqq
\frac{1}{\sqrt{2\pi}}
\int_{-R}^{R}
f_0(x)\mc{U}(x;t)\,\mathrm{d}x
\qquad
(t\in[0,\tau]).
\end{equation}
Theorem~\ref{thm:LCHS} gives an explicit bound on $\lVert\m{O}_f(t;R)-\m{U}(t)\rVert_2$
for the admissible kernel functions defined below.\\

Applying a quadrature rule to \eqref{eq:Of} gives a finite linear combination of the unitaries $\mc{U}(x;t)$, generated by the Hamiltonians $x\m{L}(t)+\m{H}(t)$. Given controlled implementations of these unitaries at the quadrature nodes, the linear combination of unitaries (LCU) construction block-encodes the resulting sum. For a nonzero initial datum $\m{u}_0$, and given a preparation procedure for its normalized state, this block encoding can then be used to prepare an approximation to $\m{u}(t)/\lVert\m{u}(t)\rVert_2$.\\

Earlier LCHS analyses \cite{An_2025,low2025,an2024laplace,huang2025,pocrnic2025} impose conditions on the kernel function that are
stronger than necessary for the underlying contour integration
argument.
These conditions restrict the class over which the kernel
function can be optimized.
We work with a natural Hardy space class\footnote{A Hardy space formulation for LCHS on the lower half-plane appears in \cite[Thm.~3]{an2024laplace}, with the additional assumption of continuity on the boundary.} on strips and the lower half-plane,
prove an LCHS approximation theorem under weaker
hypotheses, and solve the resulting infinite-dimensional constrained optimization
problem exactly.
For each target error $\varepsilon\in(0,1)$, the solution is a unique kernel function in closed form, up to solving a scalar transcendental equation.
The optimal domain is the lower half-plane and the minimum of our continuous query-cost functional has the asymptotic expansion
$\tfrac{2\ee}{\pi}\big(\log\frac{1}{\varepsilon}
-\log\log\frac{1}{\varepsilon}+o(1)\big)$ as
$\varepsilon\downarrow0$.
The optimal kernel function has branch points at the endpoints of the truncated integration interval, which are admitted by our weaker boundary assumptions. We construct a spectrally convergent quadrature adapted to this analytic structure, thereby connecting the continuous optimizer to a finite LCU implementation~\cite{berry2014}.

Concurrently and independently, Ni and Ying \cite{ni2026quantumeigenvaluetransformationlinear} proposed a numerically motivated kernel function ansatz of the same functional form as our optimizer \eqref{eq:f_optimal} and established the same leading asymptotic cost coefficient $2\ee/\pi$. Our approach derives the kernel function, without prescribing its functional form, by directly solving the underlying infinite-dimensional optimization problem. It establishes global optimality and uniqueness within the Hardy space admissible class, and yields the sharper asymptotic expansion \eqref{eq:optscaling}. We also construct a spectrally convergent quadrature that rigorously preserves this expansion after discretization. 
 \\

In the remainder of this section, we give the requisite technical
background, state our results with supporting numerics, and provide remarks on their interpretation and relation to prior work.
The proof of our main result, which constructs the extremal kernel function that minimizes the continuous LCHS cost, is given in Section~\ref{sect:proof_thmmain}.
The construction and analysis of the spectrally convergent quadrature for this optimal kernel function are given in Section~\ref{subsection:quadproof}.

\subsection{Admissible kernel functions and the LCHS approximation}

To motivate and contextualize our results, we state an LCHS theorem for a broadened class of kernel functions.
This theorem unifies \cite[Thm.~1]{low2025} and \cite[Thm.~6]{An_2025} while requiring weaker analytic hypotheses.
Our approach relies on tools from classical Hardy space theory, which we review in Appendix~\ref{appen:hardy}. \\

We begin by specifying the admissible class of kernel functions we study. Let $y_0\in (1,\infty]$. Our kernel functions are defined on the domain
\begin{equation}\label{eqn:lower_strip}
S_{y_0}\coloneqq \begin{cases}
    \{z\in \C: -y_0 <\operatorname{Im}z< 0\} & y_0\in (1,\infty) \\
    \H_- & y_0=\infty,
\end{cases}
\end{equation}
where $\H_-$ is the lower half-plane,
\begin{equation}
\H_{-}\coloneqq \{z\in \C:  \operatorname{Im}z< 0\}.
\end{equation}
Let $\mc{M}(\Omega)$ and $\mc{O}(\Omega)$ denote the classes of meromorphic and holomorphic functions, respectively, on a set $\Omega\subseteq\C$.\footnote{If $\Omega$ is not open, a function is said to be meromorphic or holomorphic on $\Omega$ if it admits, respectively, a meromorphic or holomorphic extension to an open neighborhood of $\Omega$.}
We factor out the prescribed simple pole using the reciprocal of the \textit{Blaschke factor} $B(\ccdot;y_0)\in\mc{O}(S_{y_0})$,
\begin{equation}
\label{eq:Blaschke}
B(z;y_0)\coloneqq     \begin{dcases}
\frac{\sinh\big(\frac{\pi}{2y_0}(z+\ii)\big)}{\sinh\big(\frac{\pi}{2y_0}(z-\ii)\big)} & y_0\in (1,\infty) \\
\frac{z+\ii}{z-\ii} & y_0=\infty,
\end{dcases}
\end{equation}
which has a unique simple zero in $S_{y_0}$ at $z=-\ii$.
The relevant Hardy space is defined by
\begin{equation}
\label{eqn:def_H1}
H^1(S_{y_0})
\coloneqq
\Bigg\{
h\in\mc{O}(S_{y_0}):
\sup_{y\in(0,y_0)}
\int_{\R}\lvert h(x-\ii y)\rvert\,\mathrm{d}x<\infty
\Bigg\}.
\end{equation}
The remaining freedom in the kernel function is parameterized by a function $h\in H^1(S_{y_0})\subset \mc{O}(S_{y_0})$ subject to a normalization that enforces the prescribed residue at $z=-\ii$ in \eqref{eq:res_cond}.

\begin{defn}[Admissible kernel function]
\label{def:admissible}
A function $f\in \mc{M}(S_{y_0})$ is called admissible, and we write $f\in\mc{A}(S_{y_0})$, if it admits a representation
\begin{equation}\label{eqn:kernel_f_h}
f(z)=\frac{\ii}{\sqrt{2\pi}}\frac{h(z)}{B(z;y_0)},
\end{equation}
where $h\in H^1(S_{y_0})$ satisfies 
  \begin{equation}\label{eq:norm_h}
    h(-\ii)=\lim_{z\to -\ii}\frac{B(z;y_0)}{z+\ii}
=
\begin{cases}
\dfrac{\ii\pi}{2y_0\sin\big(\tfrac{\pi}{y_0}\big)} & y_0\in(1,\infty)\\
\dfrac{\ii}{2} & y_0=\infty.
\end{cases}
\end{equation}
\end{defn}

The class $\mc{A}(S_{y_0})$ generalizes the kernel function classes considered in \cite{low2025,an2023linear,An_2025}.
Early constructions of LCHS kernel functions
\cite{an2023linear,An_2025} required $f$ to decay uniformly throughout
$\H_-$.
Together with holomorphicity in $\H_-\setminus \{-\ii\}$ and continuity on $\overline{\H}_-\setminus \{-\ii\}$,
this hypothesis precludes exponential decay along the real axis
\cite[Prop.~8]{An_2025}.
Recently, Low and Somma \cite{low2025} circumvented this restriction by requiring decay only as
$\lvert\operatorname{Re}z\rvert\to\infty$, uniformly across a finite strip.\\

Our approach further weakens the hypotheses of \cite{low2025} by replacing the boundary regularity and uniform pointwise decay assumptions with the Hardy space condition \eqref{eqn:def_H1}, which applies
both to finite strips and to the lower half-plane.
This condition uniformly
controls the $L^1$ norms of $h$ along horizontal lines, which is precisely
what is needed to bound the approximation error in the LCHS analysis
below. The comparison of admissible classes is made precise in Proposition~\ref{prop:assumptions} in Appendix~\ref{subsec:assumptionsproof}.
There, we show that if $f$ satisfies the hypotheses of \cite[Thm.~1]{low2025} with a finite error bound,
then $f\in\mc{A}(S_{y_0})$, whereas the converse fails. \\

The half-plane class $\mc{A}(\H_-)$ also has an intrinsic characterization: its boundary
values are precisely the functions $\varphi\in L^1(\R)$ satisfying the
exact scalar LCHS identity
\begin{equation}
\label{eq:exact_scalar_LCHS}
\frac{1}{\sqrt{2\pi}}\int_{\R}\varphi(x)\ee^{-\ii kx}\,\mathrm{d}x
=\ee^{-k}
\qquad (k\geq0).
\end{equation}
This connects our Hardy space formulation with the Fourier
criterion of \cite[Thm.~9]{huang2025},
without the additional regularity assumptions imposed there.
This characterization is proved in Lemma~\ref{lem:exact_LCHS_characterization}
in Appendix~\ref{subsec:exact_LCHS_characterization}.\\

We emphasize that Definition~\ref{def:admissible} does not require
$f$ to be defined on $\partial S_{y_0}$. Instead, the representation
\eqref{eqn:kernel_f_h} and Hardy space theory guarantee that
the boundary value
\begin{equation}
\label{eqn:upper_boundary_value}
f_0(x)
\coloneqq
\lim_{\delta\downarrow0}f(x-\ii\delta)
\end{equation}
exists for almost every $x\in\R$ and defines a function
$f_0\in L^1(\R)$. When $y_0\in(1,\infty)$, the boundary value
\begin{equation}
\label{eqn:lower_boundary_value}
f_{y_0}(x)
\coloneqq
\lim_{\delta\downarrow0}
f(x-\ii y_0+\ii\delta)
\end{equation}
likewise exists for almost every $x\in\R$ and defines a function
$f_{y_0}\in L^1(\R)$. \\

The following theorem provides error guarantees for the LCHS approximation, in the spectral norm.

\begin{thm}[Approximate LCHS under Hardy space hypotheses]
\label{thm:LCHS}
Let $R>0$ and $y_0\in (1,\infty]$. Suppose that the measurable functions $\m{L},\m{H}:[0,\tau]\to \C^{2^n\times 2^n}$ satisfy $\m{L}(t)^{\dag}=\m{L}(t)$, $\m{H}(t)^{\dag}=\m{H}(t)$, $\m{L}(t)\succeq 0$ for almost every $t\in [0,\tau]$ and
\begin{equation}
\int_{0}^{\tau} \lVert \m{L}(t)\rVert_2\,\mathrm{d}t<\infty, \qquad  \int_{0}^{\tau} \lVert \m{H}(t)\rVert_2\,\mathrm{d}t < \infty. 
\end{equation}
Let $f\in \mc{A}(S_{y_0})$. Then $\m{O}_f(\ccdot;R)$ and $\m{U}(\ccdot)$, defined by \eqref{eq:Of} and \eqref{eq:U}, respectively, satisfy
\begin{equation}
\lVert \m{O}_f(t;R)-\m{U}(t)\rVert_2\leq \mc{E}_{R,y_0}[f] \qquad (t\in [0,\tau]),
\end{equation}
where
\begin{equation}
\label{eq:E}
\mc{E}_{R,y_0}[f]	\coloneqq  \begin{dcases} \frac{1}{\sqrt{2\pi}}\int_{\R\setminus[-R,R]} \lvert {f}_0(x)\rvert\,\mathrm{d}x+\frac{1}{\sqrt{2\pi}}\int_{\R} \lvert f_{y_0}(x)\rvert\,\mathrm{d}x & y_0\in (1,\infty)\\
\frac{1}{\sqrt{2\pi}}\int_{\R\setminus[-R,R]} \lvert {f}_0(x)\rvert\,\mathrm{d}x	& y_0=\infty.
  \end{dcases}
\end{equation}
\end{thm}

The proof of Theorem~\ref{thm:LCHS} uses the prescribed simple pole of $f$ together
with the fact that $\mc{U}(\ccdot;t)$ is entire and contractive on $\H_-$. If
$\Gamma\subset S_{y_0}$ is a negatively oriented rectangular contour
enclosing $z=-\ii$, Cauchy's residue theorem gives
\begin{equation}
    \frac{1}{\sqrt{2\pi}}
    \oint_{\Gamma} f(z)\mc{U}(z;t)\,\mathrm{d}z
    = \m{U}(t).
\end{equation}
The contributions from the vertical sides vanish as their real parts
tend to $\pm\infty$, leaving only the two horizontal contributions.
Under the analyticity assumptions of \cite[Thm.~1]{low2025}, the
horizontal sides of the contour may be taken directly on the strip boundary.
Here, the Hardy space hypothesis guarantees that the boundary values $f_0$ and, when $y_0<\infty$, $f_{y_0}$, belong to $L^1$
but does not guarantee
continuity up to the boundary. We therefore first work with contours
whose horizontal edges lie strictly inside $S_{y_0}$ and then use
Hardy space boundary convergence in $L^1$ as the upper edge approaches $\R$
and, when $y_0<\infty$, the lower edge approaches
$\R-\ii y_0$. Bounding the error caused by truncating the
upper-boundary integral to $[-R,R]$ gives the first term in
\eqref{eq:E}, while, for finite $y_0$, bounding the lower-boundary
integral gives the second. In the half-plane case $y_0=\infty$, a
separate regularization argument shows that the contribution from the
lower horizontal edge $\R-\ii Y$ vanishes as $Y\to\infty$, leaving only
the real-axis truncation error. \\

A detailed proof of Theorem~\ref{thm:LCHS} is given in Appendix~\ref{subsec:LCHSproof}.

\subsection{Statement of results}

While Theorem~\ref{thm:LCHS} establishes a continuous approximation guarantee, implementing the resulting operator $\m{O}_f(t;R)$ as a quantum algorithm requires discretizing the integral into a finite linear combination of the unitaries $\mc{U}(x;t)$. The two leading cost contributions retained in our model are the LCU subnormalization\footnote{Preparing a normalized output state may introduce an additional success-probability factor depending on the target and input state, see~\cite[Eqs.~(14)--(15)]{low2025}. We omit this factor because it is outside the cost model for kernel functions considered here.} and the cost of simulating the constituent unitaries \cite{an2023linear,gilyen2019}.
By \eqref{eq:Of} and the unitarity of $\mc{U}(x;t)$ for $x\in\R$, the truncated operator $\m{O}_f(t;R)$ has the natural continuous LCU subnormalization
\begin{equation}
\label{eq:alpha}
\alpha_R[f]
\coloneqq
\frac{1}{\sqrt{2\pi}}
\int_{-R}^{R}\lvert f_0(x)\rvert\,\mathrm{d}x.
\end{equation}
After quadrature, the block-encoding subnormalization is given by the $\ell^1$-norm of its coefficients via the standard LCU construction \cite{berry2014,gilyen2019}. \\

Although Theorem~\ref{thm:LCHS} applies to general time-dependent dynamics, the optimality results below concern the following time-independent cost model.
The Cartesian decomposition \eqref{eq:Cartesian} gives the identity
\begin{equation}\label{eq:hamiltonian_from_A}
x\m{L}+\m{H}=\frac{(x-\ii)\m{A}+(x+\ii)\m{A}^{\dag}}{2}.
\end{equation}
In the standard block-encoding access model for $\m{A}$, this yields a block encoding of $x\m{L}+\m{H}$ with subnormalization proportional to $\sqrt{1+x^2}$. Consequently, the largest block-encoding subnormalization for the
Hamiltonians $x\m{L}+\m{H}$ with $\lvert x\rvert\leq R$ grows as
$\sqrt{1+R^2}$.
Up to fixed access-model parameters and the usual logarithmic dependence on the simulation accuracy, the corresponding Hamiltonian-simulation cost therefore scales linearly with $R$ as $R\to\infty$ \cite{low2019,low2025}. For state preparation, amplitude amplification contributes a multiplicative
factor proportional to $\alpha_R[f]$.
Therefore, we model the leading dependence of the implementation cost on the kernel function by the functional 
\begin{equation}\label{eqn:query}
Q_R[f]\coloneqq R\alpha_R[f].
\end{equation}

This leads to the following constrained optimization problem for the optimal kernel function.\footnote{The constraint $\mc{E}_{R,y_0}[f]\leq \varepsilon$ is the sufficient error bound \eqref{eq:E} rather than the operator error; optimality in Problem~\ref{prob:main} is understood relative to this bound.}

\begin{prob}[Optimizing {$Q_R[f]$}]
\label{prob:main}
Given $\varepsilon\in (0,1)$, minimize $Q_R[f]$ over $R\in (0,\infty)$, $y_0\in (1,\infty]$, and $f\in\mc{A}(S_{y_0})$, subject to the constraint $\mc{E}_{R,y_0}[f]\leq \varepsilon$.
\end{prob}

Our main result gives the unique solution of Problem~\ref{prob:main}: the optimal choice of domain is the lower half-plane, and the optimal truncation radius $R_\star(\varepsilon)$ is determined by the unique root $m_\star(\varepsilon)\in(0,1)$ of a transcendental equation.
The corresponding optimal kernel function is explicit.

\begin{thm}[Optimal kernel functions for LCHS]
\label{thm:optimal-kernel}
The unique solution of Problem~\ref{prob:main} is given by
\begin{equation}
\label{eq:f_optimal}
f_\star(z)=\sqrt{\frac{2}{\pi}}\frac{\varepsilon}{1-m_\star(\varepsilon)}
\frac{\exp\bigg( \dfrac{ m_\star(\varepsilon)}{\ii\sin(\pi m_\star(\varepsilon))}
\log\dfrac{z+R_\star(\varepsilon)}{z-R_\star(\varepsilon)}\bigg)}
{z^2+1},
\end{equation}
where the logarithm is chosen so that the composition is holomorphic on 
\begin{equation}
\label{eq:Omega_R}
\Omega_{{R}_{\star}(\varepsilon)}\coloneqq \C\setminus\big((-\infty,-R_\star(\varepsilon)]\cup [R_\star(\varepsilon),\infty)\big)
\end{equation}
and agrees with the principal branch of the logarithm on $\H_-$. Moreover, $y_0=\infty$ and
\begin{equation}\label{eq:Rstar_from_mstar}
R_\star(\varepsilon)\coloneqq\tan\bigg(\frac{\pi}{2}m_\star(\varepsilon)\bigg),
\end{equation}
where $m_\star(\varepsilon)$ is the unique root in $(0,1)$ of
\begin{equation}
\label{eq:transcendental}
\log\bigg(\frac{1-m}{\varepsilon}\bigg)=\frac{\pi m^2}{\sin(\pi m)}. 
\end{equation}
Let $Q_\star(\varepsilon)$ denote the minimum value of the continuous LCHS cost functional in \eqref{eqn:query}. Then
\begin{align}\label{eq:optscaling}
Q_\star(\varepsilon)=&\; R_\star(\varepsilon) m_\star(\varepsilon)\exp\bigg( \frac{ \pi m_\star(\varepsilon)(1-m_\star(\varepsilon))}{\sin(\pi m_\star(\varepsilon))}\bigg) \nonumber \\
=&\;  \frac{2\ee}{\pi}\bigg(\log\frac{1}{\varepsilon}-\log\log\frac{1}{\varepsilon}+o(1)\bigg) \quad (\varepsilon\downarrow 0).
\end{align}
\end{thm}

The proof of Theorem~\ref{thm:optimal-kernel} proceeds in two stages.
We first solve the constrained subproblem for fixed truncation radius $R$ and strip width
$y_0$, yielding the explicit subnormalization-optimal kernel function of Proposition~\ref{prop:subnormalization-optimal-kernel}.
This is in contrast to \cite{low2025}, where optimization is performed over an explicit finite-parameter family of admissible kernel functions.
We then show that the half-plane setting $y_0 = \infty$, where $S_\infty=\H_-$, is optimal, and determine the remaining parameter $R$ by balancing the truncation radius against the optimal LCU subnormalization.
This yields the asymptotic scaling in \eqref{eq:optscaling}.
The proofs of Proposition~\ref{prop:subnormalization-optimal-kernel} and Theorem~\ref{thm:optimal-kernel} are given in Section~\ref{sect:proof_thmmain}. \\

We can also retain the full Hamiltonian subnormalization factor $\sqrt{1+R^2}$ in the cost functional rather than using the large $R$ approximation.
The following corollary shows that this modifies the optimal kernel function \eqref{eq:f_optimal} without changing the asymptotic expansion of the minimum cost functional.
For low accuracy $\varepsilon\sim10^{-3}$, the optimal $R$ we find is modest, which motivates this inclusion.

\begin{cor}[Low-accuracy optimal kernel functions]
\label{cor:full_normalization_cost}
For $\varepsilon\in(0,1)$, replace $Q_R[f]$ in Problem~\ref{prob:main} by
\begin{equation}\label{eq:full_normalization_cost}
\widetilde Q_R[f]\coloneqq\sqrt{1+R^2}\,\alpha_R[f].
\end{equation}
The modified problem has a unique solution, with $y_0=\infty$ and
$R=\widetilde R_\star(\varepsilon)$, where $\widetilde R_\star$ is obtained
from \eqref{eq:Rstar_from_mstar} by replacing $m_\star$ with
$\widetilde m_\star$, and $\widetilde m_\star(\varepsilon)$ is the unique
root in $(0,1)$ of
\begin{equation}\label{eq:full_normalization_stationary}
\log\frac{1-m}{\varepsilon}
=\frac{\pi m^2}{2}\tan\frac{\pi m}{2}.
\end{equation}
The optimal kernel function is
\begin{equation}\label{eq:full_normalization_kernel}
\widetilde f_\star(z)=\sqrt{\frac{2}{\pi}}\frac{\varepsilon}{1-\widetilde m_\star(\varepsilon)}
\frac{\exp\bigg( \dfrac{ \widetilde m_\star(\varepsilon)\widetilde R_\star(\varepsilon)}{2\ii}
\log\dfrac{z+\widetilde R_\star(\varepsilon)}{z-\widetilde R_\star(\varepsilon)}\bigg)}
{z^2+1},
\end{equation}
where the logarithm is chosen so that the composition is holomorphic on $\Omega_{\widetilde R_{\star}(\varepsilon)}$ \eqref{eq:Omega_R} and agrees with the principal branch of the logarithm on $\H_-$.
Moreover, $\widetilde R_\star(\varepsilon)>R_\star(\varepsilon)$, and the
minimum value $\widetilde Q_\star(\varepsilon)$ satisfies
\begin{equation}\label{eq:full_normalization_asymptotics}
\widetilde Q_\star(\varepsilon)
=Q_\star(\varepsilon)
+O\!\left(\frac{1}{\log(1/\varepsilon)}\right)
\qquad (\varepsilon\downarrow0).
\end{equation}
In particular, the asymptotic expansion \eqref{eq:optscaling} also holds
for $\widetilde Q_\star$.
\end{cor}
The proof is given at the end of Section~\ref{sect:proof_col}. \\

An LCU implementation using the optimal kernel functions of Theorem~\ref{thm:optimal-kernel} requires approximating the associated integral accurately by a finite quadrature rule on the truncated interval $[-R_\star(\varepsilon),R_\star(\varepsilon)]$.
For the quadrature analysis, fix $\varepsilon\in(0,1)$ and set
\begin{equation}\label{eq:quadrature_parameters}
\begin{split}
m\coloneqq &\; m_\star(\varepsilon),\qquad
R\coloneqq R_\star(\varepsilon), \\
\Delta\coloneqq &\;  \frac{\pi m}{\sin(\pi m)},\qquad
\eta_R\coloneqq\arcsin(1-m).
\end{split}
\end{equation}
We write $f_{\star,0}$ for the boundary value of the optimal kernel function $f_{\star}$ \eqref{eq:f_optimal} on $\R$, as defined in \eqref{eqn:upper_boundary_value}. This limit exists at every $x\in \R\setminus \{-R,R\}$; it agrees with the analytic continuation of \eqref{eq:f_optimal} on $(-R,R)$. \\

We develop a spectrally convergent quadrature for the integral expression in \eqref{eq:Of}, based on the double-exponential change of variables
\begin{equation}\label{eq:phiR}
\phi_R(\zeta)\coloneqq R\tanh\bigg(\frac{\pi}{2}\sinh \zeta \bigg).
\end{equation}
Here $\zeta$ denotes the complex transformed variable, and
we write $\xi=\operatorname{Re}\zeta$.
On the real axis, the change of variables is $x=\phi_R(\xi)$.
We denote the uniform quadrature step in $\xi$ by $\delta_\xi$,
so that the corresponding nodes in the original variable
are $\phi_R(n\delta_\xi)$. \\

For each fixed $\varepsilon\in(0,1)$, define the upper bound
$\eta_0(\varepsilon)$ for the quadrature strip half-width by
\begin{equation}\label{eq:quadrature_strip_width}
\eta_{0}(\varepsilon)\coloneqq
\min\!\bigg\{
\frac{\pi}{6},\,
\arcsin\!\bigg(\frac{2}{\pi\sqrt{1+R^2}}\bigg)
\bigg\},
\end{equation}
with $m$, $R$, $\Delta$, and $\eta_R$ as in
\eqref{eq:quadrature_parameters}.
Choosing $0<\eta<\eta_0(\varepsilon)$ ensures that\footnote{Since
$R=\tan(\pi m/2)$ and $0<m<1$, we have
$\sin\eta_0\leq2/(\pi\sqrt{1+R^2})
=(2/\pi)\sin(\pi(1-m)/2)<1-m=\sin\eta_R$.
Moreover, $\Delta\sin\eta_0\leq m/\sin(\pi m/2)\leq1$, whereas
$\pi\cos\eta_0\geq\pi\sqrt{3}/2>1$ since $\eta_0\leq\pi/6$.
Thus $\eta_0<\eta_R$ and $\eta_0<\arctan(\pi/\Delta)$.}
$\eta<\eta_R$, $\eta<\pi/6$, and
\begin{equation}\label{eq:strip_bound_conditions}
\frac{\pi^2}{4}(1+R^2)\sin^2\eta<1,
\qquad
\Delta\sin\eta<\pi\cos\eta.
\end{equation}
Moreover, for $\varepsilon\leq 10^{-1}$, the second element of the set in \eqref{eq:quadrature_strip_width} is smaller than the first, so that
\begin{equation}
\eta_0(\varepsilon)=\arcsin\!\bigg(\frac{2}{\pi\sqrt{1+R^2}}\bigg).
\end{equation}

\begin{thm}[Spectrally convergent quadrature for the optimal kernel function]
\label{thm:mainquadrature}

Suppose that $\m{L}$ and $\m{H}$ satisfy the matrix hypotheses of Theorem~\ref{thm:LCHS}.
Fix $\eta\in(0,\eta_0(\varepsilon))$; such an $\eta$ exists for every $\varepsilon\in (0,1)$.
Define
\begin{equation}\label{eq:discretization_strip_constant}
\begin{aligned}
M_\eta\coloneqq &\;
\frac{4R\ee^{(1-m)\Delta+\Delta\sin\eta}}
{\big(1-\frac{\pi^2}{4}(1+R^2)\sin^2\eta\big)
(\pi\cos\eta-\Delta\sin\eta)}
\exp\!\bigg(
\frac{R\pi\sin\eta}{2-\frac{\pi^2}{2}\sin^2\eta}
\int_0^\tau\norm{\m{L}(s)}_2\,\mathrm{d}s
\bigg).
\end{aligned}
\end{equation}

Let $\delta_\xi>0$ and $N\in\N$.  Then the operator $\m{O}_{f_\star}(\tau;R)$ \eqref{eq:Of} 
is approximated by the finite quadrature
\begin{equation}
\mathcal I_N(\tau;R,\delta_\xi)
\coloneqq
\frac{\delta_\xi}{\sqrt{2\pi}}
\sum_{\substack{n\in \Z\\ |n|\le N}}
f_{\star,0}\big(\phi_R(n\delta_\xi)\big)
\mc{U}\big(\phi_R(n\delta_\xi);\tau\big)
\phi_R'(n\delta_\xi),
\end{equation}
with error
\begin{equation}
\big\lVert
\m{O}_{f_\star}(\tau;R)-\mathcal I_N(\tau;R,\delta_\xi)
\big\rVert_2
\leq
\frac{2M_\eta}{\exp(2\pi\eta/\delta_\xi)-1}
+
\frac{4R\ee^{(1-m)\Delta}}{\pi}\exp\big(-\pi\sinh(N\delta_\xi)\big).
\end{equation}
\end{thm}
In particular, for prescribed $\varepsilon_{\rm disc},\varepsilon_{\rm trunc}\in(0,1)$, choosing $\delta_\xi$ and $N$ under the hypotheses of the theorem so that
\begin{equation}
\label{eq:quadrature_parameters_choice}
\delta_\xi\leq
\frac{2\pi\eta}{\log(1+2M_\eta/\varepsilon_{\rm disc})},\qquad 
N\delta_\xi\geq
\operatorname{arcsinh}\!\bigg(
\frac1\pi\log\frac{4R\ee^{(1-m)\Delta}}{\pi\varepsilon_{\rm trunc}}
\bigg)
\end{equation}
gives an error of at most $\varepsilon_{\rm disc}+\varepsilon_{\rm trunc}$.
The proof is given in Section~\ref{subsection:quadproof}.
Together with Theorem~\ref{thm:LCHS}, this gives an approximation to $\m{U}(\tau)$ \eqref{eq:U} with error at most $\mc{E}_{R,\infty}[f_\star]+\varepsilon_{\rm disc}+\varepsilon_{\rm trunc}$. \\

The quadrature constructed in Theorem~\ref{thm:mainquadrature} preserves the continuous subnormalization in the following sense. For the optimal kernel function~\eqref{eq:f_optimal}, the discrete subnormalization is
\begin{equation}\label{eq:discrete_alpha_numerics}
\alpha_{N,\delta_\xi}
\coloneqq\frac{\delta_\xi}{\sqrt{2\pi}}
\sum_{\substack{n\in\Z \\|n|\leq N}}\big\lvert f_{\star,0}\big(\phi_R(n\delta_\xi)\big)\big\rvert\phi_R'(n\delta_\xi)
=\frac{\delta_\xi\ee^{(1-m)\Delta}}{\pi}
\sum_{\substack{n\in \Z \\|n|\leq N}}\frac{\phi_R'(n\delta_\xi)}{1+\phi_R(n\delta_\xi)^2}.
\end{equation}
For each summand, $x=\phi_R(n\delta_\xi)\in(-R,R)$.
By the branch choice in Theorem~\ref{thm:optimal-kernel},
the logarithm in \eqref{eq:f_optimal} has imaginary part $\pi$
on this interval, so the exponential factor has modulus
$\ee^\Delta$.
Together with $\varepsilon/(1-m)=\ee^{-m\Delta}$ from \eqref{eq:transcendental} and \eqref{eq:quadrature_parameters}, this yields the second expression.
This finite quadrature preserves the continuous LCU cost in the limit, with an explicit error bound given in the following result.

\begin{cor}[LCU subnormalization convergence]
\label{cor:quadrature_subnormalization}
Let $R=R_\star(\varepsilon)$ and let $f_\star$ be the optimal kernel function from
Theorem~\ref{thm:optimal-kernel}. With the notation and quadrature conditions of
Theorem~\ref{thm:mainquadrature}, the discrete subnormalization
$\alpha_{N,\delta_\xi}$ \eqref{eq:discrete_alpha_numerics} satisfies
\begin{equation}\label{eq:quantitative_subnormalization}
\big\lvert\alpha_{N,\delta_\xi}-\alpha_R[f_\star]\big\rvert
\leq\frac{2C_\alpha}{\exp(2\pi\eta/\delta_\xi)-1}
+\frac{4R\ee^{(1-m)\Delta}}{\pi}\exp\big(-\pi\sinh(N\delta_\xi)\big),
\end{equation}
where
\begin{equation}\label{eq:normalization_constant}
C_\alpha\coloneqq
\frac{4R\ee^{(1-m)\Delta}}
{\pi\cos\eta\,\big(1-\frac{\pi^2}{4}(1+R^2)\sin^2\eta\big)}.
\end{equation}
In particular, the LCU subnormalization of
$\mathcal I_N(\tau;R,\delta_\xi)$ converges to $\alpha_R[f_\star]$:
\begin{equation}
\label{eq:quadrature_subnormalization}
\lim_{\substack{\delta_\xi\downarrow0\\N\delta_\xi\to\infty}}
\frac{\delta_\xi}{\sqrt{2\pi}}
\sum_{\substack{n\in \Z \\|n|\leq N}}
\big\lvert f_{\star,0}\!\left(\phi_R(n\delta_\xi)\right)\big\rvert
\phi_R'(n\delta_\xi)
=
\alpha_R[f_\star].
\end{equation}
Consequently, multiplying the discrete subnormalizations \eqref{eq:discrete_alpha_numerics} by $R$ gives costs
converging to $Q_R[f_\star]$.
\end{cor}

The proof is given at the end of Section~\ref{subsection:quadproof}.
Corollary~\ref{cor:quadrature_subnormalization} gives the quantitative estimate
\begin{equation}\label{eq:discrete_cost_numerics}
\big\lvert R\alpha_{N,\delta_\xi}-Q_\star(\varepsilon)\big\rvert
\leq R\bigg(
\frac{2C_\alpha}{\exp(2\pi\eta/\delta_\xi)-1}
+\frac{4R\ee^{(1-m)\Delta}}{\pi}\exp\big(-\pi\sinh(N\delta_\xi)\big)
\bigg),
\end{equation}
where $C_\alpha$ is independent of $\m L$ and $\m H$ and satisfies $C_\alpha\leq M_\eta$.
Thus the quadrature parameters in Theorem~\ref{thm:mainquadrature} also control the change in the cost.
In particular, choosing $\varepsilon_{\rm disc}+\varepsilon_{\rm trunc}\leq\varepsilon$ gives
$R\alpha_{N,\delta_\xi}=Q_\star(\varepsilon)+O(\varepsilon\log(1/\varepsilon))$ as $\varepsilon\downarrow0$.
The finite quadrature therefore preserves the asymptotic expansion \eqref{eq:optscaling}, with $\varepsilon$ denoting the continuous error budget and total operator error at most $2\varepsilon$. \\

With quadrature parameters satisfying 
\eqref{eq:quadrature_parameters_choice} and
$\varepsilon_{\rm disc}+\varepsilon_{\rm trunc}\leq\varepsilon$,
the quadrature in Theorem~\ref{thm:mainquadrature} also attains the modified continuous
minimum in Corollary~\ref{cor:full_normalization_cost}
up to $o(1)$: \\
\begin{equation}
\sqrt{1+R_\star(\varepsilon)^2}\,\alpha_{N,\delta_\xi}
=\widetilde Q_\star(\varepsilon)+o(1)
\qquad (\varepsilon\downarrow0).
\end{equation}
Indeed, \eqref{eq:Rstar_from_mstar}--\eqref{eq:optscaling}
give $R_\star(\varepsilon)\sim(2/\pi)\log(1/\varepsilon)$
and $\alpha_{R_\star}[f_\star]\to\ee$, while
\eqref{eq:quantitative_subnormalization} and
$C_\alpha\leq M_\eta$ give
$\alpha_{N,\delta_\xi}
=\alpha_{R_\star(\varepsilon)}[f_\star]+O(\varepsilon)=O(1)$.
Since
$\sqrt{1+R_\star^2}-R_\star
=1/(\sqrt{1+R_\star(\varepsilon)^2}+R_\star(\varepsilon))=o(1)$,
the claim follows from \eqref{eq:discrete_cost_numerics}
and \eqref{eq:full_normalization_asymptotics}. \\

Previous works were able to employ uniform quadrature \cite[Thm.~3]{low2025}, relying on the standard result that the trapezoidal rule converges exponentially for functions that are holomorphic in a strip about the real axis with suitable decay and integrability properties \cite{trapezoidal_rule}. 
Nonuniform quadrature rules and their implementation costs have also been studied in \cite{novikau2026,aftab2026}, including sinh--sinh quadrature for Gaussian-damped kernel functions in the latter.
By expanding the admissible class of functions, we admit a wider range of analytic behaviour; in particular notice that the optimal kernel function of Theorem~\ref{thm:optimal-kernel} contains a logarithm that introduces branch points at $z=\pm R_\star(\varepsilon)$.
This necessitates a refined quadrature analysis.
Figure~\ref{fig:kernel} illustrates the modulus of the optimal kernel function for two error tolerances, with close-ups of its jumps. Continuous costs and sufficient quadrature node counts are compared in Section~\ref{subsec:numerics}.
The additional analytic work provides a rigorous certification of the exponential discretization error and double-exponential truncation error in Theorem~\ref{thm:mainquadrature}.

\begin{figure}[H]
  \centering
  \includegraphics[width=\textwidth]{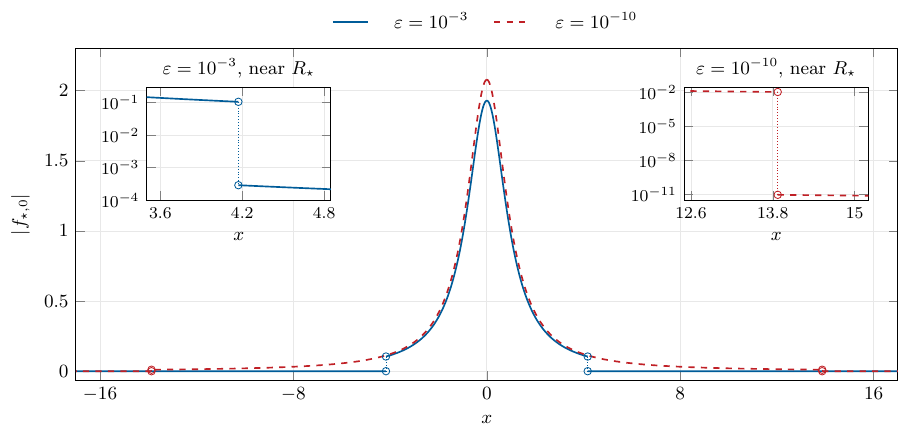}
  \caption{Moduli of the optimal kernel function boundary values \eqref{eq:f_optimal} with error parameters $\varepsilon=10^{-3}$ and $10^{-10}$. The main plot has a linear vertical scale; the insets use logarithmic vertical scales to resolve the jumps at $R_\star(\varepsilon)\approx4.172$ and $13.872$, respectively. The negative endpoints follow from the fact that $\abs{f_{\star,0}(x)}$ is even. Open circles denote one-sided modulus limits, and dotted segments are guides to the eye. The insets display the nonzero exterior tails.}\label{fig:kernel}
\end{figure}

\subsection{Continuous costs and quadrature node counts}\label{subsec:numerics}

We compare our optimal kernel functions \eqref{eq:f_optimal} with the numerically optimized kernel functions of Low and Somma \cite{low2025}.
First, we compare the continuous costs $Q_R[f]=R\alpha_R[f]$ for the two classes of kernel functions. Second, we compare sufficient quadrature node counts under the error bounds used in our analysis. In the tabulated cases, our counts are at most a factor of $2.6$ larger than those for the Low--Somma kernel functions with uniform quadrature~\cite{low2025}. \\

\paragraph{Cost comparison.}
In our notation, the four-parameter family of kernel functions considered in \cite{low2025} is
\begin{equation}\label{eq:generalized_kernel_intro}
f_{j,y}(z;\gamma,c)
=\frac{(y+1)^{j-1}}{\sqrt{2\pi}}
\frac{\ee^{c(1-\ii z)}\ee^{-(z^2+1)/(4\gamma^2)}}
{(1-\ii z)(y+\ii z)^{j-1}}.
\end{equation}
We use the two sets of numerically optimized parameters in \cite[Tab.~3]{low2025}: the choices with $j=2$, $y=1$, and the choices allowing $j$ and $y$ to vary.
We evaluate the continuous costs for $\varepsilon=10^{-1},\ldots,10^{-10}$. \\

\begin{table}[H]
\centering
\small
\setlength{\tabcolsep}{5pt}
\renewcommand{\arraystretch}{1.12}
\input{cost-table}
\caption{Comparison of numerical continuous query-cost $Q_R[f]=R\alpha_R[f]$ values.
The Low--Somma parameters are taken unchanged from
\cite[Tab.~3]{low2025}.
For these columns, $\varepsilon$ denotes the published target tolerance, the continuous error bound \eqref{eq:E} evaluated with the rounded parameters can be slightly larger than $\varepsilon$.
The corresponding discrete costs can be made arbitrarily close to these values
by refining the quadrature; see the node-count discussion accompanying
Table~\ref{tab:quadrature_nodes}. The values for $f_{\star}$ agree, to the two decimal places reported, with the numerically optimized values of Ni and Ying \cite[Tab.~1]{ni2026quantumeigenvaluetransformationlinear}.}
\label{tab:main}
\end{table}

Table~\ref{tab:main} illustrates the reduction in continuous cost
achieved by our kernel functions.
At $\varepsilon=10^{-3}$, our continuous cost is approximately
$8.572$, compared with $16.815$ and $12.755$ for the two
Low--Somma choices, giving improvements by factors of
approximately $1.96$ and $1.49$.
At $\varepsilon=10^{-10}$, the respective costs are $34.483$,
$85.047$, and $68.227$, giving improvements by factors of
approximately $2.47$ and $1.98$. \\

The factor of approximately $2.1$ quoted in the abstract
compares leading coefficients.
For their numerically optimized family with variable $j$ and $y$,
Low and Somma report the empirical cost fit
$2.57\log_2(1/(114.8\varepsilon))$
in \cite[Fig.~1]{low2025};
the coefficient of $\log(1/\varepsilon)$ is $2.57/\log 2$.
Comparing this fitted coefficient with our proved coefficient
$2\ee/\pi$ in \eqref{eq:optscaling} gives
\begin{equation}
\frac{2.57/\log 2}{2\ee/\pi}
=\frac{2.57\pi}{2\ee\log 2}
\approx 2.14.
\end{equation}
Additionally, \cite[Fig.~1]{low2025} gives an empirically fitted, numerically computed lower bound, where the leading coefficient of $\log(1/\varepsilon)$ is approximately $1.22$. By Theorem~\ref{thm:optimal-kernel}, the gap between that coefficient and $2\ee/\pi$ cannot be reduced by any kernel function in $\mc{A}(S_{y_0})$ subject to the error bound \eqref{eq:E}.

\paragraph{Sufficient quadrature node counts.}
We compare node counts at total operator-error tolerance $2\varepsilon$,
for $\varepsilon=10^{-1},\ldots,10^{-10}$, setting
$\int_0^\tau\|\m L(s)\|_2\,\mathrm ds=1$ when selecting
the quadrature nodes.
We follow the strip-width choice of \cite[Lem.~12]{low2025}, using half of
the respective admissible strip half-widths. Correspondingly, we take a strip
of half-width $1/2$ for the Low--Somma kernel functions and $\eta=\eta_0(\varepsilon)/2$
for our transformed optimal kernel functions, with $\eta_0$ defined in
\eqref{eq:quadrature_strip_width}. Neither choice is optimized
to minimize node counts; the comparison uses a common convention.
For our optimal kernel functions, we allocate error budget
$\varepsilon$ to the continuous approximation and set
$\varepsilon_{\rm disc}=\varepsilon_{\rm trunc}=\varepsilon/2$.
We then choose the largest step $\delta_\xi$ allowed by
\eqref{eq:quadrature_parameters_choice} and, with this step
fixed, the smallest $N$ satisfying the truncation condition in \eqref{eq:quadrature_parameters_choice}.
These choices also keep the discrete cost close to the continuous cost:
\eqref{eq:discrete_cost_numerics} gives
$\lvert R\alpha_{N,\delta_\xi}-Q_\star(\varepsilon)\rvert\leq R\varepsilon$.\\

For the Low--Somma kernels, we use the strip-contour estimate of \cite[Lem.~10]{low2025}, taking the strip half-width to be $1/2$,
as in the proof of \cite[Lem.~12]{low2025}.
All the tabulated choices satisfy $y\geq1$, so the kernel functions
are holomorphic on $\lvert\operatorname{Im}z\rvert<1$.
We have
\begin{equation}
\sup_{\abs{v}<\tfrac12}\frac{1}{\sqrt{2\pi}}
\int_{\R}|f_{j,y}(x+\ii v;\gamma,c)|\,\mathrm dx
\leq M,
\end{equation}
where, for the tabulated parameter choices,\footnote{To obtain \eqref{eq:ls_boundary_numerics}, write
$z=x+\ii v$ with $|v|\leq1/2$ in
\eqref{eq:generalized_kernel_intro}.
When $j=2$ and $y=1$, the denominator factors $1-\ii z$
and $1+\ii z$ each have modulus at least $\sqrt{x^2+1/4}$.
When $j>2$ and $y\geq1$, we instead use
$|1-\ii z|\geq1/2$ and
$|y+\ii z|^{j-1}\geq
(x^2+(y-1/2)^2)^{(j-1)/2}$.
The numerator factors satisfy
$\lvert\ee^{c(1-\ii z)}\rvert\leq\ee^{c+\abs{c}/2}$ and
$\lvert\ee^{-(z^2+1)/(4\gamma^2)}\rvert\leq1$. Integration then gives the stated bounds using
$\int_{\R}(x^2+a^2)^{-p}\,\mathrm dx
=\sqrt{\pi}\,a^{1-2p}\Gamma(p-1/2)/\Gamma(p)$
for $a>0$, $p>1/2$.}
\begin{equation}\label{eq:ls_boundary_numerics}
M\coloneqq
\begin{dcases}
2\ee^{c+\abs{c}/2}
& j=2,\quad y=1\\
\displaystyle
\frac{(y+1)^{j-1}\ee^{c+\abs{c}/2}}
{\sqrt{\pi}(y-\tfrac12)^{j-2}}
\frac{\Gamma((j-2)/2)}{\Gamma((j-1)/2)}
& j>2,\quad y\geq1,
\end{dcases}
\end{equation}
where $\Gamma$ is the Gamma function. Together with $\int_0^\tau\lVert\m L(s)\rVert_2\,\mathrm ds=1$,
this gives an infinite-sum discretization error bounded by $2M\ee^{1/2}/(\ee^{\pi/\delta_x}-1)$ \cite[Lem.~10]{low2025}, where $\delta_x$ is the uniform grid step size. We use the published parameter values in \cite[Tab.~3]{low2025} and choose the largest step $\delta_x=R/N$, with $N\in\N$, for which the sum of the continuous error bound \eqref{eq:E}, evaluated numerically at the optimized $y_0$ values given in \cite[Tab.~3]{low2025}, and this discretization bound is at most $2\varepsilon$. Since $\lvert f_{j,y}(x;\gamma,c)\rvert $ decreases with $\abs{x}$, the discrete tail is bounded by the integral tail already included in the continuous error bound. \\

\begin{table}[H]
\centering
\small
\setlength{\tabcolsep}{5pt}
\renewcommand{\arraystretch}{1.12}
\begin{tabular*}{\linewidth}{@{\extracolsep{\fill}}
  c S[table-format=4.0] *{2}{S[table-format=3.0]} @{}}
\toprule
& \multicolumn{1}{c}{This work}
& \multicolumn{2}{c}{Low--Somma \cite{low2025}} \\
\cmidrule(lr){2-2}\cmidrule(l){3-4}
$\varepsilon$
& \multicolumn{1}{c}{$f_\star$ \eqref{eq:f_optimal}}
& \multicolumn{1}{c}{\eqref{eq:generalized_kernel_intro} with $j=2$, $y=1$}
& \multicolumn{1}{c}{\eqref{eq:generalized_kernel_intro} with variable $j,y$} \\
\midrule
$10^{-1}$  &   15 &  11 &   9 \\
$10^{-2}$  &   47 &  29 &  25 \\
$10^{-3}$  &  101 &  59 &  49 \\
$10^{-4}$  &  177 &  99 &  81 \\
$10^{-5}$  &  279 & 149 & 123 \\
$10^{-6}$  &  405 & 213 & 173 \\
$10^{-7}$  &  559 & 287 & 233 \\
$10^{-8}$  &  739 & 375 & 301 \\
$10^{-9}$  &  947 & 473 & 379 \\
$10^{-10}$ & 1185 & 585 & 465 \\
\bottomrule
\end{tabular*}
\caption{Sufficient quadrature node counts $2N+1$ at total
operator-error tolerance $2\varepsilon$, using the parameter
choices and error bounds described in the text.}
\label{tab:quadrature_nodes}
\end{table}

Table~\ref{tab:quadrature_nodes} lists the sufficient numbers
of nodes for the grids specified above.
Our quadrature uses more nodes in every tabulated case, but fewer than $2.6$ times those of either Low--Somma choice. The constant $\pi/2$ in \eqref{eq:phiR} is conventional and is not optimized here; adapting it to the singularity structure of $f_\star$ using the methods of \cite{slevinsky2015on} can reduce the node counts significantly. We emphasize that the number of quadrature nodes does not enter the query cost \eqref{eqn:query}, since the constituent Hamiltonian propagators are selected coherently \cite[Sec.~5.1]{low2025}. However, our non-uniform quadrature additionally requires loading the node values $\phi_R(n\delta_{\xi})$, which affects the non-query gate count.

\section{Optimal kernel functions}\label{sect:proof_thmmain}

We prove Theorem~\ref{thm:optimal-kernel} in three steps.
First, in Section~\ref{subsec:Hardy_reformulation}, we reformulate the optimization problem in Hardy space and reduce it to outer functions.
Next, in Section~\ref{subsec:fixed_optimization}, we use Jensen's inequality to identify and construct the unique kernel function minimizing the continuous subnormalization \eqref{eq:alpha} for fixed truncation radius and domain, subject to the error constraint in Problem~\ref{prob:main}.
Finally, in Section~\ref{sect:proof_col}, we optimize over these
remaining parameters and derive the asymptotic expansion of the minimum cost.

\subsection{Hardy space reformulation and outer-factor reduction}
\label{subsec:Hardy_reformulation}

For the following proofs, recall the parameterization of admissible kernel
functions in terms of the holomorphic function $h$ in
\eqref{eqn:kernel_f_h}. In this representation, the reciprocal of $B(\ccdot;y_0)$ encodes the prescribed pole structure of $f$, while the
remaining freedom is encoded in $h\in H^1(S_{y_0})$. For fixed $y_0$,
\eqref{eqn:kernel_f_h} therefore identifies the optimization over admissible
kernel functions with the corresponding optimization over
functions $h\in H^1(S_{y_0})$ satisfying the normalization \eqref{eq:norm_h}. Once the optimal $h$ is determined, the
optimal kernel function $f$ is recovered from
\eqref{eqn:kernel_f_h}. \\

We emphasize that no holomorphic extension of $h$ through
$\partial S_{y_0}$ is assumed. The condition is instead that
$h\in H^1(S_{y_0})$. By Lemma~\ref{lem:Hardyboundary}, the upper boundary
value $h_0$ exists almost everywhere and belongs to $L^1(\R)$; when
$y_0\in(1,\infty)$, the same holds for the lower boundary value
$h_{y_0}$. Since $B(\ccdot;y_0)$ has unimodular boundary values, the
corresponding boundary values of $f$ also belong to $L^1(\R)$. \\

\begin{prop}[Subnormalization-optimal kernel functions for LCHS]
\label{prop:subnormalization-optimal-kernel}
Given $\varepsilon>0$, $R>0$, and $y_0\in (1,\infty]$,  the unique minimizer of $\alpha_R[f]$ over all $f\in\mc{A}(S_{y_0})$ subject to the constraint $\mc{E}_{R,y_0}[f]\leq \varepsilon$ is given by
\begin{equation}
\label{eq:fopt}
f_\star(z)\coloneqq
\begin{dcases}
\frac{1}{\sqrt{2\pi}}\frac{2\pi \varepsilon}{1- m_{R,y_0}}\frac{\sin\big(\tfrac{\pi}{y_0}\big)}{4y_0}\frac{
\exp\Bigg(\dfrac{\Delta_{R,y_0,\varepsilon}}{\pi \mathrm{i}}\log\dfrac{\sinh\big(\frac{\pi}{2y_0}(z+R)\big)}{\sinh\big(\frac{\pi}{2y_0}(z-R)\big)}\Bigg)}{\sinh\big(\frac{\pi}{2y_0}(z-\mathrm{i}) \big)\sinh\big(\frac{\pi}{2y_0}(z+\mathrm{i})\big)}
& y_0\in (1,\infty), \\
\frac{1}{\sqrt{2\pi}}\frac{2\varepsilon}{1-m_{R,\infty}}
\frac{\exp\bigg( \dfrac{\Delta_{R,\infty,\varepsilon}}{\pi\mathrm{i}}
\log\dfrac{z+R}{z-R}\bigg)}
{z^2+1}
& y_0=\infty
\end{dcases}
\end{equation}
where the logarithms are chosen so that the compositions are
holomorphic on 
\begin{equation}
\label{eq:Omega_cut}
\Omega_{R,y_0}\coloneqq \{z\in\C:\lvert\operatorname{Im}z\rvert<y_0\}
\setminus\big((-\infty,-R]\cup[R,\infty)\big)
\end{equation}
and agree with the principal branch of the logarithm on $S_{y_0}$.
Here,
\begin{equation}\label{eq:mRy0}
m_{R,y_0}\coloneqq \begin{dcases}
\frac{2}{\pi}\mathrm{arctan}\bigg(\frac{\tanh(\frac{\pi R}{2y_0})}{\tan(\frac{\pi}{2y_0})} \bigg) & y_0\in (1,\infty),\\
\frac{2}{\pi}\mathrm{arctan}(R) & y_0=\infty
\end{dcases}
\qquad
\Delta_{R,y_0,\varepsilon}\coloneqq \frac{1}{m_{R,y_0}}\log \frac{1-m_{R,y_0}}{\varepsilon}.
\end{equation}
Moreover,
\begin{equation}\label{eq:alphaoptimal}
\alpha_R[f_\star]= m_{R,y_0}\bigg( \frac{\varepsilon}{1-m_{R,y_0}} \bigg)^{-(1-m_{R,y_0})/m_{R,y_0}}.
\end{equation}
\end{prop}

For $\varepsilon\in (0,1)$, $y_0=\infty$ and $R=R_{\star}(\varepsilon)$, the expression \eqref{eq:fopt} coincides with the optimal kernel function \eqref{eq:f_optimal} in Theorem~\ref{thm:optimal-kernel}. To prove Proposition~\ref{prop:subnormalization-optimal-kernel}, we first
reformulate the fixed-\((R,y_0)\) optimization problem in terms of the associated Hardy space function.

\begin{prop}[Hardy space reformulation]
\label{prop:reformulatedprob}
Fix $\varepsilon>0$, $R>0$, and $y_0\in(1,\infty]$. Let
$f\in\mc A(S_{y_0})$ and let $h\in H^1(S_{y_0})$ be its associated
Hardy space function under the correspondence
\eqref{eqn:kernel_f_h}. Then
\begin{equation}
\alpha_R[f]=\tilde{\alpha}_R[h],
\qquad
\mc E_{R,y_0}[f]=\tilde{\mc E}_{R,y_0}[h],
\end{equation}
where
\begin{equation}\label{eq:alphatilde}
    \tilde{\alpha}_R[h]
\coloneqq
\frac{1}{2\pi}\int_{-R}^{R} \lvert h_0(x)\rvert\,\mathrm{d}x,
\end{equation}
and
\begin{equation}\label{eq:Etilde}
    \tilde{\mathcal E}_{R,y_0}[h]
\coloneqq
\begin{dcases}
\dfrac{1}{2\pi}\displaystyle\int_{\mathbb R\setminus[-R,R]}\lvert h_0(x)\rvert\,\mathrm{d}x
+\dfrac{1}{2\pi}\displaystyle\int_{\mathbb R} \lvert h_{y_0}(x)\rvert\,\mathrm{d}x
& y_0\in(1,\infty), \\
\dfrac{1}{2\pi}\displaystyle\int_{\mathbb R\setminus[-R,R]} \lvert h_0(x)\rvert\,\mathrm{d}x
& y_0=\infty.
\end{dcases}
\end{equation}
Consequently, $f_\star\in\mc A(S_{y_0})$ is a minimizer of the
fixed-parameter problem in
Proposition~\ref{prop:subnormalization-optimal-kernel} if and only if
its associated function $h_\star$ minimizes $\tilde{\alpha}_R[h]$ over
all $h\in H^1(S_{y_0})$ satisfying the normalization
\eqref{eq:norm_h} and the constraint
$\tilde{\mc E}_{R,y_0}[h]\leq\varepsilon$.
\end{prop}

\begin{proof}
By Definition~\ref{def:admissible}, the correspondence
\eqref{eqn:kernel_f_h} identifies $\mc A(S_{y_0})$ with the functions
$h\in H^1(S_{y_0})$ satisfying the normalization
\eqref{eq:norm_h}. Taking boundary values in
\eqref{eqn:kernel_f_h} and using the fact that $B(\ccdot;y_0)$ has
unimodular boundary values shows that the boundary moduli of $f$ and
$h$ differ by the constant factor $1/\sqrt{2\pi}$. Substitution into
\eqref{eq:alpha} and \eqref{eq:E}, followed by comparison with
\eqref{eq:alphatilde} and \eqref{eq:Etilde}, gives the two identities
in the statement. Thus the correspondence identifies the feasible sets
and preserves the objective, so the minimizers correspond.
\end{proof}

By Proposition~\ref{prop:reformulatedprob}, it remains to solve the
corresponding Hardy space problem. We first show that every minimizer
must be outer in the sense of Definition~\ref{def:inner_outer}.

\begin{lem}[Outer reduction]
\label{lem:outer}
Let $h\in H^1(S_{y_0})$ satisfy the normalization
\eqref{eq:norm_h} and the constraint
$\tilde{\mc E}_{R,y_0}[h]\leq\varepsilon$. If $h$ is not outer, then
there exists a feasible outer function $\hat h\in H^1(S_{y_0})$ such
that
\begin{equation}
\tilde{\alpha}_R[\hat h]
<
\tilde{\alpha}_R[h].
\end{equation}
Consequently, the infimum of the Hardy space problem may be taken over
outer functions, and every minimizer is outer.
\end{lem}

\begin{proof}
By Theorem~\ref{thm:beurling}, $h$ admits an inner--outer
factorization
\begin{equation}
h=h_{\mathrm{in}}h_{\mathrm{out}},
\end{equation}
unique up to reciprocal unimodular constants. Since $h$ is not outer,
$h_{\mathrm{in}}$ is nonconstant. Define
\begin{equation}
\label{eq:hhat}
\hat h(z)
\coloneqq
h_{\mathrm{in}}(-\ii)h_{\mathrm{out}}(z).
\end{equation}
Since $h(-\ii)\neq0$, we have
$h_{\mathrm{in}}(-\ii)\neq0$. Thus $\hat h$ is a nonzero constant
multiple of an outer function and is therefore outer. \\

Since $h_{\mathrm{in}}$ is nonconstant,
Theorem~\ref{thm:beurling} gives $\abs{h_{\mathrm{in}}(-\ii)}<1$.
Moreover,
\begin{equation}
\hat h(-\ii)
=
h_{\mathrm{in}}(-\ii)
h_{\mathrm{out}}(-\ii)
=
h(-\ii),
\end{equation}
so $\hat h$ satisfies the normalization \eqref{eq:norm_h}. The boundary-modulus identities in
Theorem~\ref{thm:beurling} give
\begin{equation}
\tilde{\mc E}_{R,y_0}[\hat h]
=
\abs{h_{\mathrm{in}}(-\ii)}
\tilde{\mc E}_{R,y_0}[h]
\leq
\tilde{\mc E}_{R,y_0}[h]
\leq
\varepsilon.
\end{equation}
Thus $\hat h$ is feasible. Similarly,
\begin{equation}
\tilde{\alpha}_R[\hat h]
=
\abs{h_{\mathrm{in}}(-\ii)}
\tilde{\alpha}_R[h].
\end{equation}
By the defining Poisson representation of an outer function,
$h_{\mathrm{out},0}$ is nonzero almost everywhere. Since $R>0$,
it follows from Theorem~\ref{thm:beurling} that $\tilde{\alpha}_R[h]>0$, and therefore $\tilde{\alpha}_R[\hat h]< \tilde{\alpha}_R[h]$.
This proves the result.
\end{proof}

\subsection{Optimization at fixed truncation radius and domain}
\label{subsec:fixed_optimization}

The proof of Proposition~\ref{prop:subnormalization-optimal-kernel} relies on Jensen's inequality with the equality case identifying the optimal boundary modulus.
We recall this standard result for completeness.

\begin{lem}[Jensen's inequality\normalfont{~\cite[Thm.~2.2]{lieb2001analysis}}]
\label{lem:Jensen}
Let $(X,\omega)$ be a measure space with $0<\omega(X)<\infty$ and let $\varphi:X\to(0,\infty)$ be measurable with
$\varphi,\log \varphi\in L^1(X,\omega)$. Then,
\begin{equation}
\frac{1}{\omega(X)}
\int_X \log \varphi\,\mathrm{d}\omega
\leq
\log\bigg(
\frac{1}{\omega(X)}
\int_X \varphi\,\mathrm{d}\omega
\bigg).
\end{equation}
Equality holds if and only if $\varphi$ is constant $\omega$-almost everywhere.
\end{lem}

We now have all the tools to prove Proposition~\ref{prop:subnormalization-optimal-kernel}.

\begin{proof}[Proof of Proposition~\ref{prop:subnormalization-optimal-kernel}]
We begin with a reduction common to the finite-strip and half-plane cases.

\paragraph{Step 1: reduction to outer boundary data.}
By Proposition~\ref{prop:reformulatedprob} and Lemma~\ref{lem:outer},
the infimum may be taken over outer functions $h\in H^1(S_{y_0})$
satisfying the normalization \eqref{eq:norm_h} and the error
constraint $\tilde{\mc E}_{R,y_0}[h]\leq\varepsilon$. Let $h$ be such
a function. By \eqref{eq:alphatilde} and \eqref{eq:Etilde}, both the
objective and the error constraint depend only on its boundary moduli. \\

The defining Poisson representations of an outer function are given by
\eqref{eqn:outer_poisson_representation} and
\eqref{eqn:outer_poisson_representation_halfplane}, with the Poisson
kernels defined explicitly in \eqref{eqn:poisson_kernels} and
\eqref{eqn:halfplane_poisson_kernel}. Setting $x=0$ and $y=1$ in these
representations, using that the Poisson kernels are even, and
substituting the normalization \eqref{eq:norm_h} gives
\begin{equation}
\label{eqn:constraint_finite1}
\int_{\R}
P_0(x,1)\log\,\abs{h_0(x)}\,\mathrm{d}x
+
\int_{\R}
P_{y_0}(x,1)\log\,\abs{h_{y_0}(x)}\,\mathrm{d}x
=
-\log\frac{2y_0\sin\!\big(\frac{\pi}{y_0}\big)}{\pi}
\end{equation}
for $y_0\in(1,\infty)$ and
\begin{equation}
\label{eqn:constraint_finite2}
\int_{\R}
\hat P_0(x,1)\log\,\abs{h_0(x)}\,\mathrm{d}x
=
-\log 2
\end{equation}
for $y_0=\infty$. Thus every outer feasible function gives boundary data satisfying
\eqref{eqn:constraint_finite1}--\eqref{eqn:constraint_finite2}. 

\paragraph{Step 2: determine the optimal boundary modulus.} We  unify both cases under a common
harmonic-measure theoretic formulation and apply Jensen's inequality. Since the Poisson kernels are nonnegative, they induce a measure
$\omega_{y_0}$ on $\partial S_{y_0}$ defined by
\begin{equation}
\label{eqn:harmonic_measure1}
\int_{\partial S_{y_0}}\psi(z)\,
\mathrm{d}\omega_{y_0}(z)
\coloneqq
\begin{dcases}
\displaystyle
\int_{\R}
\psi(x)P_0(x,1)\,\mathrm{d}x
+
\int_{\R}
\psi(x-\ii y_0)P_{y_0}(x,1)\,\mathrm{d}x
& y_0\in(1,\infty),
\\
\displaystyle
\int_{\R}
\psi(x)\hat P_0(x,1)\,\mathrm{d}x
& y_0=\infty
\end{dcases}
\end{equation}
for every measurable function
$\psi:\partial S_{y_0}\to[0,\infty]$.
The Poisson mass calculations in the proof of
Lemma~\ref{lem:PoissonL1} give
\begin{equation}
\omega_{y_0}(\partial S_{y_0})=1.
\end{equation}
Thus $\omega_{y_0}$ is a probability measure on
$\partial S_{y_0}$. \\

To apply Jensen's inequality, we normalize the boundary modulus by the
Poisson kernels so that the objective and error functionals become
integrals with respect to the probability measure $\omega_{y_0}$.
Define $\varphi:\partial S_{y_0}\to[0,\infty)$ almost everywhere by
\begin{equation}
\label{eq:varphi}
\varphi(x)
\coloneqq
\begin{dcases}
\displaystyle
\frac{\abs{h_0(x)}}{P_0(x,1)}
& y_0\in(1,\infty),
\\
\displaystyle
\frac{\abs{h_0(x)}}{\hat P_0(x,1)}
& y_0=\infty
\end{dcases}
\qquad
(x\in\R)
\end{equation}
and, when $y_0\in(1,\infty)$, by
\begin{equation}
\varphi(x-\ii y_0)
\coloneqq
\frac{\abs{h_{y_0}(x)}}{P_{y_0}(x,1)}
\qquad
(x\in\R).
\end{equation}

We first verify that $\varphi$ satisfies the integrability assumptions of
Jensen's inequality in Lemma~\ref{lem:Jensen}. Since
$h\in H^1(S_{y_0})$, its boundary values belong to $L^1(\R)$.
The definition of $\varphi$ together with
\eqref{eqn:harmonic_measure1} therefore gives
$\varphi\in L^1(\partial S_{y_0},\omega_{y_0})$.\\

The Poisson kernels at unit depth are bounded, and
$\log^+t\leq t$ for $t>0$. Hence the positive parts of the boundary
functions $\log\,\abs{h_0}$ and, in the finite-strip case,
$\log\,\abs{h_{y_0}}$ are integrable against the corresponding Poisson
kernels. Since the right-hand side of
\eqref{eqn:constraint_finite1} or \eqref{eqn:constraint_finite2}, as appropriate,
is finite, their negative parts are integrable as well. Thus the boundary values of $h$ are nonzero almost
everywhere. \\

Lemma~\ref{lem:poisson_log_integral} shows that
$\log P_0(\ccdot,1)$, $\log P_{y_0}(\ccdot,1)$, and
$\log\hat P_0(\ccdot,1)$ are integrable against their corresponding
Poisson kernels whenever they occur. It follows from the definition of
$\varphi$ that
$\log\varphi\in L^1(\partial S_{y_0},\omega_{y_0})$ and that
$\varphi>0$ almost everywhere. \\

By \eqref{eqn:constraint_finite1}--\eqref{eqn:constraint_finite2} and
Lemma~\ref{lem:poisson_log_integral},
\begin{align}
\int_{\partial S_{y_0}}
\log\varphi\,\mathrm{d}\omega_{y_0}
=&\;
\begin{dcases}
\displaystyle
-\log\frac{2y_0\sin(\pi/y_0)}{\pi}
+
\log\!\big(4y_0\sin(\pi/y_0)\big)
& y_0\in(1,\infty),
\\
-\log 2+\log(4\pi)
& y_0=\infty
\end{dcases}
\nonumber\\
=&\;
\log(2\pi).
\label{eqn:central_logarithmic_identity}
\end{align}
Thus we may write
\begin{equation}
\label{eq:constraint_omega1}
\int_{[-R,R]}\log \varphi\,\mathrm{d}\omega_{y_0}
+
\int_{\partial S_{y_0}\setminus[-R,R]}
\log \varphi\,\mathrm{d}\omega_{y_0}
=
\log(2\pi).
\end{equation}
Furthermore, by the definitions of the objective functional
\eqref{eq:alphatilde} and the error functional \eqref{eq:Etilde},
\begin{equation}
\label{eq:phi_to_alpha_E1}
\int_{[-R,R]} \varphi\,\mathrm{d}\omega_{y_0}
=
2\pi\tilde{\alpha}_R[h],
\qquad
\int_{\partial S_{y_0}\setminus[-R,R]}
\varphi\,\mathrm{d}\omega_{y_0}
=
2\pi\tilde{\mathcal{E}}_{R,y_0}[h].
\end{equation}
Moreover, by the truncated Poisson mass formula of
Lemma~\ref{lem:int_poisson},
\begin{equation}
\label{eq:omega_values1}
\omega_{y_0}([-R,R])
=
m_{R,y_0},
\qquad
\omega_{y_0}(\partial S_{y_0}\setminus[-R,R])
=
1-m_{R,y_0}.
\end{equation}
Since the Poisson kernels are strictly positive and $R>0$,
\eqref{eq:omega_values1} gives $m_{R,y_0}\in(0,1)$. \\

We now apply Jensen's inequality separately on the two boundary regions.
By Lemma~\ref{lem:Jensen}, we have
\begin{equation}
\label{eq:Jensen_1}
\int_{[-R,R]}\log\varphi\,\mathrm{d}\omega_{y_0}
\leq
m_{R,y_0}
\log\bigg(
\frac{2\pi\tilde{\alpha}_R[h]}{m_{R,y_0}}
\bigg),
\end{equation}
and
\begin{align}
\label{eq:Jensen_2}
\int_{\partial S_{y_0}\setminus[-R,R]}
\log\varphi\,\mathrm{d}\omega_{y_0}
\leq\;&
(1-m_{R,y_0})
\log\bigg(
\frac{2\pi\tilde{\mathcal E}_{R,y_0}[h]}
{1-m_{R,y_0}}
\bigg)
\nonumber\\
\leq\;&
(1-m_{R,y_0})
\log\bigg(
\frac{2\pi\varepsilon}
{1-m_{R,y_0}}
\bigg),
\end{align}
where the second inequality follows from the feasibility of $h$.
Combining these estimates with
\eqref{eq:constraint_omega1} gives
\begin{equation}
0
\leq
m_{R,y_0}
\log\!\bigg(
\frac{\tilde{\alpha}_R[h]}{m_{R,y_0}}
\bigg)
+
(1-m_{R,y_0})
\log\!\bigg(
\frac{\varepsilon}{1-m_{R,y_0}}
\bigg).
\end{equation}
Hence,
\begin{equation}
\label{eq:alpha_Jensen1}
\tilde{\alpha}_R[h]
\geq
m_{R,y_0}
\bigg(
\frac{1-m_{R,y_0}}{\varepsilon}
\bigg)^{(1-m_{R,y_0})/m_{R,y_0}}.
\end{equation}
Together with Lemma~\ref{lem:outer}, this is a lower bound for every
feasible function. Indeed, every non-outer feasible function admits a
feasible outer competitor with strictly smaller objective value, while
every outer feasible function satisfies \eqref{eq:alpha_Jensen1}. \\

Equality in \eqref{eq:alpha_Jensen1} holds if and only if equality holds
throughout the preceding chain of inequalities. Equivalently,
$\tilde{\mathcal E}_{R,y_0}[h]=\varepsilon$
and $\varphi$ is constant
$\omega_{y_0}$-almost everywhere on each of
$[-R,R]$ and
$\partial S_{y_0}\setminus[-R,R]$.
The subsequent reconstruction will show that these equality conditions
are attainable, and hence that the lower bound is sharp.
The constant value of $\varphi$ on each boundary region is determined by its normalized average over that region, which is given by \eqref{eq:phi_to_alpha_E1} and \eqref{eq:omega_values1}.
Thus, in the equality case, for $\omega_{y_0}$-almost every
$z\in\partial S_{y_0}$,
\begin{equation}
\varphi(z)=
\begin{dcases}
2\pi
\bigg(
\frac{1-m_{R,y_0}}{\varepsilon}
\bigg)^{(1-m_{R,y_0})/m_{R,y_0}}
& z\in[-R,R],
\\
\displaystyle
\frac{2\pi\varepsilon}{1-m_{R,y_0}}
& z\in
\partial S_{y_0}\setminus[-R,R].
\end{dcases}
\end{equation}

Recalling the definition of $\varphi$, the equality condition above is
equivalent to the following boundary moduli. When
$y_0\in(1,\infty)$,
\begin{equation}
\abs{h_0(x)}
=
P_0(x,1)\times
\begin{dcases}
\displaystyle
2\pi
\bigg(
\frac{1-m_{R,y_0}}{\varepsilon}
\bigg)^{(1-m_{R,y_0})/m_{R,y_0}}
& x\in[-R,R],
\\
\displaystyle
\frac{2\pi\varepsilon}{1-m_{R,y_0}}
& x\in\R\setminus[-R,R]
\end{dcases}
\label{eq:optimal_upper_boundary_modulus}
\end{equation}
and
\begin{equation}
\abs{h_{y_0}(x)}
=
\frac{2\pi\varepsilon}{1-m_{R,y_0}}
P_{y_0}(x,1).
\label{eq:optimal_lower_boundary_modulus}
\end{equation}
When $y_0=\infty$,
\begin{equation}
\abs{h_0(x)}
=
\hat P_0(x,1)
\times\begin{dcases}
\displaystyle
2\pi
\bigg(
\frac{1-m_{R,\infty}}{\varepsilon}
\bigg)^{(1-m_{R,\infty})/m_{R,\infty}}
& x\in[-R,R],
\\
\displaystyle
\frac{2\pi\varepsilon}{1-m_{R,\infty}}
& x\in\R\setminus[-R,R].
\end{dcases}
\label{eq:optimal_halfplane_boundary_modulus}
\end{equation}
The equalities
\eqref{eq:optimal_upper_boundary_modulus}--
\eqref{eq:optimal_halfplane_boundary_modulus}
hold almost everywhere. Thus equality in
\eqref{eq:alpha_Jensen1} holds precisely when the error constraint is
saturated and the boundary moduli of $h$ have the forms above.

\paragraph{Step 3: boundary logarithms.}
We now write the boundary logarithms obtained in Step~2 in a form
suitable for holomorphic reconstruction.
By \eqref{eq:mRy0},
\begin{equation}
\label{eq:boundary_coefficient_identity}
2\pi
\bigg(
\frac{1-m_{R,y_0}}{\varepsilon}
\bigg)^{(1-m_{R,y_0})/m_{R,y_0}}
=
\frac{2\pi\varepsilon}{1-m_{R,y_0}}
\mathrm{e}^{\Delta_{R,y_0,\varepsilon}}.
\end{equation}
We write $\chi_{[-R,R]}$ for the indicator function of
$[-R,R]\subset\R$, so that $\chi_{[-R,R]}(x)=1$ for
$x\in[-R,R]$ and $\chi_{[-R,R]}(x)=0$ otherwise. \\

For $y_0\in(1,\infty)$, the boundary logarithms in
\eqref{eq:optimal_upper_boundary_modulus} and
\eqref{eq:optimal_lower_boundary_modulus} may be written as
\begin{equation}
\begin{split}
\log\,\abs{h_0(x)}
=&\;
\log\!\bigg(\frac{2\pi\varepsilon}{1-m_{R,y_0}}\bigg)
+\log P_0(x,1)
+\Delta_{R,y_0,\varepsilon}\chi_{[-R,R]}(x),
\\
\log\,\abs{h_{y_0}(x)}
=&\; 
\log\!\bigg(\frac{2\pi\varepsilon}{1-m_{R,y_0}}\bigg)
+\log P_{y_0}(x,1).
\end{split}
\end{equation}
When $y_0=\infty$, the boundary logarithm in
\eqref{eq:optimal_halfplane_boundary_modulus} may be written as
\begin{equation}
\log\,\abs{h_0(x)}
=
\log\!\bigg(\frac{2\pi\varepsilon}{1-m_{R,\infty}}\bigg)
+\log\hat P_0(x,1)
+\Delta_{R,\infty,\varepsilon}\chi_{[-R,R]}(x).
\end{equation}
These identities hold almost everywhere.

\paragraph{Step 4: holomorphic reconstruction.}
Let
\begin{equation}
\label{eq:C}
C\coloneqq
\log\!\bigg(\frac{2\pi\varepsilon}{1-m_{R,y_0}}\bigg).
\end{equation}
For $\theta\in\R$, define the candidate function
\begin{equation}
\label{eq:h_theta}
h_\theta(z)
\coloneqq
\mathrm{e}^{\mathrm{i}\theta}\mathrm{e}^C\mathrm{e}^{G_{y_0}(z)},
\end{equation}
where the form of $G_{y_0}(z)$ is case-specific.

\paragraph{Step 4a: finite-strip case, $y_0\in(1,\infty)$.}
Let $G_{y_0}\coloneqq G_{1,y_0}+G_{2,y_0}$ with
\begin{equation}\label{eqn:Gs_finite}
\begin{split}
G_{1,y_0}(z)&\coloneqq\log\!\Bigg(\frac{\sin(\frac{\pi}{y_0})}{4y_0\sinh^2\!\big(\frac{\pi}{2y_0}(z-\mathrm{i})\big)}\Bigg),\\
G_{2,y_0}(z)&\coloneqq\frac{\Delta_{R,y_0,\varepsilon}}{\pi \mathrm{i}}
\log\!\Bigg(\frac{\sinh\!\big(\frac{\pi}{2y_0}(z+R)\big)}{\sinh\!\big(\frac{\pi}{2y_0}(z-R)\big)}\Bigg).
\end{split}
\end{equation}
We first show that the function $G_{y_0}$ is holomorphic on $S_{y_0}$.
The zeros of $\sinh(\frac{\pi}{2y_0}(z-\mathrm{i}))$ occur at
$z=\mathrm{i}+2n\mathrm{i}y_0$, and those of
$\sinh(\frac{\pi}{2y_0}(z\pm R))$ at
$z=\mp R+2n\mathrm{i}y_0$, none of which lie in $S_{y_0}$.
Consequently, the function appearing inside the logarithm in
$G_{1,y_0}$ is holomorphic and nonzero on $S_{y_0}$. Likewise, the
numerator and denominator of the quotient appearing in $G_{2,y_0}$ are
nonzero on $S_{y_0}$, so this quotient is also holomorphic and nonzero
there. Since $S_{y_0}$ is simply connected, both functions admit
holomorphic logarithms on $S_{y_0}$. Moreover, for
$z=x-\mathrm{i}y\in S_{y_0}$,
\begin{equation}
\operatorname{Im}\bigg(
\frac{\sinh\big(\frac{\pi}{2y_0}(z+R)\big)}
{\sinh\big(\frac{\pi}{2y_0}(z-R)\big)}
\bigg)
=
\frac{
\sin\big(\frac{\pi y}{y_0}\big)
\sinh\big(\frac{\pi R}{y_0}\big)}
{2\abs{\sinh\big(\frac{\pi}{2y_0}(z-R)\big)}^2}
>0.
\end{equation}
Thus the quotient appearing in $G_{2,y_0}$ takes values in the upper
half-plane, and we take the principal branch of the logarithm in its
definition. \\

We now verify that $h_\theta$ has the prescribed boundary moduli.
For the upper boundary,
\begin{align}
\operatorname{Re}G_{1,y_0}(x)
=&\;  \log\,\Bigg\lvert \frac{\sin(\frac{\pi}{y_0})}{4y_0\sinh^2\!\big(\frac{\pi}{2y_0}(x-\mathrm{i})\big)} \Bigg\rvert
\nonumber \\
=&\; \log\,\Bigg\lvert \frac{\sin(\frac{\pi}{y_0})}{2y_0(\cosh(\frac{\pi x}{y_0})-\cos(\frac{\pi}{y_0}))} \Bigg\rvert
= \log P_0(x,1),
\end{align}
using $\big|\sinh(u-\mathrm{i}v)\big|^2=\tfrac{1}{2}(\cosh(2u)-\cos(2v))$. \\

Moreover, since
\begin{equation}
\frac{\sinh(\frac{\pi}{2y_0}(x+R))}{\sinh(\frac{\pi}{2y_0}(x-R))}
\begin{cases}
<0 & x\in(-R,R),\\
>0 & x\in\mathbb{R}\setminus[-R,R]
\end{cases}
\end{equation}
taking boundary limits from within $S_{y_0}$ gives, for every
$x\in\R\setminus\{-R,R\}$,
\begin{equation}
G_{2,y_0}(x)=
\begin{dcases}
\frac{\Delta_{R,y_0,\varepsilon}}{\pi \mathrm{i}}\Bigg(\log\,\Bigg\lvert \frac{\sinh(\frac{\pi}{2y_0}(x+R))}{\sinh(\frac{\pi}{2y_0}(x-R))} \Bigg\rvert + \mathrm{i}\pi\Bigg) & x\in(-R,R),\\
\frac{\Delta_{R,y_0,\varepsilon}}{\pi \mathrm{i}}\log\,\Bigg\lvert \frac{\sinh(\frac{\pi}{2y_0}(x+R))}{\sinh(\frac{\pi}{2y_0}(x-R))} \Bigg\rvert & x\in\mathbb{R}\setminus[-R,R]
\end{dcases}
\end{equation}
so that, for almost every $x\in\R$,
\begin{equation}
\operatorname{Re}G_{2,y_0}(x)=
\begin{dcases}
\Delta_{R,y_0,\varepsilon} & x\in[-R,R],\\
0 & x\in\mathbb{R}\setminus[-R,R].
\end{dcases}
\end{equation}
Thus $\operatorname{Re}(G_{1,y_0}(x)+G_{2,y_0}(x))
=\log P_0(x,1)+\Delta_{R,y_0,\varepsilon}\chi_{[-R,R]}(x)$
almost everywhere,
which matches $\log\,\abs{h_0(x)}$ after adding $C$ \eqref{eq:C}. \\

For the lower boundary,
\begin{align}
\operatorname{Re}G_{1,y_0}(x-\mathrm{i}y_0)
&= \log\,\Bigg\lvert \frac{\sin(\frac{\pi}{y_0})}{4y_0\big|\sinh\!\big(\frac{\pi}{2y_0}(x-\mathrm{i}(y_0+1))\big)\big|^2} \Bigg\rvert \notag\\
&= \log\,\Bigg\lvert \frac{\sin(\frac{\pi}{y_0})}{2y_0(\cosh(\frac{\pi x}{y_0})+\cos(\frac{\pi}{y_0}))} \Bigg\rvert
= \log P_{y_0}(x,1),
\end{align}
while $G_{2,y_0}(x-\mathrm{i}y_0)$ is purely imaginary, so $\operatorname{Re}G_{2,y_0}(x-\mathrm{i}y_0)=0$.
Hence $\operatorname{Re}(G_{1,y_0}+G_{2,y_0})$ reproduces $\log\,\abs{h_{y_0}(x)}$ after adding $C$ \eqref{eq:C}.

\paragraph{Step 4b: half-plane case, $y_0=\infty$.}
Define
\begin{equation}
G_{1,\infty}(z)\coloneqq\log\!\bigg(\frac{1}{\pi(z-\mathrm{i})^2}\bigg),
\qquad
G_{2,\infty}(z)\coloneqq\frac{\Delta_{R,\infty,\varepsilon}}{\pi \mathrm{i}}
\log\!\bigg(\frac{z+R}{z-R}\bigg).
\end{equation}
The functions appearing inside these logarithms are holomorphic and
nonzero on the simply connected domain $S_\infty$.
We choose a holomorphic logarithm in the definition of $G_{1,\infty}$.
Moreover, for $z=x-\mathrm{i}y\in S_\infty$,
\begin{equation}
\operatorname{Im}\bigg(\frac{z+R}{z-R}\bigg)
=
\frac{2Ry}{(x-R)^2+y^2}
>0.
\end{equation}
Thus the quotient appearing in $G_{2,\infty}$ takes values in the upper
half-plane, and we take the principal branch of the logarithm in its
definition. In both cases, the quotient inside the logarithm defining
$G_{2,y_0}$ is holomorphic and nonzero on the simply connected domain $\Omega_{R,y_0}$. The chosen logarithm on $S_{y_0}$ thus extends uniquely to $\Omega_{R,y_0}$ \eqref{eq:Omega_cut}. \\

Let $G_\infty\coloneqq G_{1,\infty}+G_{2,\infty}$, so that $h_{\theta}$ is given by \eqref{eq:h_theta} with $y_0=\infty$. 
Then
\begin{equation}
\operatorname{Re}G_{1,\infty}(x)=\log \hat{P}_0(x,1),
\end{equation}
and, taking boundary limits from within $S_\infty$, for almost every $x\in\R$,
\begin{equation}
\operatorname{Re}G_{2,\infty}(x)=
\begin{cases}
\Delta_{R,\infty,\varepsilon} & x\in[-R,R],\\
0 & x\in\mathbb{R}\setminus[-R,R].
\end{cases}
\end{equation}
Hence
\begin{equation}
\log\,\abs{h_\theta(x)}
=
\operatorname{Re}\big(G_{1,\infty}(x)+G_{2,\infty}(x)\big)+C.
\end{equation}

It remains to verify that these candidates belong to the relevant
Hardy spaces. Since the quotients appearing in $G_{2,y_0}$ and
$G_{2,\infty}$ take values in the upper half-plane, their principal
arguments belong to $(0,\pi)$. Hence,
\begin{equation}
\abs{\mathrm{e}^{G_{2,y_0}(z)}}
\leq
\mathrm{e}^{\max\{0,\Delta_{R,y_0,\varepsilon}\}}
\end{equation}
in the finite-strip case, while
\begin{equation}
\abs{\mathrm{e}^{G_{2,\infty}(z)}}
\leq
\mathrm{e}^{\max\{0,\Delta_{R,\infty,\varepsilon}\}}
\end{equation}
in the half-plane case.
In the finite-strip case, for $z=x-\mathrm{i}y\in S_{y_0}$,
\begin{align}
\abs{\mathrm{e}^{G_{1,y_0}(x-\mathrm{i}y)}}
&=
\frac{\sin(\frac{\pi}{y_0})}
{2y_0\big(\cosh(\frac{\pi x}{y_0})
-\cos(\frac{\pi(y+1)}{y_0})\big)}
\nonumber\\
&\leq
\frac{\sin(\frac{\pi}{y_0})}
{2y_0\big(\cosh(\frac{\pi x}{y_0})
-\abs{\cos(\frac{\pi}{y_0})}\big)}.
\end{align}
Here we used
$\frac{\pi(y+1)}{y_0}\in
(\frac{\pi}{y_0},\pi+\frac{\pi}{y_0})$, on which the cosine is
bounded above by $\abs{\cos(\frac{\pi}{y_0})}$.
The final expression is integrable over $\R$ and independent of
$y\in(0,y_0)$. In the half-plane case,
\begin{equation}
\abs{\mathrm{e}^{G_{1,\infty}(x-\mathrm{i}y)}}
=
\frac{1}{\pi\big(x^2+(y+1)^2\big)}
\leq
\frac{1}{\pi(x^2+1)}.
\end{equation}
It follows directly from the definitions of $h_\theta$ that
$h_\theta\in H^1(S_{y_0})$ in both cases. Thus Step~4 produces
Hardy space functions with the boundary moduli identified in Step~2.

\paragraph{Step 5: normalize the candidate and recover the kernel function.}

Again, we consider the finite-strip and half-plane cases separately.
In each case, we determine the phase $\theta$ from the normalization
\eqref{eq:norm_h} and recover the kernel function using \eqref{eqn:kernel_f_h}.

\paragraph{Step 5a: finite-strip case, $y_0\in(1,\infty)$.}
By the definition of $m_{R,y_0}$ in \eqref{eq:mRy0} and the choice of
the principal branch,
\begin{align}
\log\!\Bigg(
\frac{\sinh\big(\frac{\pi}{2y_0}(-\mathrm{i}+R)\big)}
{\sinh\big(\frac{\pi}{2y_0}(-\mathrm{i}-R)\big)}
\Bigg)
&=
2\mathrm{i}\arctan\!\Bigg(
\frac{\tanh(\frac{\pi R}{2y_0})}
{\tan(\frac{\pi}{2y_0})}
\Bigg)
\nonumber\\
&=
\mathrm{i}\pi m_{R,y_0}.
\end{align}
It follows that
\begin{equation}
G_{2,y_0}(-\mathrm{i})
=
m_{R,y_0}\Delta_{R,y_0,\varepsilon}
=
\log\!\bigg(\frac{1-m_{R,y_0}}{\varepsilon}\bigg),
\end{equation}
and hence
\begin{equation}
\mathrm{e}^{C+G_{2,y_0}(-\mathrm{i})}=2\pi.
\end{equation}
Moreover,
\begin{equation}
\mathrm{e}^{G_{1,y_0}(-\mathrm{i})}
=
\frac{\sin(\frac{\pi}{y_0})}
{4y_0\sinh^2(-\frac{\pi\mathrm{i}}{y_0})}
=
-\frac{1}{4y_0\sin(\frac{\pi}{y_0})}.
\end{equation}
Therefore,
\begin{equation}
h_\theta(-\mathrm{i})
=
-\mathrm{e}^{\mathrm{i}\theta}
\frac{\pi}{2y_0\sin(\frac{\pi}{y_0})}.
\end{equation}
Comparison with the finite-strip normalization in \eqref{eq:norm_h}
gives $\mathrm{e}^{\mathrm{i}\theta}=-\mathrm{i}$. We denote the
resulting normalized function by $h_\star$. Substitution into
\eqref{eqn:kernel_f_h}, using \eqref{eq:Blaschke}, gives
\begin{equation}
f_\star(z)
=
\frac{1}{\sqrt{2\pi}}
\frac{2\pi\varepsilon}{1-m_{R,y_0}}
\frac{\sin(\frac{\pi}{y_0})}{4y_0}
\frac{
\exp\Bigg(
\dfrac{\Delta_{R,y_0,\varepsilon}}{\pi\mathrm{i}}
\log\dfrac{\sinh\big(\frac{\pi}{2y_0}(z+R)\big)}
{\sinh\big(\frac{\pi}{2y_0}(z-R)\big)}
\Bigg)}
{\sinh\big(\frac{\pi}{2y_0}(z-\mathrm{i})\big)
\sinh\big(\frac{\pi}{2y_0}(z+\mathrm{i})\big)}.
\end{equation}

\paragraph{Step 5b: half-plane case, $y_0=\infty$.}
In this case,
\begin{equation}
\log\!\bigg(\frac{-\mathrm{i}+R}{-\mathrm{i}-R}\bigg)
=
2\mathrm{i}\arctan(R)
=
\mathrm{i}\pi m_{R,\infty}.
\end{equation}
Consequently,
\begin{equation}
G_{2,\infty}(-\mathrm{i})
=
m_{R,\infty}\Delta_{R,\infty,\varepsilon}
=
\log\!\bigg(\frac{1-m_{R,\infty}}{\varepsilon}\bigg),
\qquad
\mathrm{e}^{C+G_{2,\infty}(-\mathrm{i})}=2\pi.
\end{equation}
Since
\begin{equation}
\mathrm{e}^{G_{1,\infty}(-\mathrm{i})}
=
\frac{1}{\pi(-2\mathrm{i})^2}
=
-\frac{1}{4\pi},
\end{equation}
we obtain
\begin{equation}
h_\theta(-\mathrm{i})
=
-\frac{\mathrm{e}^{\mathrm{i}\theta}}{2}.
\end{equation}
Comparison with the half-plane normalization in \eqref{eq:norm_h}
again gives $\mathrm{e}^{\mathrm{i}\theta}=-\mathrm{i}$. Denoting the
resulting normalized function by $h_\star$ and using
$B(z;\infty)=(z+\mathrm{i})/(z-\mathrm{i})$, we obtain
\begin{equation}
f_\star(z)
=
\frac{1}{\sqrt{2\pi}}
\frac{2\varepsilon}{1-m_{R,\infty}}
\frac{
\exp\bigg(
\dfrac{\Delta_{R,\infty,\varepsilon}}{\pi\mathrm{i}}
\log\dfrac{z+R}{z-R}
\bigg)}
{z^2+1}.
\end{equation}

In both cases, Step~4 gives $h_\star\in H^1(S_{y_0})$, while the
calculations above verify the normalization \eqref{eq:norm_h}.
Therefore, $f_\star\in\mc A(S_{y_0})$. Its boundary moduli are those
identified in Step~2, so Proposition~\ref{prop:reformulatedprob} gives $\mc E_{R,y_0}[f_\star]
=
\tilde{\mc E}_{R,y_0}[h_\star]
=
\varepsilon$. Thus $f_\star$ is feasible.

\paragraph{Step 6: evaluate the optimal value $\alpha_R[f_\star]$.}
Substituting the optimal boundary values gives, for almost every
$x\in[-R,R]$, in the finite-strip case,
\begin{equation}
\lvert f_{\star,0}(x)\rvert =
\frac{1}{\sqrt{2\pi}}\lvert h_{\star,0}(x)\rvert =
\frac{1}{\sqrt{2\pi}}
\frac{2\pi\varepsilon}{1-m_{R,y_0}}
\mathrm{e}^{\Delta_{R,y_0,\varepsilon}}P_0(x,1),
\end{equation}
while, in the half-plane case,
\begin{equation}
\lvert f_{\star,0}(x)\rvert=
\frac{1}{\sqrt{2\pi}}\lvert h_{\star,0}(x)\rvert=
\frac{1}{\sqrt{2\pi}}
\frac{2\pi\varepsilon}{1-m_{R,\infty}}
\ee^{\Delta_{R,\infty,\varepsilon}}\hat P_0(x,1).
\end{equation}
In either case, we obtain
\begin{equation}
    \alpha_R[f_\star]
= m_{R,y_0}\bigg(\frac{\varepsilon}{1-m_{R,y_0}}\bigg)^{-\frac{1-m_{R,y_0}}{m_{R,y_0}}}.
\end{equation}
The right-hand side is precisely the lower bound in
\eqref{eq:alpha_Jensen1}. Since $f_\star$ is feasible, it follows that
$f_\star$ attains the minimum. \\

Let $f$ be any minimizer and let $h$ be its associated Hardy space
function. By Lemma~\ref{lem:outer}, $h$ must be outer, since otherwise it
would admit a feasible outer competitor with strictly smaller objective
value. In particular, $h_\star$ is outer, since $f_\star$ is a minimizer.
Because $h$ attains equality in \eqref{eq:alpha_Jensen1}, the
equality characterization in Step~2 shows that $h$ and $h_\star$ have
the same boundary modulus almost everywhere.
Corollary~\ref{cor:outer_uniqueness} therefore gives
$h=h_\theta$ for some $\theta\in\R$, with $h_\theta$ as in Step~4.
Evaluating at $z=-\mathrm{i}$ and using the normalization
\eqref{eq:norm_h} as in Step~5 gives
$\mathrm{e}^{\mathrm{i}\theta}=-\mathrm{i}$. Hence $h=h_\star$, and
Proposition~\ref{prop:reformulatedprob} yields $f=f_\star$.
\end{proof}

\subsection{Optimal domain, truncation radius, and asymptotic cost}\label{sect:proof_col}

The proof begins by reducing the two-parameter optimization problem in terms of $(R,y_0)$ to a one-parameter problem in terms of only $R$ by comparing each finite-strip problem directly with the half-plane problem and showing that the global optimum occurs at $y_0=\infty$.
The reduced problem is then solved explicitly.  

\begin{proof}[Proof of Theorem~\ref{thm:optimal-kernel}]
For each fixed $R>0$ and $y_0\in(1,\infty]$, Proposition~\ref{prop:subnormalization-optimal-kernel}
gives the unique minimizer $f_\star(R,y_0)$ of $\alpha_R[f]$ under the constraint
$\mc{E}_{R,y_0}[f]\leq\varepsilon$. Therefore, writing
\begin{equation}
\alpha_\star(R,y_0)\coloneqq\alpha_R[f_\star(R,y_0)],
\end{equation}
Problem~\ref{prob:main} reduces to
\begin{equation}
\inf_{\substack{R>0,\ y_0\in(1,\infty]\\
f\in\mc{A}(S_{y_0}),\ \mc{E}_{R,y_0}[f]\leq\varepsilon}}
Q_R[f]
=
\inf_{\substack{R>0\\y_0\in(1,\infty]}}
R\alpha_\star(R,y_0).
\end{equation}
    
    \paragraph{Step 1: reduce to a one-parameter function.} We fix $y_0\in (1,\infty]$.
    From Proposition~\ref{prop:subnormalization-optimal-kernel}, we have
    \begin{equation}
        \alpha_\star (R,y_0) = A_\varepsilon(m_{R,y_0}),
        \qquad
        A_\varepsilon(m)\coloneqq
        m\left(\frac{\varepsilon}{1-m} \right)^{-(1-m)/m},
    \end{equation}
    where we recall
    \begin{equation}
        m_{R,y_0}\coloneqq \begin{dcases}
\displaystyle \int_{-R}^R P_0(x,1)\,\mathrm{d}x
=\frac{2}{\pi}\mathrm{arctan}\bigg(\frac{\tanh(\frac{\pi R}{2y_0})}{\tan(\frac{\pi}{2y_0})} \bigg), & y_0\in (1,\infty),\\[3mm]
\displaystyle \int_{-R}^R \hat P_0(x,1)\,\mathrm{d}x
=\frac{2}{\pi}\mathrm{arctan}(R), & y_0=\infty.
\end{dcases}
    \end{equation}
    with the explicit integral evaluation recorded in Lemma~\ref{lem:int_poisson}.
We now analyze the behaviour of $R\alpha_\star(R,y_0)$ at the endpoints.
Expanding the appropriate Poisson kernel about $x=0$ and using symmetry of the integral yields
\begin{equation}
m_{R,y_0}=c_{y_0}R+O(R^3) \quad (R\downarrow 0),
\qquad
c_{y_0}\coloneqq
\begin{dcases}
2P_0(0,1), & y_0\in(1,\infty),\\
2\hat P_0(0,1), & y_0=\infty,
\end{dcases}
\end{equation}
where $c_{y_0}>0$. Hence
\begin{equation}
\log\alpha_\star(R,y_0)
=
\frac{1}{c_{y_0}R}\log\frac{1}{\varepsilon}
+\log R+O(1)
\quad (R\downarrow0).
\end{equation}
Since $\log(1/\varepsilon)>0$, it follows that
$\lim_{R\downarrow 0} R\alpha_\star(R,y_0)=\infty$.
For the other endpoint, Lemma~\ref{lem:int_poisson} gives
\begin{equation}
\lim_{R\to\infty}m_{R,y_0}
=
\begin{dcases}
1-\dfrac{1}{y_0}, & y_0\in(1,\infty),\\
1, & y_0=\infty.
\end{dcases}
\end{equation}
Since $A_\varepsilon$ is positive and continuous on $(0,1)$ and
$\lim_{m\uparrow1}A_\varepsilon(m)=1$, the quantity
$\alpha_\star(R,y_0)$ has a finite positive limit as $R\to\infty$.
It follows that $\lim_{R\to \infty} R\alpha_\star (R,y_0)=\infty$.

Because $R \mapsto R\alpha_\star(R,y_0)$ is continuous on $(0,\infty)$ and diverges at both endpoints, there exists $R_\star(y_0)\in(0,\infty)$ such that
\begin{equation}
    R_\star(y_0)\alpha_\star(R_\star(y_0),y_0)
= \inf_{R>0} R\alpha_\star(R,y_0) := \beta(y_0).
\end{equation}

We now compare every finite-strip problem directly with the half-plane
problem. Fix $y_0\in(1,\infty)$ and set $a_{y_0}\coloneqq\pi/(2y_0)$.
The map $R\mapsto m_{R,y_0}$ is strictly increasing from $(0,\infty)$
onto $(0,1-1/y_0)$, and its inverse is
\begin{equation}
R_{y_0}(m)
\coloneqq
\frac{2y_0}{\pi}
\operatorname{artanh}\!\left(
\tan\!\left(\frac{\pi}{2y_0}\right)
\tan\!\left(\frac{\pi m}{2}\right)
\right),
\qquad 0<m<1-\frac{1}{y_0}.
\end{equation}
For the half-plane, the corresponding inverse is
\begin{equation}
R_\infty(m)\coloneqq\tan\!\left(\frac{\pi m}{2}\right),
\qquad 0<m<1.
\end{equation}
If $0<m<1-1/y_0$, then the argument of $\operatorname{artanh}$ above
belongs to $(0,1)$. Using $\operatorname{artanh}(u)>u$ for $u\in(0,1)$
and $\tan(a_{y_0})>a_{y_0}$ gives
\begin{equation}
R_{y_0}(m)
>
\frac{1}{a_{y_0}}
\tan(a_{y_0})\tan\!\left(\frac{\pi m}{2}\right)
>
R_\infty(m).
\end{equation}

Let $R_\star(y_0)$ be a minimizer for the finite-strip problem and set
$m_\star(y_0)\coloneqq m_{R_\star(y_0),y_0}$. Since $A_\varepsilon$
depends only on $m$, the preceding strict inequality yields
\begin{align}
\beta(y_0)
&=R_{y_0}(m_\star(y_0))A_\varepsilon(m_\star(y_0))\nonumber\\
&>R_\infty(m_\star(y_0))A_\varepsilon(m_\star(y_0))
\geq\beta(\infty).
\end{align}
Thus every finite strip has strictly larger optimal cost than the
half-plane. Consequently, the unique optimal strip width is
$y_0=\infty$, and Problem~\ref{prob:main} reduces to minimizing
$R\alpha_\star(R,\infty)$ over $R>0$.

\paragraph{Step 2: solve the reduced optimization problem.}
The remaining problem is to minimize $R\alpha_\star(R,\infty)$ over $R>0$.
We do this by first deriving a stationarity condition for $R$ and then showing that this condition has a unique solution. 
Since the objective diverges as $R\downarrow 0$
and $R\rightarrow\infty$ this solution corresponds to the unique minimizer. \\

Differentiating and using the chain rule gives
\begin{align}
\frac{\partial}{\partial R}\big(R\alpha_\star(R,\infty)\big)
=&\; \alpha_\star(R,\infty)
+R\frac{\partial\alpha_\star(R,\infty)}{\partial m_{R,\infty}}
\frac{\partial m_{R,\infty}}{\partial R} \notag\\
=&\; \alpha_\star(R,\infty)\bigg(
1-\frac{R}{m_{R,\infty}^2}
\log\bigg(\frac{1-m_{R,\infty}}{\varepsilon}\bigg)
\frac{2}{\pi(1+R^2)}
\bigg),
\end{align}
where
\begin{equation}
\frac{\partial\alpha_\star(R,\infty)}{\partial m_{R,\infty}}
= \alpha_\star(R,\infty)\frac{\partial\log \alpha_\star(R,\infty)}{\partial m_{R,\infty}} = 
-\frac{\alpha_\star(R,\infty)}{m_{R,\infty}^2}
\log\bigg(\frac{1-m_{R,\infty}}{\varepsilon}\bigg)
\end{equation}
and $m_{R,\infty}=\frac{2}{\pi}\arctan(R)$ so 
\begin{equation}
\frac{\partial m_{R,\infty}}{\partial R}
=
\frac{2}{\pi(1+R^2)}.
\end{equation}

Setting $\frac{\partial}{\partial R}[R\alpha_\star(R,\infty)] = 0$ yields the first-order optimality condition for $R$, which, using the relation between $R$ and $m_{R,\infty}$, can be written as
\begin{equation}
        \log \bigg(\frac{1-m_{R,\infty}}{\varepsilon} \bigg) = \frac{\pi m_{R,\infty}^2 (1+R^2)}{2R} = \frac{\pi m_{R,\infty}^2 }{\sin 2 \arctan (R)}= \frac{\pi m_{R,\infty}^2 }{\sin(\pi m_{R,\infty})}.\label{eqn:Rstationary_cond}
\end{equation}

Observe that the map $R\mapsto m_{R,\infty}=\frac{2}{\pi}\arctan(R)$ is continuous and strictly increasing from $(0,\infty)$ onto $(0,1)$.
So a unique solution $R>0$ is implied by the stationary condition having a unique solution $m_{R,\infty}\in(0,1)$.
Rewrite \eqref{eqn:Rstationary_cond} as 
\begin{equation}
\Lambda(m_{R,\infty},\varepsilon)\coloneqq \log\bigg(\frac{1-m_{R,\infty}}{\varepsilon}\bigg)-\frac{\pi m_{R,\infty}^2}{\sin(\pi m_{R,\infty})} = 0
\end{equation}
We have $\lim_{m_{R,\infty}\downarrow 0}\Lambda(m_{R,\infty},\varepsilon)=\log\frac{1}{\varepsilon}$ and $\lim_{m_{R,\infty}\uparrow 1}\Lambda(m_{R,\infty},\varepsilon)=-\infty$.
Since $\Lambda(\ccdot,\varepsilon)$ is continuous on $(0,1)$, there exists a solution $m_{R,\infty}(\varepsilon)\in(0,1)$. \\

Uniqueness will follow from the fact that $\Lambda(\ccdot,\varepsilon)$ is strictly decreasing.
Differentiating,
\begin{equation}
\frac{\partial}{\partial m}\Lambda(m,\varepsilon)
=
-\frac{1}{1-m}
-
\frac{2\pi m\sin(\pi m)-\pi^2 m^2\cos(\pi m)}{\sin^2(\pi m)} < 0
\end{equation}
for $m_{R,\infty}\in(0,1)$ by the inequality $2\sin(\pi x)>\pi x\cos(\pi x)$.
Hence there is a unique solution $m_\star:=m_{R_\star,\infty}$ to \eqref{eqn:Rstationary_cond} corresponding to a unique optimal $R_\star$. \\

Recalling the form of the optimal kernel function from Proposition~\ref{prop:subnormalization-optimal-kernel} we can now express the unique solution using $y_0=\infty$ and $R_\star$, $m_\star$ obey the stationary condition \eqref{eqn:Rstationary_cond}.
First we compute $\Delta_{R_\star,\infty,\varepsilon}$,
\begin{equation}       \Delta_{R_\star,\infty,\varepsilon} = \frac{1}{m_\star}\log\bigg(\frac{1-m_\star}{\varepsilon}\bigg)=\frac{\pi m_\star}{\sin(\pi m_\star)}:= \Delta_\star.
\end{equation}
We can also rewrite the prefactor in $f_\star$ as
\begin{equation}
    \frac{\varepsilon}{1-m_\star} = \mathrm{e}^{-m_\star \Delta_\star}
\end{equation}
hence the unique solution is given by
\begin{equation}
    f_\star(z) = \frac{2}{\sqrt{2\pi}}\mathrm{e}^{-m_\star \Delta_\star}\dfrac{\exp\bigg(\dfrac{\Delta_\star}{\pi \mathrm{i}}\log\dfrac{z+R_\star}{z-R_\star}\bigg)}{z^2 +1}
\end{equation}

\paragraph{Step 3: optimal continuous LCHS cost.}
We now use the optimal kernel function found to compute the minimum continuous LCHS cost and perform asymptotic analysis as $\varepsilon \downarrow 0$. \\

At the solution, the cost function is given by
    \begin{align}
        R_\star(\varepsilon) \alpha_{R_\star(\varepsilon)}[f_\star] =&\; R_\star(\varepsilon) m_\star(\varepsilon) \left( \frac{\varepsilon}{1-m_\star(\varepsilon)}\right)^{-(1-m_\star(\varepsilon))/m_\star(\varepsilon)} \nonumber\\
        =&\; R_\star(\varepsilon) m_\star(\varepsilon) \exp \left((1-m_\star(\varepsilon))\Delta_\star(\varepsilon) \right) \nonumber \\
        =&\; R_\star(\varepsilon) m_\star(\varepsilon) \exp \left(\frac{(1-m_\star(\varepsilon))\pi m_\star(\varepsilon)}{\sin(\pi m_\star(\varepsilon))}\right).
    \end{align}

We first show that $m_\star(\varepsilon)\to1$ as $\varepsilon\downarrow0$.
Otherwise, there would exist $\delta>0$ and a sequence
$\varepsilon_n\downarrow0$ such that $m_\star(\varepsilon_n)\leq1-\delta$.
Along this sequence, the left-hand side of the stationary condition satisfies
\begin{equation}
\log\!\left(\frac{1-m_\star(\varepsilon_n)}{\varepsilon_n}\right)
\geq
\log\!\left(\frac{\delta}{\varepsilon_n}\right)
\longrightarrow\infty,
\end{equation}
whereas its right-hand side remains bounded, a contradiction.
We therefore define
\begin{equation}
\overline{m}_\star:=1-m_\star(\varepsilon)
\end{equation}
so that $\overline{m}_\star\to0$ as $\varepsilon\downarrow0$. \\

The stationary condition becomes
\begin{equation}
\log\bigg(\frac{\overline{m}_\star}{\varepsilon}\bigg)
=
\frac{\pi(1-\overline{m}_\star)^2}{\sin(\pi \overline{m}_\star)}.
\end{equation}
Using $\sin(\pi x)=\pi x+O(x^3)$ as $x\downarrow 0$, we obtain
\begin{equation}
\frac{\pi(1-\overline{m}_\star)^2}{\sin(\pi \overline{m}_\star)}
=
\overline{m}_\star^{-1}-2+O(\overline{m}_\star).
\end{equation}
Multiplying the stationary condition by $\overline{m}_\star$ and
using this expansion gives
\begin{equation}
\overline{m}_\star\log\frac{1}{\varepsilon}
=1-\overline{m}_\star\log\overline{m}_\star
-2\overline{m}_\star+O(\overline{m}_\star^2)
\longrightarrow1.
\end{equation}
Hence,
\begin{equation}
\log\overline{m}_\star
=-\log\log\frac{1}{\varepsilon}+o(1).
\end{equation}
Substitution back into the stationary condition yields
\begin{equation}\label{eq:mbarscaling}
\overline{m}_\star^{-1}
=\log\frac{1}{\varepsilon}-\log\log\frac{1}{\varepsilon}+2+o(1)
\qquad (\varepsilon\downarrow0).
\end{equation}

Writing $\alpha_{R_\star(\varepsilon)}[f_\star]$ in terms of $\overline{m}_\star$:
\begin{align}\label{eq:alphascaling}
\log{\alpha_{R_\star(\varepsilon)}[f_\star]}=&\; \log (1-\overline{m}_\star)+\frac{\overline{m}_\star}{1-\overline{m}_\star}\log\frac{\overline{m}_\star}{\varepsilon} \nonumber\\
=&\; -\overline{m}_\star+O(\overline{m}_\star^2)+\overline{m}_\star(1+\overline{m}_\star)\bigg(\frac{1}{\overline{m}_\star}-2 \bigg) + O(\overline{m}_\star^2) \nonumber\\
=&\; 1-2\overline{m}_\star +O(\overline{m}_\star^2),
\end{align}
so that
\begin{equation}
{\alpha_{R_\star(\varepsilon)}[f_\star]}
=\mathrm{e}\big(1-2\overline{m}_\star+O(\overline{m}_\star^2)\big).
\label{eq:alphastarmbar}
\end{equation}
The optimal truncation radius is
\begin{equation}
\label{eqn:Rstarmbar}
        R_\star(\varepsilon) = \tan\left(\frac{\pi}{2}m_\star(\varepsilon) \right)= \cot \left(\frac{\pi}{2}\overline{m}_\star \right) = \frac{2}{\pi \overline{m}_\star} - \frac{\pi \overline{m}_\star}{6} + O(\overline{m}_\star^3).
\end{equation}
Now, combining \eqref{eq:alphastarmbar} and \eqref{eqn:Rstarmbar} gives
    \begin{align}
R_\star(\varepsilon){\alpha_{R_\star(\varepsilon)}[f_\star]}
        =&\; \frac{2\mathrm{e}}{\pi}
        {\left(\overline{m}_\star^{-1}-2+O(\overline{m}_\star)\right)} \nonumber\\
        =&\; \frac{2\mathrm{e}}{\pi}
        \left(\log\frac{1}{\varepsilon}
        -\log\log\frac{1}{\varepsilon}+o(1)\right),
        \qquad (\varepsilon\downarrow0),
    \end{align}
as desired.
\end{proof}

\begin{proof}[Proof of Corollary~\ref{cor:full_normalization_cost}]
For fixed $R$ and $y_0$, the positive factor $\sqrt{1+R^2}$ does not
change the unique minimizer of $\alpha_R[f]$ given by
Proposition~\ref{prop:subnormalization-optimal-kernel}.
Using the notation of the preceding proof, the reduced cost is
$\sqrt{1+R^2}\,A_\varepsilon(m_{R,y_0})$.
At fixed $m$, the inequality $R_{y_0}(m)>R_\infty(m)$ and the strict
increase of $R\mapsto\sqrt{1+R^2}$ show that every finite-strip
candidate is strictly more costly than the corresponding half-plane
candidate. Thus it suffices to minimize
\begin{equation}
\frac{A_\varepsilon(m)}{\cos(\pi m/2)},\qquad 0<m<1.
\end{equation}
Its logarithmic derivative is
\begin{equation}
\frac{\mathrm d}{\mathrm dm}
\log\frac{A_\varepsilon(m)}{\cos(\pi m/2)}
=\frac{\pi}{2}\tan\frac{\pi m}{2}
-\frac{1}{m^2}\log\frac{1-m}{\varepsilon}.
\end{equation}
In \eqref{eq:full_normalization_stationary}, the left-hand side is
strictly decreasing and the right-hand side is strictly increasing.
Their difference tends to $\log(1/\varepsilon)>0$ as $m\downarrow0$
and to $-\infty$ as $m\uparrow1$.
Consequently, there is a unique root, with the cost decreasing before
it and increasing afterwards. This proves global optimality and
uniqueness, including that of the kernel function.
At this root, writing $m=\widetilde m_\star(\varepsilon)$ and
$R=\widetilde R_\star(\varepsilon)$, we have
\begin{equation}
\Delta_{R,\infty,\varepsilon}
=\frac{1}{m}\log\frac{1-m}{\varepsilon}
=\frac{\pi mR}{2}.
\end{equation}
Substitution into \eqref{eq:fopt} gives \eqref{eq:full_normalization_kernel}.
Since
\begin{equation}
\frac{\pi m^2}{\sin(\pi m)}
-\frac{\pi m^2}{2}\tan\frac{\pi m}{2}
=\frac{\pi m^2}{2}\cot\frac{\pi m}{2}>0,
\end{equation}
comparison with \eqref{eqn:Rstationary_cond} gives
$\widetilde m_\star>m_\star$ and hence
$\widetilde R_\star>R_\star$. \\

Finally, $\widetilde Q_R[f]\geq Q_R[f]$ for every feasible choice.
Evaluating $\widetilde Q_R$ at the optimizer from
Theorem~\ref{thm:optimal-kernel} therefore gives
\begin{equation}
0\leq\widetilde Q_\star(\varepsilon)-Q_\star(\varepsilon)
\leq\frac{{\alpha_{R_\star(\varepsilon)}[f_\star]}}
{\sqrt{1+R_\star(\varepsilon)^2}+R_\star(\varepsilon)}.
\end{equation}
By \eqref{eq:alphastarmbar}, ${\alpha_{R_\star(\varepsilon)}[f_\star]}\to\ee$, while
\eqref{eq:mbarscaling} and \eqref{eqn:Rstarmbar} give
$R_\star(\varepsilon)\sim(2/\pi)\log(1/\varepsilon)$.
This proves \eqref{eq:full_normalization_asymptotics} and the claimed
asymptotic expansion.
\end{proof}

\section{Spectrally convergent quadrature}\label{subsection:quadproof}

To prove the spectrally convergent quadrature in Theorem~\ref{thm:mainquadrature}, we apply a double-exponential change of variables to transform the integral defining $\m{O}_f(t;R)$ in \eqref{eq:Of}, with $f=f_\star$ and $R=R_\star(\varepsilon)$, from the finite interval $(-R_\star(\varepsilon),R_\star(\varepsilon))$ to the real line.
The transformed integrand extends analytically to a symmetric strip about the real axis and decays rapidly in every sufficiently narrow strip. This will allow us to apply the exponentially convergent trapezoidal rule.
We then prove a discretization error bound in Lemma~\ref{lem:quadrature} and a truncation error bound in Lemma~\ref{lem:truncation}.
The quadrature theorem follows immediately from these two bounds.

\subsection{Double-exponential transformation and analytic bounds}
\label{subsec:quadrature_transformation}

We use the map $\phi_R$ defined in \eqref{eq:phiR}, which is designed to resolve the endpoint singularities at $\pm R$ and induce rapid decay of the transformed integrand.
This map can be factored as
\begin{equation}\label{eq:phi1phi2}
\phi_R = \phi_1\circ \phi_2, \qquad
\phi_1(w)\coloneqq R\tanh\!\left(\frac{w}{2}\right), \qquad
\phi_2(\zeta)\coloneqq \pi\sinh(\zeta).
\end{equation}
To analyze this map, we first define the symmetric strip
\begin{equation}
\label{eq:strip_symmetric}
    \Sigma_\eta\coloneqq \{ z\in\mathbb{C} : |\operatorname{Im} z|< \eta\}  \qquad (\eta \in (0,\pi)).
\end{equation}
The transformed kernel function has desirable properties on the strip $\Sigma_{\eta_0}$, where the strip width is defined by \eqref{eq:quadrature_strip_width}. \\

In particular, the transformed kernel function admits a holomorphic continuation in $\zeta$ throughout this strip.
Near the real axis, it agrees with $f_\star(\phi_R(\zeta))\phi_R'(\zeta)/\sqrt{2\pi}$, using the branch prescribed in Theorem~\ref{thm:optimal-kernel}. We define the continuation away from the real axis by the explicit expression in Lemma~\ref{lem:phi}~\ref{no:transformedkernel}, which is analytic throughout $\Sigma_{\eta_0}$, even at points where $\phi_R(\zeta)$ lies on an original branch cut.
Thus we do not require the image $\phi_R(\Sigma_{\eta_0})$ to avoid these branch cuts.\\

We now establish the key analytic properties of the map $\phi_R$.

\begin{lem}[Properties of $\phi_R$]\label{lem:phi}
Let $m$, $R$, and $\Delta$ be as in \eqref{eq:quadrature_parameters}, $\eta_0$ as in \eqref{eq:quadrature_strip_width} and let $f_\star$ be the optimal kernel function from Theorem~\ref{thm:optimal-kernel}.
For $\eta\in(0,\eta_0)$, the map $\phi_R$ has the following properties:
\begin{enumerate}[label=(\roman*)]
\item \emph{Analytic structure and image of the real axis.}\label{no:analyticstructure}\label{no:realimage}
The map $\phi_R$ is holomorphic on $\Sigma_{\pi/2}$ and takes the real axis bijectively onto $(-R,R)$, with
\begin{equation}
\phi_R'(\xi)>0 \qquad (\xi\in\R).
\end{equation}
\item \emph{Image of the symmetric strip.}\label{no:symmetricimage}
For all $\zeta\in\overline{\Sigma}_\eta$,
\begin{equation}\label{eq:phi_strip_image}
\abs{\operatorname{Im}\phi_R(\zeta)}
\leq
\frac{R\pi\sin\eta}{2-\frac{\pi^2}{2}\sin^2\eta}.
\end{equation}
\item \emph{Analyticity of the transformed optimal kernel function.}\label{no:transformedkernel}
The function
\begin{equation}\label{eq:transformed_kernel_real}
F_\star(\xi)\coloneqq
\frac{1}{\sqrt{2\pi}}f_{\star,0}(\phi_R(\xi))\phi_R'(\xi)
\qquad (\xi\in\R),
\end{equation}
extends holomorphically to $\Sigma_{\eta_0}$, where it is given by
\begin{equation}\label{eqn:feval}
F_\star(\zeta)=
\frac{R\ee^{(1-m)\Delta}}{2}
\frac{\cosh\zeta\,\ee^{-\ii\Delta\sinh\zeta}}
{\cosh^2\!\left(\frac{\pi}{2}\sinh\zeta\right)
+R^2\sinh^2\!\left(\frac{\pi}{2}\sinh\zeta\right)}.
\end{equation}
\item \emph{Bounds on the symmetric strip.}\label{no:boundary}
The function $F_\star$ tends to zero uniformly as $\abs{\operatorname{Re}\zeta}\to\infty$ in $\overline{\Sigma}_\eta$, and
\begin{equation}\label{eq:transformed_kernel_strip_integral}
\sup_{\abs{y} <\eta}
\int_\R\abs{F_\star(x+\ii y)}\,\mathrm{d}x
\leq
\frac{4R\ee^{(1-m)\Delta+\Delta\sin\eta}}
{\big(1-\frac{\pi^2}{4}(1+R^2)\sin^2\eta\big)
(\pi\cos\eta-\Delta\sin\eta)}.
\end{equation}
\item \emph{Real-axis modulus.}\label{no:realaxis}
On the real axis,
\begin{equation}\label{eq:transformed_kernel_modulus}
\abs{F_\star(\xi)}
=
\frac{\ee^{(1-m)\Delta}}{\pi}
\frac{\phi_R'(\xi)}{1+\phi_R(\xi)^2}.
\end{equation}
The function $\xi\mapsto \abs{F_\star(\xi)}$ is strictly positive, even, and strictly decreasing on $(0,\infty)$.
\end{enumerate}
\end{lem}

\begin{proof}
We will prove the map properties in turn.
\paragraph{Analytic structure and image of the real axis \ref{no:analyticstructure}.}
We use the factorization $\phi_R=\phi_1\circ\phi_2$ from \eqref{eq:phi1phi2}.
The map $\phi_2=\pi\sinh$ is entire, and the poles of $\phi_1$ occur at $w=(2n+1)\ii\pi$, $n\in\Z$.
If $\zeta=x+\ii y\in\Sigma_{\pi/2}$ satisfied $\pi\sinh\zeta=(2n+1)\ii\pi$, then
\begin{equation}
\sinh x\cos y=0,\qquad \cosh x\sin y=2n+1.
\end{equation}
Since $\cos y>0$, the first equality gives $x=0$, while the second would give $\sin y=2n+1$, which is impossible for $\abs{y}<\pi/2$.
Thus $\phi_R$ is holomorphic on $\Sigma_{\pi/2}$. \\

Since $\sinh\colon\R\to\R$ is a bijection, we have
\begin{equation}\label{eqn:phi2realtoreal}
\phi_2(\R)=\pi\sinh(\R)=\R.
\end{equation}
Moreover, $\tanh\colon\R\to(-1,1)$ is a bijection, and hence
\begin{equation}
\phi_R(\R)=\phi_1(\phi_2(\R))=(-R,R).
\end{equation}
By the chain rule,
\begin{equation}\label{eq:phi_strip_derivative}
\phi_R'(\zeta)
=\frac{R\pi}{2}
\frac{\cosh\zeta}{\cosh^2\!\left(\frac{\pi}{2}\sinh\zeta\right)}
\qquad (\zeta\in\Sigma_{\pi/2}),
\end{equation}
which is positive for $\zeta\in\R$.

\paragraph{Image of the symmetric strip \ref{no:symmetricimage}}
Let $\zeta=x+\ii y\in\overline{\Sigma}_\eta$, recall that $\eta<\eta_0\leq\pi/6$, and write
\begin{equation}
w=\phi_2(\zeta)=u+\ii v,\qquad
u=\pi\sinh x\cos y,\qquad v=\pi\cosh x\sin y.
\end{equation}
Then
\begin{equation}
\abs{\operatorname{Im}\phi_R(\zeta)}
=R\frac{\abs{\sin v}}{\cosh u+\cos v}.
\end{equation}
Using $\cosh u\geq1+u^2/2$ and $\cos v\geq1-v^2/2$, we obtain
\begin{align}
\cosh u+\cos v
\geq&\; 2+\frac{u^2-v^2}{2} \nonumber\\
=&\; 2-\frac{\pi^2}{2}\sin^2y
+\frac{\pi^2}{2}\cos(2y)\sinh^2x\nonumber\\
\geq &\ 2-\frac{\pi^2}{2}\sin^2\eta
+\frac{\pi^2}{2}\cos(2\eta)\sinh^2x\nonumber\\
\geq&\; 
\left(2-\frac{\pi^2}{2}\sin^2\eta\right)\cosh x.
\end{align}
For the last inequality, we used $\cosh x\leq1+\frac12\sinh^2x$ and
\begin{equation}
\frac{\pi^2}{2}\cos(2\eta)
\geq\frac{\pi^2}{4}
>1
\geq\frac12\left(2-\frac{\pi^2}{2}\sin^2\eta\right).
\end{equation}
The factor $2-\frac{\pi^2}{2}\sin^2\eta$ is positive.
Together with $\abs{\sin v}\leq\abs{v}\leq\pi\cosh x\sin\eta$, this proves \eqref{eq:phi_strip_image}.

\paragraph{Analyticity of the transformed kernel function \ref{no:transformedkernel}}
We again consider the factorization $f_\star\circ\phi_R=f_\star\circ\phi_1\circ\phi_2$.
Since
\begin{equation}
\frac{\phi_1(w)+R}{\phi_1(w)-R}=-\ee^w
\qquad (w\in\Sigma_\pi),
\end{equation}
the branch of the logarithm prescribed in Theorem~\ref{thm:optimal-kernel} satisfies, near the real axis,
\begin{equation}\label{eq:pullback_logarithm}
\log\frac{\phi_1(w)+R}{\phi_1(w)-R}=w+\ii\pi.
\end{equation}
Indeed, for $-\pi<\operatorname{Im}w<0$, the point $\phi_1(w)$ lies in the lower half-plane and $w+\ii\pi$ is the principal logarithm of $-\ee^w$.
The identity then extends through the real axis by analytic continuation.
Since $m$ and $\Delta$ are the optimal parameters determined by $\varepsilon$, we have $\varepsilon=(1-m)\ee^{-m\Delta}$ and near the real axis, obtain
\begin{equation}\label{eq:kernel_pullback_branch}
f_\star(\phi_1(w))
=\frac{2}{\sqrt{2\pi}}\ee^{(1-m)\Delta}
\frac{\exp\bigg(\dfrac{\Delta}{\pi\ii}w\bigg)}
{1+R^2\tanh^2(w/2)}.
\end{equation}
Substituting $w=\pi\sinh\zeta$ and multiplying by $\phi_R'(\zeta)/\sqrt{2\pi}$ using \eqref{eq:phi_strip_derivative} gives \eqref{eqn:feval}. \\

By Property~\ref{no:analyticstructure}, $\cosh(\frac{\pi}{2}\sinh\zeta)$ does not vanish in $\Sigma_{\pi/2}$.
Consequently, the denominator in \eqref{eqn:feval} vanishes there precisely when
\begin{equation}
R\tanh\!\left(\frac{\pi}{2}\sinh\zeta\right)=\pm\ii,
\end{equation}
or, equivalently,
\begin{equation}
\pi\sinh\zeta
=\pm2\ii\arctan(1/R)+2\pi\ii n
\qquad (n\in\Z).
\end{equation}
Writing $\zeta=x+\ii y\in\Sigma_{\pi/2}$, the vanishing real part again forces $x=0$.
Since $\pi\sin y\in(-\pi,\pi)$ and $2\arctan(1/R)\in(0,\pi)$, only $n=0$ is possible, giving
\begin{equation}
\zeta=\pm\ii\arcsin\left(\frac{2}{\pi}\arctan\frac1R\right)
=\pm\ii\eta_R.
\end{equation}
The denominator therefore has no zeros in $\Sigma_{\eta_0}$, and \eqref{eqn:feval} defines the claimed holomorphic continuation.

\paragraph{Bounds on the symmetric strip \ref{no:boundary}}
Let $\zeta=x+\ii y\in\overline{\Sigma}_\eta$ and assume \eqref{eq:strip_bound_conditions} holds.
Write
\begin{equation}
\frac{\pi}{2}\sinh\zeta=u+\ii v,\qquad
u=\frac{\pi}{2}\sinh x\cos y,\qquad
v=\frac{\pi}{2}\cosh x\sin y.
\end{equation}
A direct calculation gives
\begin{equation}
\begin{aligned}
&\abs{\cosh^2(u+\ii v)+R^2\sinh^2(u+\ii v)}^2\\
&\quad=
\left((1+R^2)\sinh^2u+1-(1+R^2)\sin^2v\right)^2
+4R^2(1+R^2)\sinh^2u\sin^2v.
\end{aligned}
\end{equation}
In particular,
\begin{equation}
\abs{\cosh^2(u+\ii v)+R^2\sinh^2(u+\ii v)}
\geq (1+R^2)\sinh^2u+1-(1+R^2)\sin^2v.
\end{equation}
Since $\frac{\pi}{2}\cos\eta>1$, we have
\begin{equation}
\abs{u}\geq\frac{\pi}{2}\cos\eta\,\abs{\sinh x}\geq\abs{x},
\end{equation}
and hence
\begin{equation}
\sin^2v\leq v^2
\leq\frac{\pi^2}{4}\sin^2\eta\,\cosh^2x
\leq\frac{\pi^2}{4}\sin^2\eta\,\cosh^2u.
\end{equation}
It follows that
\begin{align}
\abs{\cosh^2(u+\ii v)+R^2\sinh^2(u+\ii v)}\geq &\; 
\left(1-\frac{\pi^2}{4}(1+R^2)\sin^2\eta\right)\cosh^2u
+R^2\sinh^2u \nonumber\\
\geq &\; 
\frac14\left(1-\frac{\pi^2}{4}(1+R^2)\sin^2\eta\right)
\ee^{\pi\cos\eta\,\abs{\sinh x}}.\label{eq:transformed_kernel_denominator_bound}
\end{align}
Here the coefficient is positive by \eqref{eq:strip_bound_conditions}.
Moreover,
\begin{equation}
\abs{\cosh(x+\ii y)}\leq\cosh x,
\qquad
\abs{\ee^{-\ii\Delta\sinh(x+\ii y)}}
=\ee^{\Delta\cosh x\sin y}
\leq\ee^{\Delta\cosh x\sin\eta}.
\end{equation}
Combining these inequalities with \eqref{eqn:feval} and using $\cosh x\leq1+\abs{\sinh x}$ gives
\begin{equation}\label{eq:transformed_kernel_decay}
\abs{F_\star(x+\ii y)}
\leq
\frac{2R\ee^{(1-m)\Delta+\Delta\sin\eta}}
{1-\frac{\pi^2}{4}(1+R^2)\sin^2\eta}
\cosh x\,
\ee^{-(\pi\cos\eta-\Delta\sin\eta)\abs{\sinh x}}.
\end{equation}
Since $\pi\cos\eta-\Delta\sin\eta>0$, the right-hand side tends to zero as $\abs{x}\to\infty$, uniformly for $\abs{y}\leq\eta$.
Finally, the substitution $t=\sinh x$ gives
\begin{equation}
\int_\R\cosh x\,
\ee^{-(\pi\cos\eta-\Delta\sin\eta)\abs{\sinh x}}\,\mathrm{d}x
=\frac{2}{\pi\cos\eta-\Delta\sin\eta},
\end{equation}
which proves \eqref{eq:transformed_kernel_strip_integral}.

\paragraph{Real-axis modulus \ref{no:realaxis}}
On the real axis,
\begin{equation}
\abs{\ee^{-\ii\Delta\sinh\xi}}=1.
\end{equation}
Equations \eqref{eqn:feval} and \eqref{eq:phi_strip_derivative} therefore give the claimed expression \eqref{eq:transformed_kernel_modulus}.
This function is positive by Property~\ref{no:realimage}, and even because $\phi_R$ is odd and $\phi_R'$ is even. \\

To prove strict monotonicity, put $t=\sinh\xi>0$.
Then
\begin{equation}
\abs{F_\star(\operatorname{arcsinh}t)}
=\frac{R\ee^{(1-m)\Delta}}{2}
\frac{\sqrt{1+t^2}}
{1+(1+R^2)\sinh^2\!\left(\frac{\pi}{2}t\right)},
\end{equation}
and
\begin{align}
\frac{\mathrm{d}}{\mathrm{d}t}
\log\abs{F_\star(\operatorname{arcsinh}t)}
=&\; \frac{t}{1+t^2}
-\frac{\pi(1+R^2)\sinh\!\big(\frac{\pi}{2}t\big)
\cosh\!\left(\frac{\pi}{2}t\right)}
{1+(1+R^2)\sinh^2\!\big(\frac{\pi}{2}t\big)} \nonumber\\
\leq &\;\frac{t}{1+t^2}
-\pi\tanh\!\big(\tfrac{\pi}{2}t\big) \nonumber\\
\leq &\; \frac{t}{1+t^2}-\frac{\pi t}{\sqrt{1+t^2}}<0.
\end{align}
The first inequality uses
$1+(1+R^2)\sinh^2(\frac{\pi}{2}t)\leq(1+R^2)\cosh^2(\frac{\pi}{2}t)$.
For the second, we used $\frac{\pi}{2}t\geq\operatorname{arcsinh}t$ and the monotonicity of $\tanh$.
Since $\sinh$ is strictly increasing on $\R$, the real-axis modulus is strictly decreasing for $\xi>0$.
\end{proof}

\subsection{Discretization and truncation errors}
\label{subsec:quadrature_errors}

We now establish the discretization and truncation error bounds.
We first need the following lemma, established in
\cite[Lems.~3--4]{low2025}.

\begin{lem}[Properties of $\mc{U}(z;t)$]\label{lem:propertiesofU}
Let $\tau>0$. Suppose that $\m{L}$ and $\m{H}$ satisfy the matrix hypotheses of Theorem~\ref{thm:LCHS}. Then, for every $t\in[0,\tau]$, the time-ordered exponential $\mc{U}(\ccdot;t)$, defined by \eqref{eq:UE}, is entire.
Moreover, if $\m{L}(t)\succeq0$ for almost every $t\in[0,\tau]$, then
\begin{equation}
\norm{\mc{U}(z;t)}_2
\leq
\begin{dcases}
\exp\!\left(\abs{\operatorname{Im}z}
\int_0^t\norm{\m{L}(s)}_2\,\mathrm{d}s\right),
&\operatorname{Im}z\geq0,\\
1,&\operatorname{Im}z<0,
\end{dcases}
\qquad (z\in\C,\ t\in[0,\tau]).
\end{equation}
\end{lem}

Recall the continuous integral from Theorem~\ref{thm:LCHS}, with $f=f_\star$:
\begin{equation}
    \m{O}_{f_\star}(t;R)\coloneqq \frac{1}{\sqrt{2\pi}}\int_{-R}^R f_{\star,0}(x)\mc{U}(x;t) \,\mathrm{d}x
\end{equation}
With the change of variables $x = \phi_R(\xi)$, we have $\mathrm{d}x = \phi_R'(\xi)\,\mathrm{d}\xi$ and therefore
\begin{equation}
    \m{O}_{f_\star}(t;R) = \int_{\R}\frac{1}{\sqrt{2\pi}} f_{\star,0}(\phi_R(\xi))\mc{U}(\phi_R(\xi);t) \phi_R'(\xi)\,\mathrm{d}\xi.
\end{equation}
The transformed integration range is $\R$, since $\phi_R(\R)=(-R,R)$. 

We now show that after the change of variables, the discretization error admits an exponentially convergent trapezoidal-rule bound.
\begin{lem}[Discretization error after the change of variables]\label{lem:quadrature}
Suppose that $\m{L}$ and $\m{H}$ satisfy the matrix hypotheses of Theorem~\ref{thm:LCHS} and fix $\eta\in(0,\eta_0)$, where
$\eta_0$ is defined in \eqref{eq:quadrature_strip_width}.
Define
\begin{equation}\label{eq:transformed_integrand}
\Phi(\zeta;\tau)
\coloneqq
F_\star(\zeta)\mc{U}(\phi_R(\zeta);\tau),
\qquad \zeta\in\Sigma_{\eta_0},
\end{equation}
where $F_\star$ is given in \eqref{eqn:feval}.
With $M_\eta$ defined by \eqref{eq:discretization_strip_constant}, for every $\delta_\xi>0$,
\begin{equation}\label{eq:strip_discretization_error}
\left\|
\int_{\R}\Phi(\xi;\tau)\,\mathrm{d}\xi
-
\delta_\xi\sum_{n\in\Z}\Phi(n\delta_\xi;\tau)
\right\|_2
\leq
\frac{2M_\eta}{\exp(2\pi\eta/\delta_\xi)-1},
\end{equation}
where the series converges absolutely.
\end{lem}

\begin{proof}
By Lemmas~\ref{lem:phi}\ref{no:analyticstructure},~\ref{lem:phi}\ref{no:transformedkernel}, and~\ref{lem:propertiesofU}, the function $\Phi(\ccdot;\tau)$ is holomorphic on $\Sigma_{\eta}$.
Moreover, Lemma~\ref{lem:phi}\ref{no:symmetricimage} and Lemma~\ref{lem:propertiesofU} give
\begin{equation}
\norm{\mc{U}(\phi_R(\zeta);\tau)}_2
\leq
\exp\!\left(
\frac{R\pi\sin\eta}{2-\frac{\pi^2}{2}\sin^2\eta}
\int_0^\tau\norm{\m{L}(s)}_2\,\mathrm{d}s
\right),
\qquad \zeta\in\overline{\Sigma}_\eta.
\end{equation}
Combining this estimate with Lemma~\ref{lem:phi}\ref{no:boundary}, we see that $\Phi(\zeta;\tau)$ tends to zero uniformly as $\abs{\operatorname{Re}\zeta}\to\infty$ in $\overline{\Sigma}_\eta$, and
\begin{equation}
\sup_{\abs{y}\leq\eta}
\int_\R\norm{\Phi(x+\ii y;\tau)}_2\,\mathrm{d}x
\leq M_\eta.
\end{equation}
The decay estimate \eqref{eq:transformed_kernel_decay} also implies absolute convergence of the trapezoidal sum for every $\delta_\xi>0$. \\

For unit vectors $\m{u},\m{v}\in\C^{2^n}$, the scalar function $\m{u}^{\dag}\Phi(\ccdot;\tau)\m{v}$ satisfies the same analyticity, decay, and integral bounds.
Applying \cite[Thm.~5.1]{trapezoidal_rule} to this function and taking the supremum over $\m{u}$ and $\m{v}$ gives \eqref{eq:strip_discretization_error}.
\end{proof}

The following lemma shows that the infinite trapezoidal sum may be truncated with only a double-logarithmic growth in the truncation range: for fixed kernel parameters, the value of $N\delta_\xi$ needed to achieve error $\varepsilon_{\rm trunc}$ grows as $O(\log\log(1/\varepsilon_{\rm trunc}))$ as $\varepsilon_{\rm trunc} \downarrow 0$.

\begin{lem}[Truncation error]\label{lem:truncation}
Suppose that $\m{L}$ and $\m{H}$ satisfy the matrix hypotheses of Theorem~\ref{thm:LCHS}, and let $\Phi$ be defined by \eqref{eq:transformed_integrand}.
Then, for every $N\in\N$ and $\delta_\xi>0$,
\begin{align}
\Bigg\lVert
\delta_\xi\sum_{n\in\Z}\Phi(n\delta_\xi;\tau)
-
\delta_\xi\sum_{\substack{n\in \Z \\|n|\leq N}}\Phi(n\delta_\xi;\tau)
\Bigg\rVert_2
\leq &\; 
\frac{2\ee^{(1-m)\Delta}}{\pi}
\big(\arctan R-\arctan\phi_R(N\delta_\xi)\big)\nonumber\\
\leq &\; 
\frac{4R\ee^{(1-m)\Delta}}{\pi}
\exp\!\big(-\pi\sinh(N\delta_\xi)\big).\label{eq:quadrature_truncation}
\end{align}
\end{lem}

\begin{proof}
On the real axis, $\norm{\mc{U}(\phi_R(\xi);\tau)}_2=1$.
By the triangle inequality and Lemma~\ref{lem:phi}\ref{no:realaxis},
\begin{align}
\left\|
\delta_\xi\sum_{n\in\Z}\Phi(n\delta_\xi;\tau)
-
\delta_\xi\sum_{|n|\leq N}\Phi(n\delta_\xi;\tau)
\right\|_2
\leq &\; 2\delta_\xi\sum_{n>N}\abs{F_\star(n\delta_\xi)}\notag\\
\leq &\; 2\int_{N\delta_\xi}^{\infty}\abs{F_\star(\xi)}\,\mathrm{d}\xi\notag\\
=&\;\frac{2\ee^{(1-m)\Delta}}{\pi}
\big(\arctan R-\arctan\phi_R(N\delta_\xi)\big).
\end{align}
Here we used the monotonicity and explicit modulus in \eqref{eq:transformed_kernel_modulus}, together with $\phi_R(\xi)\to R$ as $\xi\to\infty$.
For every $X\geq0$,
\begin{align}
\arctan R-\arctan\phi_R(X)
&\leq R-\phi_R(X)\notag\\
&=\frac{2R}{\exp(\pi\sinh X)+1}
\leq 2R\exp(-\pi\sinh X),
\end{align}
which proves the second inequality.
\end{proof}

The quadrature result is now immediate:
\begin{proof}[Proof of Theorem~\ref{thm:mainquadrature}]
By the change of variables above and Lemma~\ref{lem:quadrature},
\begin{equation}
\left\|
\m{O}_{f_\star}(\tau;R)
-
\delta_\xi\sum_{n\in\Z}\Phi(n\delta_\xi;\tau)
\right\|_2
\leq\frac{2M_\eta}{\exp(2\pi\eta/\delta_\xi)-1}.
\end{equation}
Lemma~\ref{lem:truncation} bounds the omitted terms by
$\frac{4R\ee^{(1-m)\Delta}}{\pi}\exp(-\pi\sinh(N\delta_\xi))$.
The definition of $\Phi$ identifies the finite sum with $\mathcal I_N(\tau;R,\delta_\xi)$, so the result follows from the triangle inequality.
\end{proof}

\subsection{LCU subnormalization}
\label{subsec:quadrature_subnormalization}

It remains to show that the discrete LCU subnormalization converges
to $\alpha_R[f_\star]$, with the error bound stated in
Corollary~\ref{cor:quadrature_subnormalization}.

\begin{proof}[Proof of Corollary~\ref{cor:quadrature_subnormalization}]
For real $\xi$, the map $\phi_R$ is strictly increasing from $\R$ onto
$(-R,R)$, so $\phi_R'(\xi)>0$. By the definition of $\Phi$ in
Lemma~\ref{lem:quadrature}, the expression on the left-hand side of
\eqref{eq:quadrature_subnormalization} before taking the limit is precisely
the coefficient $\ell^1$-norm, including the factor $1/\sqrt{2\pi}$, of the
LCU representation of $\mathcal I_N(\tau;R,\delta_\xi)$. \\

On the real axis, the LCU coefficient density is the restriction of
the holomorphic function
\begin{equation}\label{eq:normalization_density}
g(\zeta)\coloneqq\frac{\ee^{(1-m)\Delta}}{\pi}
\frac{\phi_R'(\zeta)}{1+\phi_R(\zeta)^2}.
\end{equation}
Indeed, \eqref{eqn:feval} and $\varepsilon=(1-m)\ee^{-m\Delta}$ give
$g(\xi)=|f_{\star,0}(\phi_R(\xi))|\phi_R'(\xi)/\sqrt{2\pi}$ for $\xi\in\R$.
The function $g$ is holomorphic on $\Sigma_{\eta}$, since it equals
$F_\star(\zeta)\ee^{\ii\Delta\sinh\zeta}$ there.
The absolute-value density has the explicit holomorphic extension $g$; applying the strip estimate to this function makes the convergence of the discrete cost quantitative.
Using \eqref{eq:phi_strip_derivative} and the denominator estimate \eqref{eq:transformed_kernel_denominator_bound}, we obtain
\begin{equation}
\abs{g(x+\ii y)}
\leq
\frac{2R\ee^{(1-m)\Delta}}
{1-\frac{\pi^2}{4}(1+R^2)\sin^2\eta}
\cosh x\,\ee^{-\pi\cos\eta\,\abs{\sinh x}},
\qquad \abs{y}\leq\eta.
\end{equation}
Thus $g$ tends to zero uniformly as $\abs{x}\to\infty$ in $\overline\Sigma_\eta$, and the substitution $t=\sinh x$ gives
\begin{equation}
\sup_{\abs{y}\leq\eta}\int_\R\abs{g(x+\ii y)}\,\mathrm{d}x
\leq C_\alpha.
\end{equation}
The same trapezoidal theorem used in Lemma~\ref{lem:quadrature} therefore gives
\begin{equation}
\left|\int_\R g(\xi)\,\mathrm{d}\xi
-\delta_\xi\sum_{n\in\Z}g(n\delta_\xi)\right|
\leq\frac{2C_\alpha}{\exp(2\pi\eta/\delta_\xi)-1}.
\end{equation}
On the real axis, $g(\xi)=\abs{F_\star(\xi)}$, so its omitted sum is bounded by the same integral as in Lemma~\ref{lem:truncation}.
This gives the second term in \eqref{eq:quantitative_subnormalization}.
Finally, the change of variables $x=\phi_R(\xi)$ gives
\begin{equation}
\frac{1}{\sqrt{2\pi}}
\int_{\R}
|f_{\star,0}(\phi_R(\xi))|\phi_R'(\xi)\,\mathrm d\xi
=
\frac{1}{\sqrt{2\pi}}
\int_{-R}^{R}|f_{\star,0}(x)|\,\mathrm dx
=
\alpha_R[f_\star].
\end{equation}
The definitions give $C_\alpha\leq M_\eta$. The triangle inequality proves \eqref{eq:quantitative_subnormalization}, and
letting $\delta_\xi\downarrow0$ and $N\delta_\xi\to\infty$ proves
\eqref{eq:quadrature_subnormalization}.
\end{proof}

\section*{Acknowledgments}
We thank all our colleagues at PsiQuantum for helpful discussions. We particularly thank David Jennings for useful comments on the manuscript.

\section*{AI contribution statement}
The project concept and proofs of the main results, including Theorems~\ref{thm:LCHS}, \ref{thm:optimal-kernel}, and \ref{thm:mainquadrature}, were developed by
the authors. OpenAI GPT-5.6, GPT-6.0, and GPT-6.1 were used to assist with
verifying and revising proofs, writing code for the numerical
results, and editing the manuscript. The authors reviewed the
AI-assisted material and take full responsibility for the final
manuscript.

\bibliographystyle{unsrt}
\bibliography{LCHS.bib}  

\pagebreak

\appendix

\section{Hardy spaces on \texorpdfstring{$S_{y_0}$}{Sy0}}
\label{appen:hardy}

Hardy spaces are spaces of holomorphic functions with integrable boundary values.
Although classically formulated on the unit disk and the half-plane
\cite{garnett2007bounded}, the finite-strip version used in this work follows by conformal mapping from the half-plane \cite{bakan2007hardy}.
Throughout this appendix, $H^1(S_{y_0})$ denotes the Hardy space introduced in~\eqref{eqn:def_H1}.
In this appendix, we collect the key definitions and results from Hardy space
theory on $S_{y_0}$, including the case $S_\infty=\H_-$ specified in~\eqref{eqn:lower_strip}.\\

\subsection{Boundary values and Poisson representation}

We first recall the standard boundary value theorem for $H^1(S_{y_0})$ functions.
For finite $y_0$, the result follows, after translation and rescaling, from the strip Hardy space theory of Bakan and Kaijser \cite[Thms.~2.1--2.2, Cor.~2.1]{bakan2007hardy}, while the lower half-plane case follows from
\cite[Chap.~II, Thm.~3.1]{garnett2007bounded}.
We state both results in the notation used throughout this work.
\begin{lem}[Boundary values of $H^1(S_{y_0}) $ functions]
\label{lem:Hardyboundary}
Let $y_0\in(1,\infty]$ and $h\in H^1(S_{y_0})$. Then the upper boundary value 
\begin{equation}
h_0(x)
\coloneqq
\lim_{\delta\downarrow 0}h(x-\ii\delta)
\end{equation}
exists for almost every $x\in\R$ and belongs to $L^1(\R)$. If
$y_0\in(1,\infty)$, the lower boundary value
\begin{equation}
h_{y_0}(x)
\coloneqq
\lim_{\delta\downarrow 0}
h(x-\ii y_0+\ii\delta)
\end{equation}
exists for almost every $x\in\R$ and belongs to $L^1(\R)$.
\end{lem}

A consequence of Lemma~\ref{lem:Hardyboundary} is that Hardy functions are completely determined by their boundary values.
The interior values are recovered from these boundary values through the Poisson integral formula.

\begin{lem}[Poisson representation {\normalfont(finite strip: \cite[Prop.~3.1]{bakan2007hardy}; half-plane: \cite[Chap.~II, Thm.~3.6]{garnett2007bounded})}]\label{lem:poisson_representation}
For $y_0\in(1,\infty)$, define the Poisson kernels associated with
the upper boundary $\R$ and the lower boundary $\R-\ii y_0$,
respectively, by
\begin{equation}\label{eqn:poisson_kernels}
\begin{aligned}
P_0(x,y)
&\coloneqq
\frac{1}{2y_0}
\frac{\sin\!\big(\frac{\pi y}{y_0}\big)}
{\cosh\!\big(\frac{\pi x}{y_0}\big)
 -\cos\!\big(\frac{\pi y}{y_0}\big)},\\
P_{y_0}(x,y)
&\coloneqq
\frac{1}{2y_0}
\frac{\sin\!\big(\frac{\pi y}{y_0}\big)}
{\cosh\!\big(\frac{\pi x}{y_0}\big)
 +\cos\!\big(\frac{\pi y}{y_0}\big)}
\end{aligned}
\quad
(x\in\R,\; y\in(0,y_0)).
\end{equation}
Then every $h\in H^1(S_{y_0})$ satisfies
\begin{equation}\label{eqn:poissonint}
h(x-\ii y)
=
\int_{\R}
P_0(x-x',y)h_0(x')\,\mathrm{d}x'
+
\int_{\R}
P_{y_0}(x-x',y)h_{y_0}(x')\,\mathrm{d}x'
\end{equation}
for every $x\in\R$ and $y\in(0,y_0)$.
When $y_0=\infty$, define
\begin{equation}\label{eqn:halfplane_poisson_kernel}
\hat{P}_0(x,y)
\coloneqq
\frac{1}{\pi}\frac{y}{x^2+y^2}
\qquad
(x\in\R,\; y\in(0,\infty)).
\end{equation}
Then every $h\in H^1(S_\infty)$ satisfies
\begin{equation}\label{eqn:halfplane_poisson_integral}
h(x-\ii y)
=
\int_{\R}
\hat{P}_0(x-x',y)h_0(x')\,\mathrm{d}x'
\end{equation}
for every $x\in\R$ and $y\in (0,\infty)$.
\end{lem}

The Poisson representation not only reconstructs harmonic functions from their boundary data but also provides an approximation of the boundary values as the boundary is approached.
The following lemma establishes this approximation in $L^1(\mathbb R)$. The convergence is uniform when the boundary data also belong to $C_0(\R)$, defined by
\begin{equation}
\label{eq:C0R}
C_0(\R)
\coloneqq
\bigg\{
g\in C(\R):
\lim_{\lvert x\rvert\to\infty}g(x)=0
\bigg\}.
\end{equation}

\begin{lem}[Poisson integrals and boundary convergence]
\label{lem:PoissonL1}
Let $y_0\in(1,\infty)$ and let
$g_0,g_{y_0}\in L^1(\R)$. For $x-\ii y\in S_{y_0}$, define
\begin{equation}
\label{eq:u_harmonic}
u(x-\ii y)
\coloneqq
\int_{\R}
P_0(x-x',y)g_0(x')\,\mathrm{d}x'
+
\int_{\R}
P_{y_0}(x-x',y)g_{y_0}(x')\,\mathrm{d}x'.
\end{equation}
Then, $u$ is harmonic on $S_{y_0}$. Moreover, for each $y\in (0,y_0)$,  $u(\cdot-\ii y)\in L^1(\R)$ and
\begin{equation}
\lim_{\delta\downarrow0}
\norm{u(\cdot-\ii\delta)-g_0}_1
=
\lim_{\delta\downarrow0}
\norm{u(\cdot-\ii y_0+\ii\delta)-g_{y_0}}_1
=
0.
\end{equation}
If, in addition, $g_0,g_{y_0}\in C_0(\R)$, then
\begin{equation}
\lim_{\delta\downarrow0}
\norm{u(\cdot-\ii\delta)-g_0}_\infty
=
\lim_{\delta\downarrow0}
\norm{u(\cdot-\ii y_0+\ii\delta)-g_{y_0}}_\infty
=
0.
\end{equation}

In the case $y_0=\infty$, let $g_0\in L^1(\R)$ and define
\begin{equation}
u(x-\ii y)
\coloneqq
\int_{\R}
\hat{P}_0(x-x',y)g_0(x')\,\mathrm{d}x'.
\end{equation}
Then, $u$ is harmonic on $S_{\infty}=\H_-$. Moreover, for each $y\in (0,\infty)$,  $u(\cdot-\ii y)\in L^1(\R)$ and
\begin{equation}
\label{eq:L1_convergence_infty}
\lim_{\delta\downarrow0}
\norm{u(\cdot-\ii\delta)-g_0}_1
=
0.
\end{equation}
If, in addition, $g_0\in C_0(\R)$, then
\begin{equation}
\label{eq:uniform_convergence_infty}
\lim_{\delta\downarrow0}
\norm{u(\cdot-\ii\delta)-g_0}_\infty
=
0.
\end{equation}
\end{lem}

\begin{proof}

The harmonicity assertions
follow from \cite[Thm.~B.1]{corso2026generalized}, in the finite-strip case, after translation and rescaling, and from \cite[Chap.~I, Sec.~3]{garnett2007bounded} in the half-plane case. It remains to prove the stated $L^1$ and uniform boundary convergence. \\

For $y\in(0,y_0)$, both Poisson kernels are nonnegative. The
substitution $s = \tanh(\tfrac{\pi x}{2y_0})$ gives
\begin{align}
\label{eq:poisson_mass_0}
\norm{P_0(\ccdot,y)}_1
=&\; 
\int_{\R}P_0(x,y)\,\mathrm{d}x \nonumber \\
=&\; \frac{\sin\!\big(\frac{\pi y}{y_0}\big)}{2\pi}
\int_{-1}^{1}
\frac{\mathrm{d}s}
{\sin^2\!\big(\frac{\pi y}{2y_0}\big)
+\cos^2\!\big(\frac{\pi y}{2y_0}\big)s^2}
\nonumber\\
=&\;
\frac{\sin\!\big(\frac{\pi y}{y_0}\big)}{2\pi}
\left[
\frac{
\arctan\!\big(
s\cot(\frac{\pi y}{2y_0})
\big)}
{
\sin\!\big(\frac{\pi y}{2y_0}\big)
\cos\!\big(\frac{\pi y}{2y_0}\big)
}
\right]_{-1}^{1}
\nonumber\\
=&\;
\frac{2}{\pi}
\arctan
\cot\!\bigg(\frac{\pi y}{2y_0}
\bigg)
=
1-\frac{y}{y_0}.
\end{align}
Since $P_{y_0}(x,y)=P_0(x,y_0-y)$, \eqref{eq:poisson_mass_0} gives
\begin{equation}
\|P_{y_0}(\cdot,y)\|_1
=
\|P_0(\cdot,y_0-y)\|_1
=
1-\frac{y_0-y}{y_0}
=
\frac{y}{y_0}.
\end{equation}

Now, Young's convolution inequality 
applied to \eqref{eq:u_harmonic} gives 
\begin{equation}
\norm{u(\cdot-\ii y)}_1
\leq
\bigg(1-\frac{y}{y_0}\bigg)\norm{g_0}_1
+
\frac{y}{y_0}\norm{g_{y_0}}_1,
\end{equation}
so that $u(\cdot-\ii y)\in L^1(\R)$. \\

For $\delta\in (0,y_0)$, define the normalized kernel
\begin{equation}
\rho_\delta(x)
\coloneqq
\frac{1}{1-\frac{\delta}{y_0}}
P_0(x,\delta) \quad (x\in \R).
\end{equation}
The function $\rho_\delta$ has unit $L^1$ norm by \eqref{eq:poisson_mass_0}. Moreover, for every $r>0$,
\begin{equation}
\int_{\R\setminus(-r,r)}
\lvert \rho_\delta(x)\rvert\,\mathrm{d}x
\leq 
\frac{\sin\!\big(\frac{\pi\delta}{y_0}\big)}
{2(y_0-\delta)}
\int_{\R\setminus(-r,r)}
\frac{\mathrm{d}x}
{\cosh\!\big(\frac{\pi x}{y_0}\big)-1}=\frac{y_0}{\pi}\frac{\sin\!\big(\frac{\pi\delta}{y_0}\big)}
{y_0-\delta}\big(\coth\!\big(\tfrac{\pi r}{2y_0}\big)-1\big).
\end{equation}
The right-hand side tends to zero as $\delta\downarrow0$, so $(\rho_\delta)_{\delta\in(0,y_0)}$ is an approximate identity.
Since every function in $C_0(\R)$ is bounded and uniformly continuous, the standard approximate identity theorem \cite[Thm.~1.2.19]{Grafakos2014} gives
\begin{equation}
\label{eq:Qdelta_L1}
\lim_{\delta\downarrow 0}\norm{\rho_\delta\ast g-g}_1= 0 \quad (g\in L^1(\R))
\end{equation}
and
\begin{equation}
\label{eq:Qdelta_uniform}
\lim_{\delta\downarrow 0}\norm{\rho_\delta\ast g-g}_\infty= 0\quad 
(g\in C_0(\R)).
\end{equation}

Since
\begin{equation}
P_0(\ccdot,\delta)
=
\bigg(1-\frac{\delta}{y_0}\bigg)\rho_\delta,
\qquad
\norm{P_{y_0}(\ccdot,\delta)}_1
=
\frac{\delta}{y_0},
\end{equation}
we have
\begin{equation}
\begin{aligned}
u(\cdot-\ii\delta)-g_0
=
\bigg(1-\frac{\delta}{y_0}\bigg)
(\rho_\delta\ast g_0-g_0)+
P_{y_0}(\ccdot,\delta)\ast g_{y_0}
-
\frac{\delta}{y_0}g_0.
\end{aligned}
\end{equation}
For $p=1$, and also for $p=\infty$ when
$g_0,g_{y_0}\in C_0(\R)$, Young's convolution inequality therefore
gives
\begin{equation}
\begin{aligned}
\norm{u(\cdot-\ii\delta)-g_0}_p
\leq{}&
\bigg(1-\frac{\delta}{y_0}\bigg)
\norm{\rho_\delta\ast g_0-g_0}_p+
\frac{\delta}{y_0}
\big(
\norm{g_0}_p+\norm{g_{y_0}}_p\big).
\end{aligned}
\end{equation}
The upper-boundary assertions now follow from
\eqref{eq:Qdelta_L1}--\eqref{eq:Qdelta_uniform}. The identities
\begin{equation}
P_{y_0}(x,y_0-\delta)
=
P_0(x,\delta),
\quad
P_0(x,y_0-\delta)
=
P_{y_0}(x,\delta)
\end{equation}
give the lower-boundary assertions by a similar argument with
$g_0$ and $g_{y_0}$ interchanged. \\

In the half-plane case, define
\begin{equation}
\hat{\rho}_\delta(x)
\coloneqq
\hat{P}_0(x,\delta)
=
\frac{1}{\pi}\frac{\delta}{x^2+\delta^2}.
\end{equation}
Then
\begin{equation}
\lVert \hat{\rho}_{\delta}\rVert_1=\lVert \hat{P}_0(\ccdot,\delta)\rVert_1=\int_{\R}\hat{P}_0(x,\delta)\,\mathrm{d}x=\frac{\delta}{\pi}\int_\R \frac{\mathrm{d}x}{x^2+\delta^2}=1
\end{equation}
and, for every $r\in (0,\infty)$,
\begin{equation}
\int_{\R\setminus(-r,r)}
\hat{\rho}_\delta(x)\,\mathrm{d}x
\leq
\frac{\delta}{\pi}\int_{\R\setminus (-r,r)} \frac{\mathrm{d}x}{x^2}
=\frac{2\delta}{\pi r}.
\end{equation}
Thus 
\begin{equation}
\lim_{\delta\downarrow 0}	\int_{\R\setminus(-r,r)}
\hat{\rho}_\delta(x)\,\mathrm{d}x=0
\end{equation}
and $(\hat{\rho}_\delta)_{\delta\in (0,\infty)}$ is an approximate identity. Young's
convolution inequality gives, for every $y\in(0,\infty)$,
\begin{equation}
\norm{u(\cdot-\ii y)}_1
\leq
\norm{g_0}_1.
\end{equation}
The same approximate identity argument proves the $L^1$ and uniform boundary convergence in the half-plane case.
\end{proof}

\subsection{Inner--outer factorization}

Hardy functions being determined by their boundary values leads to a natural decomposition according to boundary behaviour.
An inner function has unimodular boundary values almost everywhere, while an outer function is determined up to a unimodular constant by the modulus of its boundary values.
More precisely,

\begin{defn}[Inner and outer functions]
\label{def:inner_outer}
Let $y_0\in(1,\infty]$. A function $h_{\mathrm{in}}\in\mc{O}(S_{y_0})$ is called
\emph{inner} if
\begin{equation}
\lvert h_{\mathrm{in}}(z)\rvert\leq 1
\qquad (z\in S_{y_0})
\end{equation}
and its non-tangential boundary values are unimodular almost everywhere
on each boundary component. A nonzero function $h_{\mathrm{out}}\in H^1(S_{y_0})$ is called
\emph{outer} if, for every $x-\ii y\in S_{y_0}$, the following
representation holds. When $y_0\in(1,\infty)$,
\begin{equation}
\label{eqn:outer_poisson_representation}
\begin{aligned}
\log\,\abs{h_{\mathrm{out}}(x-\ii y)}
={}&
\int_{\R}
P_0(x-x',y)\log\,\abs{h_{\mathrm{out},0}(x')}\,
\mathrm{d}x'\\
&+
\int_{\R}
P_{y_0}(x-x',y)
\log\,\abs{h_{\mathrm{out},y_0}(x')}\,
\mathrm{d}x'.
\end{aligned}
\end{equation}
When $y_0=\infty$,
\begin{equation}
\label{eqn:outer_poisson_representation_halfplane}
\log\,\abs{h_{\mathrm{out}}(x-\ii y)}
=
\int_{\R}
\hat{P}_0(x-x',y)\log\,\abs{h_{\mathrm{out},0}(x')}\,
\mathrm{d}x'.
\end{equation}
\end{defn}

It is a standard result that Hardy space functions on the half-plane factorize~\cite[Chap.~11.3]{duren1970theory}.
For finite strips, it follows by transporting the half-plane result through a conformal map.
For completeness, we include a proof.

\begin{thm}[Beurling factorization on $S_{y_0}$]
\label{thm:beurling}
Let $y_0\in(1,\infty]$ and let
$h\in H^1(S_{y_0})$ be nonzero. Then,
\begin{equation}
h=h_{\mathrm{in}}h_{\mathrm{out}},
\end{equation}
where $h_{\mathrm{in}}$ is inner and $h_{\mathrm{out}}$ is outer. This
factorization is unique up to multiplication of the two factors by
reciprocal unimodular constants. Moreover,
\begin{equation}
\label{eq:hout0_equality}
\lvert h_{\mathrm{out},0}(x)\rvert
=
\lvert h_0(x)\rvert
\end{equation}
for almost every $x\in\R$ and, when $y_0\in (1,\infty)$, 
\begin{equation}
\lvert h_{\mathrm{out},y_0}(x)\rvert
=
\lvert h_{y_0}(x)\rvert
\end{equation}
for almost every $x\in\R$. If
$h_{\mathrm{in}}$ is nonconstant, then
\begin{equation}
\lvert h_{\mathrm{in}}(z)\rvert<1
\end{equation}
for every $z\in S_{y_0}$. 
\end{thm}

\begin{proof}
The case $y_0=\infty$ is given in \cite[Chap.~11.3]{duren1970theory}.
Let
\begin{equation}
    \psi(w)=\frac{y_0}{\pi}\log w,
\end{equation}
which maps $\H_-$ biholomorphically onto $S_{y_0}$, and define
\begin{equation}
    h^\infty(w)
=
\psi'(w)h(\psi(w)).
\end{equation}
By \cite[Thms.~2.1--2.2]{bakan2007hardy}, this identifies
$H^1(S_{y_0})$ with $H^1(\H_-)$.
Applying the classical half-plane factorization to $h^\infty$ and transporting the resulting inner and outer factors back to $S_{y_0}$ yields the stated decomposition. Explicitly, if $h^{\infty}=h^{\infty}_{\mathrm{in}}h^{\infty}_{\mathrm{out}}$, with $h^{\infty}_{\mathrm{in}}$ and $h^{\infty}_{\mathrm{out}}$ respectively inner and outer in $H^1(\H_-)$, then
\begin{equation}
h_{\mathrm{in}}\coloneqq h^{\infty}_{\mathrm{in}}\circ\psi^{-1}, \qquad h_{\mathrm{out}}\coloneqq \frac{h^{\infty}_{\mathrm{out}}\circ\psi^{-1}}{\psi'\circ\psi^{-1}}.
\end{equation}
Under $\psi^{-1}$, the half-plane Poisson kernel \eqref{eqn:halfplane_poisson_kernel} pulls back to the finite-strip kernels \eqref{eqn:poisson_kernels}, so $\log\,\lvert  h^{\infty}_{\mathrm{out}}\circ\psi^{-1}\rvert$ is the finite-strip Poisson integral of its boundary values. Moreover,
\begin{equation}
\log\bigg\lvert\frac{1}{\psi'(\psi^{-1}(z))} \bigg\rvert=\log\frac{\pi}{y_0}+\frac{\pi}{y_0}\operatorname{Re}z 
\end{equation}
equals its own finite-strip Poisson integral, since the Poisson kernels \eqref{eqn:poisson_kernels} are even, decay exponentially, and have combined unit mass. Hence $h_{\mathrm{out}}$ satisfies \eqref{eqn:outer_poisson_representation}. The preservation of boundary modulus, uniqueness up to reciprocal unimodular constants, and the strict inequality $|h_{\mathrm{in}}|<1$ for nonconstant inner functions follow immediately from the corresponding half-plane results under the conformal correspondence.
\end{proof}

The Poisson representation of outer functions gives the following uniqueness result.

\begin{cor}[Uniqueness of outer functions]
\label{cor:outer_uniqueness}
Let $y_0\in(1,\infty)\cup\{\infty\}$, and let
$h_{\mathrm{out}},\tilde{h}_{\mathrm{out}}\in H^1(S_{y_0})$ be
outer functions. Suppose that their boundary values have equal modulus
almost everywhere on $\partial S_{y_0}$. Then, there
exists $\theta\in [0,2\pi)$ such that
\begin{equation}
\tilde{h}_{\mathrm{out}}
=
\ee^{\ii\theta}h_{\mathrm{out}}.
\end{equation}
\end{cor}

\begin{proof}
The Poisson representations in Definition~\ref{def:inner_outer}
give
\begin{equation}
\lvert\tilde{h}_{\mathrm{out}}(z)\rvert
=
\lvert h_{\mathrm{out}}(z)\rvert
\qquad (z\in S_{y_0}).
\end{equation}
Outer functions are zero-free in $S_{y_0}$, so the quotient
$\tilde{h}_{\mathrm{out}}/{h_{\mathrm{out}}}$
is holomorphic and has constant modulus one. The maximum-modulus principle
therefore implies that this quotient is constant, which proves the result.
\end{proof}

\subsection{Identities for Poisson kernels}

We conclude by establishing two identities for the Poisson kernels
that are used in the main text.
We will later require explicit expressions for the Poisson mass over finite intervals of the boundary.

\begin{lem}[Truncated Poisson mass]
\label{lem:int_poisson}
The truncated Poisson masses over the boundary interval $[-R,R]$, where
$R>0$, are given by
\begin{equation}
m_{R,y_0}
\coloneqq
\int_{-R}^{R}P_0(x,1)\,\mathrm{d}x
=
\frac{2}{\pi}
\arctan\!\left(
\frac{\tanh\!\big(\frac{\pi R}{2y_0}\big)}
{\tan\!\big(\frac{\pi}{2y_0}\big)}
\right),
\qquad
y_0\in(1,\infty),
\end{equation}
and
\begin{equation}
m_{R,\infty}
\coloneqq
\int_{-R}^{R}\hat P_0(x,1)\,\mathrm{d}x
=
\frac{2}{\pi}\arctan(R).
\end{equation}

The corresponding limiting values are
\begin{equation}\label{eqn:truncated_poisson_mass_limits}
\lim_{R\downarrow0}m_{R,y_0}=0,
\qquad
\lim_{R\to\infty}m_{R,y_0}
=
\begin{dcases}
1-\dfrac{1}{y_0}
& y_0\in(1,\infty)\\
1
& y_0=\infty.
\end{dcases}
\end{equation}
\end{lem}

\begin{proof}
For $y_0\in(1,\infty)$, the substitution $s = \tanh(\tfrac{\pi x}{2y_0})$ gives
\begin{align}
\int_{-R}^{R}P_0(x,1)\,\mathrm{d}x
=&\; 
\frac{\sin\!\big(\frac{\pi}{y_0}\big)}{2\pi}
\int_{-\tanh(\frac{\pi R}{2y_0})}
^{\tanh(\frac{\pi R}{2y_0})}
\frac{\mathrm{d}s}
{\sin^2\!\big(\frac{\pi}{2y_0}\big)
+\cos^2\!\big(\frac{\pi}{2y_0}\big)s^2}\notag\\
=&\; 
\frac{\sin\!\big(\frac{\pi}{y_0}\big)}{2\pi}
\left[
\frac{
\arctan\!\big(
s\cot(\frac{\pi}{2y_0})
\big)}
{
\sin\!\big(\frac{\pi}{2y_0}\big)
\cos\!\big(\frac{\pi}{2y_0}\big)
}
\right]_
{-\tanh(\frac{\pi R}{2y_0})}
^{\tanh(\frac{\pi R}{2y_0})}\notag\\
=&\; 
\frac{2}{\pi}
\arctan\!\left(
\frac{
\tanh\!\big(\frac{\pi R}{2y_0}\big)}
{\tan\!\big(\frac{\pi}{2y_0}\big)}
\right).
\end{align}

When $y_0=\infty$, equation
\eqref{eqn:halfplane_poisson_kernel} gives
\begin{equation}
m_{R,\infty}
=
\frac{1}{\pi}
\int_{-R}^{R}\frac{\mathrm{d}x}{1+x^2}
=
\frac{2}{\pi}\arctan(R).
\end{equation}
The limits follow directly from the explicit expressions above using the limits $\tanh(t)\to1$ and $\arctan(t)\to\pi/2$ as $t\to\infty$.
\end{proof}

\begin{lem}[Logarithmic Poisson integrals]
\label{lem:poisson_log_integral}
For $y_0\in(1,\infty)$, \begin{equation}\label{eqn:poisson_log_integral}
\begin{aligned}
\int_{\R}
\big(
P_0(x,1)\log P_0(x,1)+
P_{y_0}(x,1)\log P_{y_0}(x,1)
\big)\,\mathrm{d}x
=
-\log\!\left(
4y_0\sin\!\left(\frac{\pi}{y_0}\right)
\right).
\end{aligned}
\end{equation}
In the half-plane case,
\begin{equation}\label{eqn:halfplane_poisson_log_integral}
\int_{\R}
\hat{P}_0(x,1)\log\hat{P}_0(x,1)\,\mathrm{d}x
=
-\log(4\pi).
\end{equation}
Moreover, these integrals are absolutely convergent. 
\end{lem}

\begin{proof}
Fix $y_0\in(1,\infty)$ and define $a\coloneqq\tfrac{\pi}{y_0}$.
On the upper boundary make the substitution
$u=\exp(\pi x/y_0)$, and on the lower boundary make the
substitution $u=-\exp(\pi x/y_0)$.
Under these substitutions, corresponding to the conformal map $z\mapsto\exp(\pi z/y_0)$ from $S_{y_0}$ to the lower half-plane, the integral over the two strip boundary components combines into a single integral against the half-plane Poisson kernel:
\begin{align}
&\int_{\R}
\big(
P_0(x,1)\log P_0(x,1)
+
P_{y_0}(x,1)\log P_{y_0}(x,1)
\big)\,\mathrm{d}x
\nonumber\\
&\qquad=
\int_{\R}
\hat P_0(u-\cos a,\sin a)
\log\left(
\frac{\pi}{y_0}|u|
\hat P_0(u-\cos a,\sin a)
\right)\mathrm{d}u.
\label{eq:poisson_log_halfplane_reduction}\\
&\qquad= \log\left(\frac{\pi}{y_0}\right)
+
\int_{\R}
\hat P_0(u-\cos a,\sin a)\log\,\abs{u}\,\mathrm{d}u
\nonumber\\
&\qquad\quad+
\int_{\R}
\hat P_0(u-\cos a,\sin a)
\log\hat P_0(u-\cos a,\sin a)\,\mathrm{d}u.
\label{eq:poisson_log_split}
\end{align}
The integrand on the right-hand side in ~\eqref{eq:poisson_log_halfplane_reduction} is $O(|\log\,\abs{u}|)$ near $u=0$ and
$O(\log\,\abs{u}/u^2)$ as $|u|\to\infty$.
It is therefore absolutely integrable, and we may split the logarithm to obtain the final line. \\

In both remaining integrals, make the substitution
\begin{equation}
u=\cos a+\sin a\tan s,
\qquad
s\in\left(-\frac{\pi}{2},\frac{\pi}{2}\right).
\end{equation}
Then
\begin{equation}
\hat P_0(u-\cos a,\sin a)\,\mathrm{d}u
=
\frac{\mathrm{d}s}{\pi}.
\end{equation}
Consequently,
\begin{align}
\int_{\R}
\hat P_0(u-\cos a,\sin a)\log\,\abs{u}\,\mathrm{d}u
=
\frac{1}{\pi}
\int_{-\pi/2}^{\pi/2}
\log\,\left|
\frac{\cos(s-a)}{\cos s}
\right|\mathrm{d}s=0,
\label{eq:shifted_cauchy_log}
\end{align}
because $\log\,\abs{\cos s}$ is $\pi$-periodic and both integrals are
taken over intervals of length $\pi$.
Moreover,
\begin{align}
\int_{\R}
\hat P_0(u-\cos a,\sin a)
\log\hat P_0(u-\cos a,\sin a)\,\mathrm{d}u &=
\frac{1}{\pi}
\int_{-\pi/2}^{\pi/2}
\log\left(
\frac{\cos^2 s}{\pi\sin a}
\right)\mathrm{d}s
\nonumber\\
&=
-\log(\pi\sin a)-2\log 2
\nonumber\\
&=
-\log(4\pi\sin a),
\label{eq:shifted_cauchy_log_density}
\end{align}
where we used
\begin{equation}
\int_{-\pi/2}^{\pi/2}\log(\cos s)\,\mathrm{d}s
=
-\pi\log 2.
\end{equation}
Substituting \eqref{eq:shifted_cauchy_log} and
\eqref{eq:shifted_cauchy_log_density} into
\eqref{eq:poisson_log_split}, gives
\begin{align}
\int_{\R}
\big(
P_0(x,1)\log P_0(x,1)
+
P_{y_0}(x,1)\log P_{y_0}(x,1)
\big)\,\mathrm{d}x
&=
\log\left(\frac{\pi}{y_0}\right)
-\log(4\pi\sin a)
\nonumber\\
&=
-\log\left(
4y_0\sin\left(\frac{\pi}{y_0}\right)
\right).
\end{align}

For the half-plane kernel,
\begin{equation}
\hat P_0(x,1)=\frac{1}{\pi(1+x^2)}.
\end{equation}
The substitution $x=\tan s$ gives
\begin{align}
\int_{\R}
\hat P_0(x,1)\log\hat P_0(x,1)\,\mathrm{d}x
=&\ 
\frac{1}{\pi}
\int_{-\pi/2}^{\pi/2}
\log\bigg(\frac{\cos^2s}{\pi}\bigg)\mathrm{d}s
\nonumber\\
=&\;
-\log\pi-2\log2
=
-\log(4\pi).
\end{align}
The half-plane integrand is also absolutely integrable, since it is
bounded near the origin and is
$O(\log\,\abs{x}/x^2)$ as $|x|\to\infty$.
\end{proof}

\section{Hardy space LCHS and admissibility conditions}

We prove Theorem~\ref{thm:LCHS}, establish the strict relaxation
of the kernel function hypotheses, and characterize exact LCHS kernel functions.

\subsection{Proof of Theorem~\ref{thm:LCHS}}
\label{subsec:LCHSproof}

We first establish interior integrability of admissible kernel functions $f\in \mc{A}(S_{y_0})$. 

\begin{lem}[Interior integrability of admissible kernel functions]\label{lem:squareint}
Let $y_0\in (1,\infty]$ and $f\in\mathcal{A}(S_{y_0})$. Then for each $y\in (0,y_0)\setminus\{1\}$, $f(\,\cdot-\ii y)\in L^1(\R)$.
When $y_0=\infty$, we also have
\begin{equation}
\label{eq:halfplane_uniform_slice_bound}
\norm{f(\,\cdot-\ii y)}_1
\leq
\frac{3}{\sqrt{2\pi}}
\sup_{s>0}\norm{h(\,\cdot-\ii s)}_1
\qquad (y\geq2).
\end{equation}
Here, $h$ is the Hardy space function associated with $f$ via \eqref{eqn:kernel_f_h}.
\end{lem}

\begin{proof}
We write
    \begin{equation}
    \label{eq:f_to_h_thm4}
        f(z) = \frac{\mathrm{i}}{\sqrt{2\pi}}\frac{h(z)}{B(z;y_0)},
    \end{equation}
recalling that $h\in H^1(S_{y_0})$ implies $h(\cdot-\ii y)\in L^1(\R)$ for every $y\in (0,y_0)$. \\    
    
   Fix $y\in (0,y_0)\setminus \{1\}$. We show that $1/B(\cdot-\ii y)$ is bounded on $\R$. If $y_0\in (1,\infty)$, we have
        \begin{align}
            \bigg\lvert\frac{1}{B(x-\mathrm{i}y;y_0)}\bigg\rvert^2 & = \frac{\mathrm{sinh}^2\big(\frac{\pi}{2y_0}x\big)+\sin^2\big(\frac{\pi}{2y_0}(1+y)\big)}{\mathrm{sinh}^2\big(\frac{\pi}{2y_0}x\big)+\sin^2\big(\frac{\pi}{2y_0}(1-y)\big)} \notag\\
            &\leq \max\bigg\{1,\frac{\sin^2\big(\frac{\pi}{2y_0}(1+y)\big)}{\sin^2\big(\frac{\pi}{2y_0}(1-y)\big)} \bigg\}\eqqcolon b_0(y;y_0)^2\label{eq:b0}
        \end{align}
    using the trigonometric identity $\abs{\sinh(u \pm\mathrm{i}v)}^2 = \sinh^2u + \sin^2 v$.
    We observe that $b_0(y;y_0)$ is finite for $y\in (0,y_0)\setminus\{1\}$. \\
    
    Similarly, in the case $y_0 = \infty$,
        \begin{align}
            \bigg\lvert\frac{1}{B(x-\mathrm{i}y;\infty)}\bigg\rvert^2 & = \frac{x^2 + (1+y)^2}{x^2 + (1-y)^2}\notag \\
            &\leq \frac{(1+y)^2}{(1-y)^2} \eqqcolon b_0(y;\infty)^2.\label{eq:b0_infty}
        \end{align}
    which is finite for $y\in (0,\infty)\setminus\{1\}$. \\

    Now, the representation \eqref{eq:f_to_h_thm4} and the estimates \eqref{eq:b0}--\eqref{eq:b0_infty} give, for each $y\in (0,y_0)\setminus \{1\}$, 
        \begin{align}
            \int_{\mathbb{R}}\lvert f(x-\ii y)\rvert\,\mathrm{d}x = &\ \frac{1}{\sqrt{2\pi}}\int_\mathbb{R}\frac{\lvert h(x-\ii y)\rvert}{\lvert B(x-\ii y;y_0)\rvert} \,\mathrm{d}x\notag\\
             \leq &\;  \frac{b_0(y;y_0)}{\sqrt{2\pi}}\int_{\mathbb{R}}\lvert h(x-\ii y)\rvert\,\mathrm{d}x < \infty,
        \end{align}
    as desired. When $y_0=\infty$ and $y\geq2$, \eqref{eq:b0_infty} gives
    \begin{equation}
    \sup_{x\in\R}
    \bigg\lvert\frac{1}{B(x-\ii y;\infty)}\bigg\rvert
    \leq
    \frac{y+1}{y-1}
    \leq3.
    \end{equation}
    Combining this estimate with \eqref{eq:f_to_h_thm4} proves
    \eqref{eq:halfplane_uniform_slice_bound}.
\end{proof}

We can then proceed with the proof of Theorem~\ref{thm:LCHS}. 

\paragraph{Step 1: contour integral representation.}
To treat the two cases with a common contour argument, set $\eta=0$ when
$y_0\in(1,\infty)$ and fix $\eta>0$ when $y_0=\infty$. Thus, the
regularization is used only in the half-plane case. Define
\begin{equation}
\m{L}_\eta(t)\coloneqq \m{L}(t)+\eta\m{I}
\end{equation}
and let
\begin{equation}
\mc{U}_\eta(z;t)
\coloneqq
\mathcal{T}\exp\bigg(
-\ii\int_0^t
\big(z\m{L}_\eta(t')+\m{H}(t')\big)\,\mathrm{d}t'
\bigg).
\end{equation}
Since the scalar term $z\eta\m{I}$ commutes with the remaining
generator,
\begin{equation}
\label{eq:regularized_propagator_factorization}
\mc{U}_\eta(z;t)
=
\mathrm{e}^{-\ii\eta zt}\mc{U}(z;t).
\end{equation}
In particular, $\mc{U}_\eta(\ccdot;t)$ is entire and, by
\cite[Lem.~4]{low2025},
\begin{equation}
\label{eq:regularized_contraction}
\norm{\mc{U}_\eta(x-\ii y;t)}_2
\leq
\mathrm{e}^{-\eta yt}
\leq1
\qquad (x\in\R,\ y>0).
\end{equation}

Fix $\delta\in(0,1)$ and $Y>1$. In the finite-strip case, we take
$Y=y_0-\delta$ and require
$\delta<y_0-1$. In the half-plane case, $\delta$ and $Y$ are arbitrary. For $M>0$, define a negatively-oriented contour by 
\begin{equation}
\label{eq:Gamma}
\Gamma(M,\delta, Y)\coloneqq \Gamma_\text{top}(M,\delta,Y)\cup\Gamma_\text{right}(M,\delta, Y)\cup\Gamma_\text{bottom}(M,\delta, Y)\cup\Gamma_\text{left}(M,\delta, Y),	
\end{equation}
where
\begin{equation}
\label{eq:Gamma_segments}
\begin{split}
\Gamma_\text{top}(M,\delta, Y)\coloneqq &\;  [-M-\ii \delta,M-\ii \delta ], \\
\Gamma_\text{right}(M,\delta, Y)\coloneqq &\;  [M-\ii \delta,M-\ii Y ], \\
\Gamma_\text{bottom}(M,\delta, Y)\coloneqq &\;  [M-\ii Y ,-M-\ii Y ], \\
\Gamma_\text{left}(M,\delta, Y)\coloneqq &\;  [-M-\ii Y, -M-\ii\delta ]. 
\end{split}
\end{equation}
The function $z\mapsto f(z)\mc{U}_\eta(z;t)$ is meromorphic in $z$ on $S_{y_0}$; the only singularity is a simple pole at $z=-\ii$ with residue $\frac{\ii}{\sqrt{2\pi}}\mc{U}_\eta(-\ii;t)$.
Hence by Cauchy's residue theorem, we have
\begin{equation}
\label{eq:Cauchy}
\frac{1}{\sqrt{2\pi}}\oint_{\Gamma(M,\delta,Y)}f(z)\mc{U}_\eta(z;t)\,\mathrm{d}z
=
\mc{U}_\eta(-\ii;t)
=
\mathrm{e}^{-\eta t}\m{U}(t).
\end{equation}

\paragraph{Step 2: vanishing of the vertical contributions.}

We next show that the contributions of the vertical segments in \eqref{eq:Gamma_segments} to the integral \eqref{eq:Cauchy} vanish as $M \to \infty$. Consider the contribution of $\Gamma_{\mathrm{right}}(M,\delta,Y)$. We have, using \eqref{eq:f_to_h_thm4},
    \begin{equation}
\bigg\lVert\int_{\Gamma_{\textup{right}}(M,\delta,Y)} f(z)\mathcal{U}_\eta(z;t)\,\mathrm{d}z\bigg\rVert_2 \leq \frac{1}{\sqrt{2\pi}}\int_{\delta}^{Y} \frac{\lvert h(M-\ii y)\rvert}{\abs{B(M-\ii y;y_0)}}\norm{\mathcal{U}_\eta(M-\ii y;t)}_2 \,\mathrm{d}y.
    \end{equation}
By \eqref{eq:regularized_contraction}, $\lVert\mathcal{U}_\eta(M-\ii y;t)\rVert_2\leq 1$ for $y>0$. Moreover, we show that for sufficiently large $M$, the estimate $\abs{B(M-\ii y;y_0)}^{-1}\leq \sqrt{2}$ holds. For $y_0\in (1,\infty)$, take $M$ large enough so that $\sinh\big(\frac{\pi}{2y_0}M\big)>1$ holds. Then,    
\begin{align}
\bigg\lvert\frac{1}{B(M-\mathrm{i}y;y_0)}\bigg\rvert^2
=&\; 
\frac{\sinh^2\!\big(\frac{\pi}{2y_0}M\big)+\sin^2\!\big(\frac{\pi}{2y_0}(y+1)\big)}
{\sinh^2\!\big(\frac{\pi}{2y_0}M\big)+\sin^2\!\big(\frac{\pi}{2y_0}(y-1)\big)} \notag\\
\leq &\ 
1+\frac{1}{\sinh^2\!\big(\frac{\pi}{2y_0} M\big)}
< 2.
\end{align}
In the case $y_0=\infty$, take $M>Y+1$. Then,
\begin{equation}
\bigg\lvert\frac{1}{B(M-\mathrm{i}y;\infty)}\bigg\rvert^2=\frac{M^2+(1+y)^2}{M^2+(1-y)^2}< 2.
\end{equation}
It follows that, in both cases, for sufficiently large $M$,
    \begin{equation} \label{eqn:subPoisson} \bigg\lVert\int_{\Gamma_{\textup{right}}(M,\delta,Y)} f(z)\mathcal{U}_\eta(z;t)\,\mathrm{d}z\bigg\rVert_2 \leq \frac{1}{\sqrt{\pi}}\int_{\delta}^Y \lvert h(M-\ii y)\rvert \,\mathrm{d}y.
    \end{equation}
To proceed, we use the Poisson integral formula in Lemma~\ref{lem:poisson_representation}. In the case $y_0\in (1,\infty)$, we have
\begin{align}
\int_{\delta}^Y\lvert h(M-\ii y)\rvert \,\mathrm{d}y=&\;  \int_{\delta}^Y \bigg\lvert	\int_{\R} P_{0}(M-x',y)h_0(x')\,\mathrm{d}x'+\int_{\R} P_{y_0}(M-x',y)h_{y_0}(x')\,\mathrm{d}x'  \bigg\rvert\,\mathrm{d}y \nonumber \\
\leq &\; \int_{\delta}^Y\int_\R P_0(M-x',y)\abs{h_0(x')}\,\mathrm{d}x'\,\mathrm{d}y \nonumber \\
&\; +\int_{\delta}^Y\int_\R P_{y_0}(M-x',y)\abs{h_{y_0}(x')}\,\mathrm{d}x'\,\mathrm{d}y. 
\end{align}
Now, Tonelli's theorem gives
\begin{align}
\int_{\delta}^{Y}\int_{\R}
P_0(M-x',y)|h_0(x')|\,\mathrm{d}x'\,\mathrm{d}y
&=
\int_{\R}K_{0}(M-x')\abs{h_0(x')}\,\mathrm{d}x',
\\
\int_{\delta}^{Y}\int_{\R}
P_{y_0}(M-x',y)\abs{h_{y_0}(x')}\,\mathrm{d}x'\,\mathrm{d}y
&=
\int_{\R}K_{y_0}(M-x')\abs{h_{y_0}(x')}\,\mathrm{d}x',
\end{align}
where
\begin{equation}
\label{eq:K_definitions}
K_0(x) := \int_{\delta}^{Y} P_0(x,y)\, \mathrm{d}y,\quad K_{y_0}(x) := \int_{\delta}^{Y} P_{y_0}(x,y)\, \mathrm{d}y.
\end{equation}
Thus, 
\begin{equation}
\label{eq:conv_bound}
\int_{\delta}^Y \abs{h(M-\ii y)}\,\mathrm{d}y\leq (K_0 \ast \abs{h_0})(M)+(K_{y_0} \ast \abs{h_{y_0}})(M)
\end{equation}

Observe that $K_0,K_{y_0}\in L^1(\R)$ as Tonelli's theorem gives
\begin{equation}
    \int_\mathbb{R} K_{0}(x) \,\mathrm{d
    }x = \int_{\delta}^{Y} \bigg(\int_\mathbb{R} P_{0}(x,y)\,\mathrm{d}x \bigg) \,\mathrm{d}y 
    =\int_{\delta}^{Y}\bigg(1-\frac{y}{y_0} \bigg)\,\mathrm{d}y
    =Y-\delta-\frac{Y^2-\delta^2}{2y_0}<\infty
\end{equation}
and
\begin{equation}
\int_{\R} K_{y_0}(x)\,\mathrm{d}x
= \int_{\delta}^{Y}\bigg( \int_{\R} P_{y_0}(x,y)\,\mathrm{d}x \bigg)\,\mathrm{d}y
= \int_{\delta}^{Y}\frac{y}{y_0}\,\mathrm{d}y
= \frac{Y^2-\delta^2}{2y_0}<\infty.
\end{equation}

We next show that $K_0,K_{y_0}\in C_0(\R)$ (recall~\eqref{eq:C0R}).
Since the Poisson kernels $P_0$ and $P_{y_0}$ are continuous on $\R\times[\delta,Y]$, it follows from \eqref{eq:K_definitions} that $K_0$ and $K_{y_0}$ are continuous on $\R$. \\

Moreover, for $x\in\R\setminus\{0\}$ and $y\in[\delta,Y]$, the explicit formulas
\eqref{eqn:poisson_kernels} give
\begin{equation}
0\leq P_0(x,y)
\leq
\frac{1}
{2y_0\big(\!\cosh(\frac{\pi x}{y_0})-1\big)},
\qquad
0\leq P_{y_0}(x,y)
\leq
\frac{1}
{2y_0\big(\!\cosh(\frac{\pi x}{y_0})-1\big)}.
\end{equation}
Consequently,
\begin{equation}
0\leq K_0(x)
\leq
\frac{Y-\delta}
{2y_0\big(\!\cosh(\frac{\pi x}{y_0})-1\big)},
\qquad
0\leq K_{y_0}(x)
\leq
\frac{Y-\delta}
{2y_0\big(\!\cosh(\frac{\pi x}{y_0})-1\big)}.
\end{equation}
The right-hand side tends to zero as $\abs{x}\to\infty$. Hence $K_0,K_{y_0}\in C_0(\R)\subseteq L^\infty(\R)$. \\ 

In the half-plane case, the Poisson representation and Tonelli's
theorem give
\begin{equation}
\begin{aligned}
\int_{\delta}^Y\abs{h(M-\ii y)}\,\mathrm{d}y
&=
\int_{\delta}^Y
\left|
\int_{\R}
\hat{P}_0(M-x',y)h_0(x')\,\mathrm{d}x'
\right|
\mathrm{d}y
\\
&\leq
\int_{\delta}^Y\int_{\R}
\hat{P}_0(M-x',y)\abs{h_0(x')}
\,\mathrm{d}x'\,\mathrm{d}y
\\
&=
\int_{\R}
\hat{K}_0(M-x')\abs{h_0(x')}\,\mathrm{d}x'=(\hat{K}_0\ast\abs{h_0})(M),
\label{eq:halfplane_convolution_bound}
\end{aligned}
\end{equation}
where
\begin{equation}
\hat{K}_0(x)
\coloneqq
\int_{\delta}^Y\hat{P}_0(x,y)\,\mathrm{d}y.
\end{equation}
Since $\hat{P}_0$ is nonnegative, Tonelli's theorem gives
\begin{equation}
\int_{\R}\hat{K}_0(x)\,\mathrm{d}x
=
\int_{\delta}^Y
\left(
\int_{\R}\hat{P}_0(x,y)\,\mathrm{d}x
\right)
\mathrm{d}y
=
\int_{\delta}^Y\mathrm{d}y
=
Y-\delta.
\end{equation}
Thus, $\hat{K}_0\in L^1(\R)$. The function $\hat{K}_0$ also belongs to $C_0(\R)$. Indeed,
$\hat{P}_0$ is continuous on $\R\times[\delta,Y]$, so
$\hat{K}_0$ is continuous on $\R$. Moreover,
\eqref{eqn:halfplane_poisson_kernel} gives
\begin{equation}
0\leq
\hat{K}_0(x)
\leq
\frac{Y(Y-\delta)}
{\pi(x^2+\delta^2)}.
\end{equation}
The right-hand side tends to zero as $\abs{x}\to\infty$. Hence $\hat{K}_0\in C_0(\R)\subseteq L^\infty(\R)$. \\

By Lemma~\ref{lem:Hardyboundary}, all the boundary functions appearing above belong to $L^1(\R)$. Each convolution in \eqref{eq:conv_bound} and \eqref{eq:halfplane_convolution_bound} has the form $\int_{\R}K(M-x')\abs{g(x')}\,\mathrm{d}x'$, where $K\in C_0(\R)\cap L^\infty(\R)$ and $g\in L^1(\R)$. For almost every $x'\in\R$, the integrand tends to zero as $\abs{M}\to\infty$ and is bounded by $\norm{K}_\infty\abs{g(x')}$. The dominated convergence theorem therefore shows that every convolution in \eqref{eq:conv_bound} and \eqref{eq:halfplane_convolution_bound} vanishes as $\abs{M}\to\infty$. Combining these bounds with \eqref{eqn:subPoisson}, and applying the same estimates with $M$ replaced by $-M$, gives, in both cases,
\begin{equation}
\lim_{M\to\infty}
\bigg\lVert
\int_{\Gamma_{\textup{right}}(M,\delta,Y)}
f(z)\mc{U}_\eta(z;t)\,\mathrm{d}z
\bigg\rVert_2
=
\lim_{M\to\infty}
\bigg\lVert
\int_{\Gamma_{\textup{left}}(M,\delta,Y)}
f(z)\mc{U}_\eta(z;t)\,\mathrm{d}z
\bigg\rVert_2
=
0.
\end{equation}
By Lemma~\ref{lem:squareint} and \eqref{eq:regularized_contraction},
the horizontal integrals converge absolutely as $M\to\infty$. Their
orientations give
\begin{align}
\lim_{M\to\infty}
\int_{\Gamma_{\textup{top}}(M,\delta,Y)}
f(z)\mc{U}_\eta(z;t)\,\mathrm{d}z
&=
\int_{\R}
f(x-\ii\delta)\mc{U}_\eta(x-\ii\delta;t)\,\mathrm{d}x,
\\
\lim_{M\to\infty}
\int_{\Gamma_{\textup{bottom}}(M,\delta,Y)}
f(z)\mc{U}_\eta(z;t)\,\mathrm{d}z
&=
-\int_{\R}
f(x-\ii Y)\mc{U}_\eta(x-\ii Y;t)\,\mathrm{d}x.
\label{eqn:bottomMlim}
\end{align}
Taking $M\to\infty$ in \eqref{eq:Cauchy} therefore gives
\begin{equation}
\label{eq:finite_rectangle_identity}
\mathrm{e}^{-\eta t}\m{U}(t)
=
\frac{1}{\sqrt{2\pi}}
\left(
\int_{\R}
f(x-\ii\delta)\mc{U}_\eta(x-\ii\delta;t)\,\mathrm{d}x
-
\int_{\R}
f(x-\ii Y)\mc{U}_\eta(x-\ii Y;t)\,\mathrm{d}x
\right).
\end{equation}
This completes Step~2. \\

\paragraph{Step 3: passage to the lower boundary.}
Steps~1 and~2 apply to a rectangular contour with fixed upper and lower
depths $\delta$ and $Y$, and yield \eqref{eq:finite_rectangle_identity}.
The order of the limits is important. In the finite-strip case, we first
let $M\to\infty$ for fixed $\delta$, with $\eta=0$ and
$Y(\delta)=y_0-\delta$, and then let $\delta\downarrow0$. In the
half-plane case, for fixed $t>0$, we first let $M\to\infty$, then
$Y\to\infty$, then $\eta\downarrow0$, and finally
$\delta\downarrow0$.

\paragraph{Step 3a: finite strip case, $y_0\in(1,\infty)$.}
Set $\eta=0$ and $Y(\delta)=y_0-\delta$. We now take the limit of the
second horizontal integral in \eqref{eq:finite_rectangle_identity} as
$\delta\downarrow0$,
\begin{align}
\lim_{\delta\downarrow0}
\int_{\R}
f(x-\ii Y(\delta))\mathcal{U}(x-\ii Y(\delta);t)\,\mathrm{d}x
&=
\int_{\mathbb{R}}
f_{y_0}(x)\mathcal{U}(x-\ii y_0;t)\,\mathrm{d}x.
\label{eqn:limdeltafinite1}
\end{align}
To justify this limit, Lemmas~\ref{lem:poisson_representation}
and~\ref{lem:PoissonL1} first give
\begin{equation}
\lim_{\delta\downarrow0}
\norm{h(\,\cdot-\ii Y(\delta))-h_{y_0}}_1
=0.
\end{equation}
By \eqref{eq:b0}, the functions
$1/B(\,\cdot-\ii Y(\delta);y_0)$ are uniformly bounded for all
sufficiently small $\delta$ and converge pointwise to
$1/B(\,\cdot-\ii y_0;y_0)$. Combining these facts with
\eqref{eq:f_to_h_thm4} and applying dominated convergence to the term
containing $h_{y_0}$ gives
\begin{equation}
\lim_{\delta\downarrow0}
\norm{f(\,\cdot-\ii Y(\delta))-f_{y_0}}_1
=0.
\end{equation}
Moreover, $\mathcal{U}(\ccdot;t)$ is entire and contractive on $\H_-$. Thus,
\begin{align}
&\bigg\lVert
\int_{\R}
f(x-\ii Y(\delta))\mathcal{U}(x-\ii Y(\delta);t)\,\mathrm{d}x
-
\int_{\R}
f_{y_0}(x)\mathcal{U}(x-\ii y_0;t)\,\mathrm{d}x
\bigg\rVert_2
\nonumber\\
&\leq
\norm{f(\,\cdot-\ii Y(\delta))-f_{y_0}}_1+
\int_{\R}
\abs{f_{y_0}(x)}
\norm{
\mathcal{U}(x-\ii Y(\delta);t)
-
\mathcal{U}(x-\ii y_0;t)
}_2
\,\mathrm{d}x.
\end{align}
The first term tends to zero by the $L^1$ convergence above. In the second
term, entireness gives pointwise convergence of the difference to zero,
while contractivity bounds the integrand by $2\abs{f_{y_0}}\in L^1(\R)$.
The dominated convergence theorem therefore proves
\eqref{eqn:limdeltafinite1}.
Thus, unlike in the half-plane case, the lower horizontal contribution
survives; in \eqref{eq:finite_rectangle_identity}, it converges to the
negative of the right-hand side of \eqref{eqn:limdeltafinite1}.

\paragraph{Step 3b: half-plane case, $y_0=\infty$.}
Fix $\delta\in(0,1)$, $\eta>0$, and $t>0$. In the half-plane case,
$Y>1$ is arbitrary. We now let $Y\to\infty$ in
\eqref{eq:finite_rectangle_identity}.
For $Y\geq2$, \eqref{eq:regularized_contraction} and
\eqref{eq:halfplane_uniform_slice_bound} give
\begin{align}
\bigg\lVert
\int_{\R}
f(x-\ii Y)\mc{U}_\eta(x-\ii Y;t)\,\mathrm{d}x
\bigg\rVert_2\leq &\; 
\mathrm{e}^{-\eta Yt}
\norm{f(\,\cdot-\ii Y)}_1
\nonumber\\
\leq &\; 
\frac{3\mathrm{e}^{-\eta Yt}}{\sqrt{2\pi}}
\sup_{s>0}\norm{h(\,\cdot-\ii s)}_1.
\label{eq:halfplane_bottom_bound}
\end{align}
Since $\eta>0$ and $t>0$, the right-hand side tends to zero as
$Y\to\infty$. Therefore,
\begin{equation}
\label{eq:halfplane_bottom_limit}
\lim_{Y\to\infty}
\bigg\lVert
\int_{\R}
f(x-\ii Y)\mc{U}_\eta(x-\ii Y;t)\,\mathrm{d}x
\bigg\rVert_2
=0.
\end{equation}
Equivalently, by \eqref{eqn:bottomMlim},
\begin{equation}
\lim_{Y\to\infty}\lim_{M\to\infty}
\bigg\lVert
\int_{\Gamma_\textup{bottom}(M,\delta,Y)}
f(z)\mc{U}_\eta(z;t)\,\mathrm{d}z
\bigg\rVert_2
=0.
\end{equation}
Taking $Y\to\infty$ in \eqref{eq:finite_rectangle_identity} now gives
\begin{equation}
\label{eq:halfplane_regularized_top_identity}
\mathrm{e}^{-\eta t}\m{U}(t)
=
\frac{1}{\sqrt{2\pi}}
\int_{\R}
f(x-\ii\delta)\mc{U}_\eta(x-\ii\delta;t)\,\mathrm{d}x.
\end{equation}
The case $t=0$ will be recovered after passage to the upper boundary by
continuity.

\paragraph{Step 4: passage to the upper boundary.}
We first treat the upper boundary in both cases. Lemmas~\ref{lem:poisson_representation}
and~\ref{lem:PoissonL1} give
\begin{equation}
\lim_{\delta\downarrow0}
\norm{h(\,\cdot-\ii\delta)-h_0}_1
=0.
\end{equation}
By \eqref{eq:b0}--\eqref{eq:b0_infty}, the functions
$1/B(\,\cdot-\ii\delta;y_0)$ are uniformly bounded for all sufficiently
small $\delta$ and converge pointwise to $1/B(\,\cdot;y_0)$. Combining
these facts with \eqref{eq:f_to_h_thm4} and applying dominated convergence
to the term containing $h_0$ gives
\begin{equation}
\label{eq:upper_boundary_L1_convergence}
\lim_{\delta\downarrow0}
\norm{f(\,\cdot-\ii\delta)-f_0}_1
=0.
\end{equation}
Moreover, $\mc{U}(\ccdot;t)$ is entire and contractive on $\H_-$. Thus,
\begin{align}
&\bigg\lVert
\int_{\R}
f(x-\ii\delta)\mc{U}(x-\ii\delta;t)\,\mathrm{d}x
-
\int_{\R}
f_0(x)\mc{U}(x;t)\,\mathrm{d}x
\bigg\rVert_2
\nonumber\\
&\quad\leq
\norm{f(\,\cdot-\ii\delta)-f_0}_1
\nonumber\\
&\qquad+
\int_{\R}
\abs{f_0(x)}
\norm{
\mc{U}(x-\ii\delta;t)-\mc{U}(x;t)
}_2
\,\mathrm{d}x.
\end{align}
The first term tends to zero by \eqref{eq:upper_boundary_L1_convergence}.
In the second term, entireness gives pointwise convergence of the
difference to zero, while contractivity in $\H_-$ and unitarity on $\R$
bound the integrand by $2\abs{f_0}\in L^1(\R)$. The dominated convergence
theorem therefore gives
\begin{equation}
\label{eq:upper_boundary_integral_convergence}
\lim_{\delta\downarrow0}
\bigg\lVert
\int_{\R}
f(x-\ii\delta)\mc{U}(x-\ii\delta;t)\,\mathrm{d}x
-
\int_{\R}
f_0(x)\mc{U}(x;t)\,\mathrm{d}x
\bigg\rVert_2
=0.
\end{equation}

Suppose first that $y_0\in(1,\infty)$. In this case, $\eta=0$ and
$Y(\delta)=y_0-\delta$. Taking $\delta\downarrow0$ in
\eqref{eq:finite_rectangle_identity} and using
\eqref{eqn:limdeltafinite1} and
\eqref{eq:upper_boundary_integral_convergence} gives
\begin{equation}
\m{U}(t)
=
\frac{1}{\sqrt{2\pi}}
\int_{\R}f_0(x)\mc{U}(x;t)\,\mathrm{d}x
-
\frac{1}{\sqrt{2\pi}}
\int_{\R}f_{y_0}(x)\mc{U}(x-\ii y_0;t)\,\mathrm{d}x.
\label{eq:finite_strip_boundary_representation}
\end{equation}

Suppose next that $y_0=\infty$. Fix $\delta\in(0,1)$ and $t>0$, and let
$\eta\downarrow0$ in \eqref{eq:halfplane_regularized_top_identity}. The
left-hand side converges to $\m{U}(t)$. On the right-hand side,
\eqref{eq:regularized_propagator_factorization} gives pointwise convergence
of $\mc{U}_\eta(x-\ii\delta;t)$ to $\mc{U}(x-\ii\delta;t)$, while
\eqref{eq:regularized_contraction} bounds the integrand by
$\abs{f(x-\ii\delta)}\in L^1(\R)$. Hence, dominated convergence gives
\begin{equation}
\m{U}(t)
=
\frac{1}{\sqrt{2\pi}}
\int_{\R}
f(x-\ii\delta)\mc{U}(x-\ii\delta;t)\,\mathrm{d}x.
\end{equation}
Taking $\delta\downarrow0$ and using
\eqref{eq:upper_boundary_integral_convergence} now gives
\begin{equation}
\m{U}(t)
=
\frac{1}{\sqrt{2\pi}}
\int_{\R}f_0(x)\mc{U}(x;t)\,\mathrm{d}x
\qquad (t>0).
\label{eq:halfplane_boundary_representation}
\end{equation}
Finally, $\m{U}(t)$ is continuous at $t=0$, and, for each $x\in\R$,
$\mc{U}(x;t)$ converges to $\m{I}$ as $t\downarrow0$. Since
$\mc{U}(x;t)$ is unitary for $x\in\R$, dominated convergence extends
\eqref{eq:halfplane_boundary_representation} to $t=0$. \\

Combining the two cases, we have proved
\begin{equation}
\label{eqn:pre_truncation}
\m{U}(t)
=
\begin{dcases}
\displaystyle
\frac{1}{\sqrt{2\pi}}
\int_{\R}f_0(x)\mc{U}(x;t)\,\mathrm{d}x
-
\frac{1}{\sqrt{2\pi}}
\int_{\R}f_{y_0}(x)\mc{U}(x-\ii y_0;t)\,\mathrm{d}x
& y_0\in(1,\infty)
\\
\displaystyle
\frac{1}{\sqrt{2\pi}}
\int_{\R}f_0(x)\mc{U}(x;t)\,\mathrm{d}x,
& y_0=\infty.
\end{dcases}
\end{equation}

\paragraph{Step 5: truncate the integral.}
Suppose first that $y_0\in(1,\infty)$. Subtracting
\eqref{eqn:pre_truncation} from \eqref{eq:Of} gives
\begin{equation}
\begin{aligned}
\m{O}_f(t;R)-\m{U}(t)
=&\;
-\frac{1}{\sqrt{2\pi}}
\int_{\R\setminus[-R,R]}
f_0(x)\mc{U}(x;t)\,\mathrm{d}x
\nonumber\\
&\;+
\frac{1}{\sqrt{2\pi}}
\int_{\R}
f_{y_0}(x)\mc{U}(x-\ii y_0;t)\,\mathrm{d}x.
\end{aligned}
\end{equation}
For $x\in\R$, the propagator $\mc{U}(x;t)$ is unitary, while
$\norm{\mc{U}(x-\ii y_0;t)}_2\leq1$. The triangle inequality therefore
gives
\begin{align}
\begin{multlined}
   \norm{\m{O}_f(t;R)-\m{U}(t)}_2
\leq
\frac{1}{\sqrt{2\pi}}
\int_{\R\setminus[-R,R]}\abs{f_0(x)}\,\mathrm{d}x +
\frac{1}{\sqrt{2\pi}}
\int_{\R}\abs{f_{y_0}(x)}\,\mathrm{d}x
=\mc{E}_{R,y_0}[f]. 
\end{multlined}
\end{align}

When $y_0=\infty$, the same subtraction gives
\begin{equation}
\m{O}_f(t;R)-\m{U}(t)
=
-\frac{1}{\sqrt{2\pi}}
\int_{\R\setminus[-R,R]}
f_0(x)\mc{U}(x;t)\,\mathrm{d}x.
\end{equation}
Using unitarity on the real axis, we obtain
\begin{equation}
\norm{\m{O}_f(t;R)-\m{U}(t)}_2
\leq
\frac{1}{\sqrt{2\pi}}
\int_{\R\setminus[-R,R]}\abs{f_0(x)}\,\mathrm{d}x
=\mc{E}_{R,\infty}[f].
\end{equation}
This proves Theorem~\ref{thm:LCHS}. \\

\subsection{Strict relaxation of the class of admissible kernel functions}
\label{subsec:assumptionsproof}

We now show that the Hardy space hypotheses strictly relax the
earlier kernel function assumptions when the associated error bound is finite.

\begin{prop}[Strict relaxation of kernel function hypotheses]
\label{prop:assumptions}
Let $y_0\in(1,\infty)$ and suppose that $f\in\mc{M}(\overline{S}_{y_0})$ satisfies
\begin{enumerate}[label=(\roman*)]
	\item $z\mapsto (z+\ii)f(z)$ is in $\mc{O}(\overline{S}_{y_0})$; \label{item:pole}
	\item $\operatorname{Res}_{z=-\ii} f(z)=\frac{\ii}{\sqrt{2\pi}}$; \label{item:residue}
	\item $\lim_{\lvert x\rvert\to\infty} \sup_{y\in (0,y_0)} \lvert f(x-\ii y)\rvert=0$; \label{item:uniform_decay}
	\item $f\rvert_{\partial S_{y_0}}\in L^1(\partial S_{y_0})$.\label{item:finite_error}
\end{enumerate}
Then, $f\rvert_{S_{y_0}}\in \mc{A}(S_{y_0})$.
The converse fails as there exists $f\in\mc{A}(S_{y_0})$ such that the
function $z\mapsto(z+\ii)f(z)$ does not admit a holomorphic extension to
a neighborhood of $\partial S_{y_0}$, and hence does not belong to
$\mc{O}(\overline{S}_{y_0})$.
\end{prop}

Note that conditions \ref{item:pole}--\ref{item:uniform_decay} coincide with the kernel function assumptions in \cite[Thm.~1]{low2025}, while \ref{item:finite_error} is the additional condition needed to make
the error bound in \cite[Eq.~(8)]{low2025} finite.\\

Adapted to our
notation, the error bound in \cite[Thm.~1]{low2025} is
\begin{equation}\label{eqn:lowadapted}
\mc{E}_{R,y_0}[f]    = \frac{1}{\sqrt{2\pi}}
\int_{\R\setminus[-R,R]} \lvert f(x)\rvert\,\mathrm{d}x
+
\frac{1}{\sqrt{2\pi}}
\int_{\R} \lvert f(x-\ii y_0)\rvert\,\mathrm{d}x;
\end{equation}
since the kernel functions in \cite{low2025} are continuous, the boundary values are well-defined without having to take limits.
Finiteness of the second term requires
$f(\,\cdot-\ii y_0)\in L^1(\R)$.
Under the regularity hypotheses of that theorem, $f|_{\R}$ is locally integrable, so finiteness of the first term for some $R>0$ implies $f|_{\R}\in L^1(\R)$; its tail then tends to zero as $R\to\infty$. Thus the condition
$f|_{\partial S_{y_0}}\in L^1(\partial S_{y_0})$ is implicit in any
application of \cite[Thm.~1]{low2025} for which the displayed error
bound is finite.

\begin{proof}

Conditions \ref{item:pole}--\ref{item:residue} imply that $f$ admits a representation \eqref{eqn:kernel_f_h} with $h\in \mc{O}(\overline{S}_{y_0})$ satisfying \eqref{eq:norm_h}.
Define
\begin{equation}
\label{eq:f_to_h}
h(z)
\coloneqq
\frac{\sqrt{2\pi}}{\ii}f(z)B(z;y_0)=
\frac{\sqrt{2\pi}}{\ii}
\bigl((z+\ii)f(z)\bigr)
\frac{B(z;y_0)}{z+\ii} \quad 
(z\in \overline{S}_{y_0}\setminus\{-\ii\}),
\end{equation}
where $B(\ccdot;y_0)\in \mc{O}(\overline{S}_{y_0})$ is the Blaschke factor defined in \eqref{eq:Blaschke}.
By condition~\ref{item:pole}, the function $(z+\ii)f(z)\in \mc{O}(\overline{S}_{y_0})$.
Moreover, the quotient $B(z;y_0)/(z+\ii)$ extends holomorphically through $z=-\ii$ since $B(z;y_0)$ has a simple zero there.
Consequently, \eqref{eq:f_to_h} extends holomorphically to $\overline{S}_{y_0}$ and rearranging the first equality in \eqref{eq:f_to_h} gives the representation \eqref{eqn:kernel_f_h}.
At $z=-\ii$, condition~\ref{item:residue} gives
\begin{equation}
h(-\ii)=
\frac{\sqrt{2\pi}}{\ii}
\operatorname*{Res}_{z=-\ii}f(z)
\lim_{z\to-\ii}\frac{B(z;y_0)}{z+\ii}
=\lim_{z\to-\ii}\frac{B(z;y_0)}{z+\ii}.
\end{equation}

It remains to show that $h$ satisfies
\begin{equation}
\label{eq:h_sup_norm}
\sup_{y\in (0,y_0)} \int_{\R} \lvert h(x-\ii y)\rvert\,\mathrm{d}x<\infty 
\end{equation}
and hence $h\in H^1(S_{y_0})$. Because $h\in\mc{O}(\overline{S}_{y_0})$, its boundary values are its classical restrictions
\begin{equation}
h_0(x)=h(x)
\quad
h_{y_0}(x)=h(x-\ii y_0) \quad (x\in \R).
\end{equation}
The same holds for $f$, as its pole at $-\ii$ lies strictly inside the strip.
Since $B(\ccdot;y_0)$ is the Blaschke factor for $S_{y_0}$,
\begin{equation}\label{eqn:blaschke_strip_modulus}
    |B(z;y_0)|\le 1
\qquad
(z\in\overline S_{y_0}),
\end{equation}
with equality if and only if $z\in\partial S_{y_0}$.
Hence, condition~\ref{item:finite_error} and the
definition of $h$ in \eqref{eq:f_to_h} give
\begin{equation}
\norm{h_0}_1=
\sqrt{2\pi}\norm{f_0}_1
<\infty,\quad 
\norm{h_{y_0}}_1=
\sqrt{2\pi}\norm{f_{y_0}}_1
<\infty,
\end{equation}
i.e., the boundary values of $h$ belong to $L^1(\R)$.
Since $|B|\le1$ on $\overline S_{y_0}$, condition~\ref{item:uniform_decay} implies
\begin{equation}
\label{eqn:h_uniform_decay}
\lim_{|x|\to\infty}
\sup_{y\in[0,y_0]}
|h(x-\ii y)|
=0.
\end{equation}
We now use the boundary values of $h$ to construct a harmonic majorant for $|h|$.
The decay \eqref{eqn:h_uniform_decay} ensures that the maximum-principle argument extends from bounded rectangles to the whole strip.\\

For $z=x-\ii y\in S_{y_0}$, let $u(z)$ be the Poisson extension of the boundary values $|h_0|$ and $|h_{y_0}|$,
\begin{equation}
u(x-\ii y)
\coloneqq
\int_{\R}
P_0(x-x',y)\lvert h_0(x')\rvert\,\mathrm{d}x'+
\int_{\R}
P_{y_0}(x-x',y)
\lvert h_{y_0}(x')\rvert\,\mathrm{d}x'.
\end{equation}
Applying Lemma~\ref{lem:PoissonL1}
implies that $u$ is harmonic on $S_{y_0}$ and extends continuously to $\overline{S}_{y_0}$ with
\begin{equation}
\label{eqn:u_boundary_values}
u(x)=\lvert h_0(x)\rvert,
\quad
u(x-\ii y_0)=\lvert h_{y_0}(x)\rvert
\quad
(x\in\R).
\end{equation}
Since $|h|$ need not be harmonic, we instead compare $u$ with the harmonic functions $\operatorname{Re}(\ee^{\ii\theta}h)$.
For each $\theta\in\R$, define
\begin{equation}
w_\theta(z)
\coloneqq
\operatorname{Re}\!\big(\ee^{\ii\theta}h(z)\big)-u(z)
\qquad
(z\in\overline{S}_{y_0}).
\end{equation}
Then $w_\theta$ is harmonic on $S_{y_0}$ and extends continuously to $\overline{S}_{y_0}$.
We now prove that $w_\theta\le0$ on $S_{y_0}$.\\

We apply the maximum principle on the bounded rectangle
\begin{equation}
\mc{R}_M
\coloneqq
\big\{
x-\ii y:
x\in(-M,M),\;
y\in(0,y_0)
\big\}\qquad (M>0).
\end{equation}
Since $u=|h|$ on the horizontal sides of $\partial\mc{R}_M$,
\begin{equation}
w_\theta(z)
=
\operatorname{Re}\big(\ee^{\ii\theta}h(z)\big)
-\lvert h(z)\rvert
\leq0.
\end{equation}
Fix $z_\ast\in S_{y_0}$, $\theta\in\R$, and $\eta>0$.
To estimate the vertical sides of $\partial\mc{R}_M$, the decay
\eqref{eqn:h_uniform_decay} allows us to choose $M>\lvert\operatorname{Re}z_\ast\rvert$ sufficiently large that
\begin{equation}
\lvert h(\pm M-\ii y)\rvert\leq\eta
\qquad
(y\in[0,y_0]).
\end{equation}
Since the Poisson kernels and the boundary data defining $u$ are
nonnegative, $u\geq0$.
Hence, on the vertical sides of
$\partial\mc{R}_M$,
\begin{equation}
w_\theta
\leq
\lvert h\rvert-u
\leq
\lvert h\rvert
\leq\eta.
\end{equation}
Therefore $w_\theta\le\eta$ on $\partial\mc{R}_M$ and the maximum principle for harmonic functions gives $w_\theta(z_\ast)\leq\eta$. Moreover, since $z_\ast$ and $\theta$ were arbitrary, letting $\eta\downarrow0$ yields
\begin{equation}
\operatorname{Re}\bigl(\ee^{\ii\theta}h(z)\bigr)
\leq u(z)
\qquad
(z\in S_{y_0},\;\theta\in\R).
\end{equation}
Choosing $\theta$ so that $\operatorname{Re}(\ee^{\ii\theta}h(z))=|h(z)|$, which is possible whenever $h(z)\neq0$, therefore gives
\begin{equation}
\label{eqn:h_poisson_majorant}
\lvert h(z)\rvert\leq u(z)
\qquad
(z\in S_{y_0}),
\end{equation}
while the inequality is trivial if $h(z)=0$.\\

For each $y\in(0,y_0)$, Tonelli's theorem,
\eqref{eqn:h_poisson_majorant}, and the Poisson-kernel mass identities
give
\begin{equation}
\int_{\R}\lvert h(x-\ii y)\rvert\,\mathrm{d}x
\leq 
\int_{\R}u(x-\ii y)\,\mathrm{d}x
=
\left(1-\frac{y}{y_0}\right)\norm{h_0}_1
+
\frac{y}{y_0}\norm{h_{y_0}}_1<\infty
\end{equation}
Hence $h\in H^1(S_{y_0})$, together with the representation
\eqref{eqn:kernel_f_h} and the normalization \eqref{eq:norm_h}, 
proves that $f\rvert_{S_{y_0}}\in\mc{A}(S_{y_0})$. \\

To see that the converse does not hold, fix $R>0$, set
$\varepsilon=(1-m_{R,y_0})/2$, and let $f_\star$ be the function in
Proposition~\ref{prop:subnormalization-optimal-kernel}. By
\eqref{eqn:kernel_f_h}, the unimodularity of $B(\ccdot;y_0)$ on the
boundary, and \eqref{eq:optimal_upper_boundary_modulus}, the modulus of the upper-boundary value of $f_\star$ has unequal one-sided
limits at $x=R$. Hence its $L^1$-equivalence class contains no continuous
representative. Since $x+\ii\neq0$ on the upper boundary,
condition~\ref{item:pole} would give such a continuous representative and
is therefore violated.

\end{proof}

\subsection{Characterization of exact LCHS kernel functions}
\label{subsec:exact_LCHS_characterization}

Taking $R\to\infty$ in Theorem~\ref{thm:LCHS}, the following lemma shows that the scalar identity
\eqref{eq:exact_scalar_LCHS} is sufficient for the exact matrix-valued
LCHS representation.

\begin{lem}[Exact LCHS kernel functions]
\label{lem:exact_LCHS_characterization}
Let $\varphi\in L^1(\R)$. There exists a unique
$f\in\mc{A}(\H_-)$ with boundary value $f_0=\varphi$ almost everywhere
if and only if \eqref{eq:exact_scalar_LCHS} holds.
\end{lem}

\begin{proof}
Suppose first that $\varphi=f_0$ for some $f\in\mc{A}(\H_-)$.
Apply Theorem~\ref{thm:LCHS} with $\m{L}(t)=\m{I}$ and $\m{H}(t)=0$,
and let $R\to\infty$. Since $f_0\in L^1(\R)$, both the truncation
error and the tail of the integral tend to zero, giving
\eqref{eq:exact_scalar_LCHS} for $k>0$. The case $k=0$ follows by
dominated convergence. \\

Conversely, suppose that \eqref{eq:exact_scalar_LCHS} holds. Define
\begin{equation}
q(z)\coloneqq\sqrt{\frac{2}{\pi}}\frac{1}{z^2+1},
\qquad g_0\coloneqq\varphi-q|_{\R},
\qquad
\hat g_0(k)\coloneqq\frac{1}{\sqrt{2\pi}}
\int_{\R}g_0(x)\ee^{-\ii kx}\,\mathrm{d}x.
\end{equation}
The Fourier transform of $q|_{\R}$ is $\ee^{-|k|}$, so
$\hat g_0(k)=0$ for $k\geq0$. Using the half-plane Poisson kernel
in \eqref{eqn:halfplane_poisson_kernel}, set
\begin{equation}
\label{eq:exact_LCHS_poisson_extension}
g(x-\ii y)\coloneqq
\big(\hat P_0(\ccdot,y)*g_0\big)(x)
=\frac{1}{\sqrt{2\pi}}\int_{-\infty}^{0}
\hat g_0(k)\ee^{\ii(x-\ii y)k}\,\mathrm{d}k
\qquad(y>0).
\end{equation}
The second equality follows by Fourier inversion, since
$\hat g_0$ is bounded and multiplication by $\ee^{-y\abs{k}}$ makes
it integrable. Differentiation under the integral is valid locally
uniformly on $\H_-$, so $g$ is holomorphic there. Moreover, positivity
and unit mass of the Poisson kernel give
\begin{equation}
\sup_{y>0}\norm{g(\ccdot-\ii y)}_1\leq\norm{g_0}_1.
\end{equation}
Thus $g\in H^1(\H_-)$, and Poisson boundary convergence identifies
its boundary value with $g_0$ almost everywhere. \\

Set $f=q+g$. With $B(z;\infty)=(z+\ii)/(z-\ii)$, we have
\begin{equation}
B(z;\infty)f(z)
=\sqrt{\frac{2}{\pi}}\frac{1}{(z-\ii)^2}
+B(z;\infty)g(z)\in H^1(\H_-).
\end{equation}
Indeed, $|B(z;\infty)|\leq1$ on $\H_-$ and the first term has
$L^1$ norm $\sqrt{2\pi}/(1+y)$ on the line $\operatorname{Im}z=-y$.
Consequently,
\begin{equation}
h\coloneqq\frac{\sqrt{2\pi}}{\ii}B(\ccdot;\infty)f
\in H^1(\H_-),
\qquad
h(-\ii)=\frac{\ii}{2}
=\lim_{z\to-\ii}\frac{B(z;\infty)}{z+\ii}.
\end{equation}
This proves $f\in\mc{A}(\H_-)$, and its boundary value is
$q|_{\R}+g_0=\varphi$. Finally, if two admissible functions have the
same boundary value, their corresponding Hardy functions $h$ also
have the same boundary value. The Poisson representation in
Lemma~\ref{lem:poisson_representation} then proves uniqueness.
\end{proof}

\end{document}

%% file: cost-table.tex
\begin{tabular*}{\linewidth}{@{\extracolsep{\fill}} c *{3}{S[table-format=2.4]} @{}}
\toprule
&\multicolumn{1}{c}{This work} &\multicolumn{2}{c}{Low--Somma \cite{low2025}}\\
\cmidrule(lr){2-2}\cmidrule(l){3-4}
$\varepsilon$ &\multicolumn{1}{c}{$f_\star$ \eqref{eq:f_optimal}} &\multicolumn{1}{c}{\eqref{eq:generalized_kernel_intro} with $j=2$, $y=1$} &\multicolumn{1}{c}{\eqref{eq:generalized_kernel_intro} with variable $j,y$}\\
\midrule
$10^{-1}$ & 2.0673 & 3.3213 & 2.5589 \\
$10^{-2}$ & 5.1793 & 9.3434 & 7.0591 \\
$10^{-3}$ & 8.5720 & 16.8148 & 12.7546 \\
$10^{-4}$ & 12.1076 & 25.2656 & 19.2720 \\
$10^{-5}$ & 15.7309 & 34.3612 & 26.4191 \\
$10^{-6}$ & 19.4136 & 43.9410 & 34.0847 \\
$10^{-7}$ & 23.1391 & 53.8354 & 42.1281 \\
$10^{-8}$ & 26.8970 & 64.0510 & 50.5333 \\
$10^{-9}$ & 30.6801 & 74.4881 & 59.2398 \\
$10^{-10}$ & 34.4835 & 85.0466 & 68.2266 \\
\bottomrule
\end{tabular*}